\pdfoutput=1
\documentclass[a4paper,11pt]{article} 
\usepackage{float}
\usepackage[utf8x]{inputenc} 
\usepackage[english]{babel}
\usepackage{amsmath,amsthm,amsfonts,amssymb,mathrsfs, bbm}
\usepackage{latexsym}
\usepackage{authblk}
\usepackage{lineno}
\usepackage{eucal}
\usepackage[dvipsnames]{xcolor}
\usepackage{hyperref}
\usepackage{afterpage}
\usepackage{lmodern}
\usepackage[a4paper,left=2.9cm,right=2.9cm]{geometry}
\usepackage{graphicx}
\usepackage{bm}
\usepackage{mathtools}
\DeclarePairedDelimiter{\ceil}{\lceil}{\rceil}

\definecolor{BeauBlue}{rgb}{0, 0.2, .9}
\definecolor{BeauOrange}{rgb}{.8, .1, 0}
\usepackage{hyperref}
\hypersetup{
	colorlinks = true,
	linkcolor = BeauBlue,
	urlcolor = cyan,
	citecolor = BeauOrange,
}

\numberwithin{equation}{section}

\newtheorem{theorem}{Theorem}[section] 
\newtheorem{proposition}[theorem]{Proposition} 

\newtheorem{lemma}[theorem]{Lemma}
\newtheorem{definition}[theorem]{Definition}  
\newtheorem{example}[theorem]{Example}
\newtheorem{remark}[theorem]{Remark}

\DeclareMathOperator{\Tr}{Tr} 

\PassOptionsToPackage{dvipsnames}{xcolor}
\usepackage{graphicx} 

\usepackage{tikz}
\usepackage{tikz-3dplot}
\usepackage{pgfplots}
\pgfplotsset{compat=1.17}
\usepackage{comment}
\usepackage{xcolor}
\usepackage{subfig}
\usepackage[compat=1.1.0]{tikz-feynman}
\usepackage{amsmath}

\usetikzlibrary{matrix,calc,positioning,decorations.markings,decorations.pathmorphing,decorations.pathreplacing,cd}
\usetikzlibrary{arrows,decorations.markings}
\usetikzlibrary{shapes.misc}
\usetikzlibrary{shapes.geometric}
\usetikzlibrary{shapes.multipart}
\usetikzlibrary{decorations.markings}
\usetikzlibrary{backgrounds}
\usetikzlibrary{shadows}

\usetikzlibrary{automata,positioning}
\tikzset{snake it/.style={decorate, decoration=snake}}

\pgfarrowsdeclaredouble{doubled}{doubled}{stealth}{stealth}
\tikzset{
->-/.style={postaction={decorate,
   decoration={markings,mark=at position .53 with {\arrow{stealth};}}}
   },
->>-/.style={postaction={decorate,
   decoration={markings,mark=at position .52 with {\arrow{doubled};}}}
   },   
}

\usetikzlibrary{shapes.geometric} 
\tikzset{cross/.style={cross out, draw=black, minimum size=2*(#1-\pgflinewidth), inner sep=0pt, outer sep=0pt},
cross/.default={4pt}}

\pgfdeclarelayer{bottom}
\pgfdeclarelayer{middleb}
\pgfdeclarelayer{middlet}
\pgfdeclarelayer{top}
\pgfsetlayers{bottom,middleb,main,middlet,top}
\usepgfplotslibrary{colorbrewer}
\pgfplotsset{compat=newest}
\pgfplotsset{colormap/violet}

\pgfdeclarelayer{background layer}
\pgfdeclarelayer{foreground layer}
\pgfsetlayers{background layer,main,foreground layer}

\title{Stability of the $\pi$-Flux Phase for $\mathbb{Z}_{2}$ Lattice Gauge Theory Coupled to Fermionic Matter}
\author[1]{Leonardo Goller}
\author[1]{Marcello Porta}
\affil[1]{Mathematics Area, SISSA, Via Bonomea 265, 34136 Trieste, Italy}

\date{\today}

\begin{document}

\maketitle

\begin{abstract} We consider the two-dimensional $\mathbbm{Z}_{2}$ Ising gauge theory coupled to fermionic matter. In absence of electric fields, we prove that, at half-filling, the ground state of the gauge theory coincides with the $\pi$-flux phase, associated with magnetic flux equal to $\pi$ in every elementary lattice plaquette, provided the fermionic hopping is large enough. This proves in particular the semimetallic behavior of the ground state of the model. Furthermore, we compute the magnetic susceptibility of the gauge theory, and we prove that it is given by the one of massless $2d$ Dirac fermions, thus rigorously justifying recent numerical computations. The proof is based on reflection positivity and chessboard estimates, and on lattice conservation laws for the computation of the transport coefficient.
\end{abstract}

\tableofcontents

\section{Introduction}

The flux-phase conjecture stated that the optimum, energy minimizing magnetic flux for electrons hopping on a planar, bipartite graph at half-filling is equal to $\pi$ per lattice plaquette. This remarkable statement should be compared with the phenomenon of diamagnetism, namely the fact that the optimal magnetic flux per plaquette is zero if the density of particles is small enough. This conjecture has been proved by Lieb in a seminal paper \cite{Lieb}, in a form that applies to positive temperature, to higher dimensions and to a class of interacting models. An improved proof, that allows to further generalize the result, has been given by Macris and Nachtergaele in \cite{NM}. 

In \cite{Lieb, NM}, the gauge field enters via the choice of the hopping parameters of the lattice model, and it is static: the fermions hop on the two-dimensional lattice in the presence of a background gauge field. It is a natural question to understand what happens when the gauge field is promoted to a dynamical one, and this is the question that motivates the present work. We shall consider a system of fermions on a square lattice $\Gamma_{L}$ of side $L$, with periodic boundary conditions, coupled to a $\mathbb{Z}_{2}$-valued gauge field, living on the bonds of the lattice. The Hilbert space of the system is:
\begin{equation}
\mathcal{H}_{L} = \mathbb{C}^{2|E(\Gamma_{L})|} \otimes \mathcal{F}_{L}\;,
\end{equation}
where $E(\Gamma_{L})$ is the set of edges of $\Gamma_{L}$, and $\mathcal{F}_{L}$ is the usual fermionic Fock space. The Hamiltonian of the system is:
\begin{equation}\label{eq:Hdef}
H = -\sum_{\Lambda \in P(\Gamma_{L})} \Big[\prod_{(x,\mu) \in \partial \Lambda} Z_{x,\mu}\Big] -t \sum_{x\in \Gamma_{L}} \sum_{\mu = 1,2} (a^{+}_{x} Z_{x,\mu} a^{-}_{x+e_{\mu}} + \text{h.c.})\;,
\end{equation}
where: $P(\Gamma_{L})$ is the set of all elementary lattice plaquettes $\Lambda$ of $\Gamma_{L}$, $e_{1} = (1,0)$, $e_{2} = (0,1)$, and $(x,\mu) \equiv (x, x+e_{\mu})$ is a bond on $E(\Gamma_{L})$; $Z_{x,\mu}$ is the third Pauli matrix, associated with the bond $(x,\mu)$; $a^{\pm}_{x}$ are the usual fermionic creation and annihilation operators, creating or destroying a particle at $x\in \Gamma_{L}$. The quantity $t>0$ is the hopping parameter, and it defines the coupling constant of the system. This model in (\ref{eq:Hdef}) is actually a special case of the following more general Hamiltonian:
\begin{equation}\label{eq:Hgen}
H_{\varepsilon} = H + \varepsilon \sum_{(x,\mu) \in E(\Gamma_{L})} X_{x,\mu}\;,
\end{equation}
where $X_{x,\mu}$ is the first Pauli matrix. As discussed later in Section \ref{sec:model}, the $X$-operators have the interpretation of electric field operators, while the $Z$-operators have the interpretation of magnetic vector potential operators.

The Hamiltonian acts on a subspace $\mathcal{H}_{L}^{\text{phys}}$ of $\mathcal{H}_{L}$, called the physical subspace, formed by elements of $\mathcal{H}_{L}$ that satisfy a lattice version of the Gauss' law. Namely:
\begin{equation}\label{eq:Qgauss}
Q_{x} \psi = \psi\;,\qquad \text{for all $x\in \Gamma_{L}$}
\end{equation}
with $Q_{x} = A_{x} (-1)^{a^{+}_{x} a^{-}_{x}}$, and $A_{x}$ given by the ``star operator''
\begin{equation}
A_{x} = X_{x, 1} X_{x, 2} X_{x-e_{1},1} X_{x-e_{2},2}\;.
\end{equation}
The Hamiltonian is gauge invariant, in the sense that it commutes with all $Q_{x}$ operators. Eq. (\ref{eq:Qgauss}) can be viewed as an ``exponentiated'' version of the usual Gauss law, as will be discussed in more detail in Section \ref{sec:model}.

The pure gauge part is obtained setting $t=0$. The resulting model is called Ising gauge theory, and it is the simplest example of lattice gauge theory; see \cite{Fra} for a review. It has been first discussed in a seminal paper of Wegner \cite{Weg}, as a statistical mechanics system exhibiting a phase transition that cannot be detected via a local order parameter. The two phases correspond respectively to an area law (high temperature) versus a perimeter law (low temperature) for the expectation value of Wilson loop operators. The work \cite{Weg} actually considered a Euclidean model, where the gauge fields are represented by classical spins. In \cite{FS}, Fradkin and Susskind introduced the Hamiltonian formalism for the quantum system, on a physical Hilbert space that takes into account the Gauss' law with no sources. Then, in \cite{FrSc} Fradkin and Shenker considered the Ising gauge theory coupled to a dynamical matter field, described by a quantum Ising model, and discussed the existence of a continuous interpolation between the confined phase and a Higgs-like phase. We refer to \cite{Fra} for a pedagogical introduction to the subject, and for further references.

In the last years, there has been a lot of activity in the condensed matter community in the study of this gauge theory coupled to fermionic matter, see {\it e.g.} \cite{AG, gazit, prosko, gazit2, koenig}. The coupled model displays a rich phase diagram \cite{gazit}, as shown by numerical simulations; away from half-filling, the gauge fields mediate an attractive interaction between the fermions, which opens a spectral gap and realizes a superfluid state. Instead, at half-filling, the ground state of the gauge theory is described by a so-called orthogonal semimetallic state \cite{NMS, gazit2}, characterized by Dirac-like low energy excitations on the fermionic sector, and it is stable against against superconducting pairing; this phase turns out to be related to the emergence of the $\pi$-flux phase at zero temperature. The emergence of the semimetallic phase has been further investigated in \cite{prosko}, theoretically and numerically, and in \cite{koenig}, where the authors study the toric code model coupled to fermions. A physical signature of the emergent semimetallic phase is the value of a certain transport coefficient, the magnetic susceptibility. This transport coefficient has been studied in \cite{gazit}, where it has been shown numerically that, in a suitable parameter range, it agrees with the value predicted by $2+1$ dimensional massless Dirac fermions and for graphene \cite{Ando}. Most of the mentioned results are based on numerical simulations and Monte Carlo methods; in fact, the Monte Carlo analysis of the Ising gauge theory coupled to matter turns out to be unaffected by the sign problem, for an even number of fermion flavours ({\it e.g.} spinning fermions). Finally, let us also mention that the coupling of Majorana fermions on the honeycomb lattice with $\mathbb{Z}_{2}$ gauge fields has been studied in the context of an exactly solvable model for anyons in \cite{Kitaev}.

Here we will be interested in the rigorous analysis of the gauge theory coupled to fermionic matter. We will focus on the half-filling case, at $\varepsilon = 0$ (no external electric field). As $t\to \infty$, the model becomes purely fermionic, and it is a special case of the large class of systems studied in \cite{Lieb, NM}: at half-filling, the lowest energy is attained in correspondence with the $\pi$-flux phase. This is in contrast with what happens at $t=0$, where the ground state of the system is realized by gauge configuration with flux $0$ per lattice plaquette. Thus, a phase transition must occur at an intermediate value of the hopping parameter.

In this work we will prove that the ground state of the gauge theory at $\varepsilon = 0$ and at half-filling is still described by the $\pi$-flux phase, provided $t$ is large enough (of course, uniformly in the system's size). Thus, our result proves the stability of the $\pi$-flux phase in presence of a dynamical gauge field, and it also provides a quantitative estimate for how large the hopping parameter should be in order to fall within this phase. The result provides a precise control of the correlation functions of the ground state, and it shows the emergence of Dirac-like behavior at large distances. Also, we compute the magnetic susceptibility of the gauge theory, and we rigorously prove the agreement with massless Dirac fermions, observed numerically in \cite{gazit}.

The proof builds on the seminal paper \cite{Lieb}, and as \cite{Lieb} it is based on reflection positivity techniques. We adapt the strategy of \cite{Lieb} to the present case, where the main difference is the presence of the gauge constraint in the definition of the physical Hilbert space. Reflection positivity allows to prove that indeed the lowest fermionic energy is attained in correspondence with the $\pi$-flux phase; but this is not enough to prove that this phase is stable against dynamical fluctuations of the gauge field in the infinite volume limit. To prove stability, one has to quantity the price in (free) energy of the insertion of $0$-flux plaquettes in the $\pi$-flux phase, equivalent to the removal of magnetic monopoles from the system. To achieve this, we use a beautiful consequence of reflection positivity, the chessboard estimate \cite{FILS, FL}: this ultimately allows to prove a lower bound on the free energy cost of the monopoles' removal, and in particular it allows to prove that this cost grows {\it linearly} with the number of removed monopoles. The chessboard strategy actually provides a way to compute this cost, in terms of the free energy of a staggered monopoles' configuration, in which $\pi$ and $0$ flux plaquettes are alternated. Our proof of the chessboard inequality follows the exposition of \cite{Tasaki}, adapted to the present $\mathbb{Z}_{2}$-gauge theory context.

Thus, the strategy allows to prove that the zero temperature physics of the gauge theory is effectively described by the $\pi$-flux phase, as observed numerically. This phase is however not unique: the $\pi$-flux phase can be realized for different values of the $\mathbb{Z}_{2}$-fluxes across the non-contractible loops of the torus. This amounts to four gauge-inequivalent realizations of the $\pi$-flux phase, which correspond to effective periodic/antiperiodic boundary conditions for the fermionic degrees of freedom. The ground state energy of the fermionic system in the $\pi$-flux phase turns out to be very weakly dependent on the choice of these boundary conditions; in this sense, our result shows that the ground state of the gauge theory is approximately four-fold degenerate, in the infinite volume limit.

Finally, we use the explicit structure of the ground state of the gauge theory to compute the magnetic susceptibility. The analysis is based on lattice conservation laws and on the emergent Dirac-like form of the low energy spectrum, and it is based on a non-trivial adaptation of the strategy used in \cite{GMPcond} for graphene, see also \cite{GJMP, GMPhald} for the application to the critical Haldane-Hubbard model. With respect to \cite{GMPcond, GJMP, GMPhald}, here we study the response to a static and space-varying external magnetic field, instead of a space homogeneous and time-dependent electric field. Also, another difference with respect to the computation of the conductivity for semimetals is that here one ends up considering $8$ different cases, associated with periodic/antiperiodic boundary conditions in space and in imaginary time. After isolating an exact cancellation due to the presence of a zero mode, associated with a particular choice of space-time boundary conditions, the remaining cases turn out to contribute equally to the transport coefficient.

The paper is organized as follows. In Section \ref{sec:model} we define the model, we prove some basic facts about the Hamiltonian and the Gibbs state, and we state our main results: Theorem \ref{thm:main} for the stability of the $\pi$-flux phase, and Proposition \ref{prp:susc} for the computation of the magnetic susceptibility. In Section \ref{sec:RP} we adapt the argument of \cite{Lieb} to our gauge theory, and we prove the chessboard estimate, following \cite{Tasaki}. In particular, we show the optimality of the $\pi$-flux phase; also, we compute the fermionic ground state energy with $\pi$-flux background, and we discuss the weak dependence with respect to the fluxes across non-contractible loops. Then, using the chessboard estimate, we prove a lower bound on the increase in free energy due to the monopoles' removal, that scales linearly with the number of $0$-flux insertions. In Section \ref{sec:proofmain} we put everything together and we prove Theorem \ref{thm:main}, while in Section \ref{sec:proofsusc} we prove Proposition \ref{prp:susc}.

\paragraph{Acknowledgements.} We thank Alessandro Giuliani for insightful discussions about reflection positivity. We also thank Fabrizio Caragiulo, Simone Fabbri and Harman Preet Singh for useful comments and Davide Morgante for helpful advices with Ti\textit{k}Z pictures. L. G. and M. P. acknowledge support by the European Research Council through the ERC-StG MaMBoQ, n. 802901. M. P. acknowledges support from the MUR, PRIN 2022 project MaIQuFi cod. 20223J85K3. This work has been carried out under the auspices of the GNFM of INdAM. We gratefully acknowledge hospitality from the University of Z\"urich, where part of this work has been carried out.

\section{The model}\label{sec:model}

We consider a system of fermions on a two-dimensional lattice, coupled to a dynamical, $\mathbb{Z}_{2}$-valued gauge field. In this section we shall introduce the various objects entering the theory. Let $\Gamma_{L}$ be the square lattice of side $L$, with periodic boundary conditions:
\begin{equation}
\Gamma_{L} = \mathbb{Z}^{2} / L \mathbb{Z}^{2}\;.
\end{equation}
We shall denote by $x$ the points on $\Gamma_{L}$, which we will also call the vertices of $\Gamma_{L}$. Let us denote by $E(\Gamma_{L})$ the set of edges of $\Gamma_{L}$, which will be denoted by pairs $(x, x+ e_{\mu})$, with $\mu = 1,2$ and $e_{\mu}$ the standard basis of $\mathbb{R}^{2}$. We shall use the short-hand notation $(x, \mu)$ to denote the edge $(x, x+ e_{\mu})$. We shall denote by $P(\Gamma_{L})$ the set of elementary lattice plaquettes of $\Gamma_{L}$, which we will denote by $\Lambda$.

\subsection{Pure gauge sector}\label{sec:puregauge}
For each edge $(x,\mu)$ we associate a finite dimensional Hilbert space, which we identify with $\mathbb{C}^{2}$. The total Hilbert space of the gauge sector of the model is:
\begin{equation}
\mathcal{H}^{\text{g}}_{L} = \bigotimes_{(x,\mu) \in E(\Gamma_{L})} \mathbb{C}^{2} \simeq \mathbb{C}^{2|E(\Gamma_{L})|}\;.
\end{equation}
On this space, a special role will be played by the Pauli matrices $X_{x, \mu}$, $Z_{(x,\mu)}$. Recall the anticommutation relations:
\begin{equation}\label{eq:XZ}
\begin{split}
X_{x, \mu} Z_{x', \nu} &= (-1)^{\delta_{x, x'} \delta_{\mu,\nu}}  Z_{x', \nu} X_{x, \mu}\\
X_{x, \mu} X_{x', \nu} &=    X_{x', \nu} X_{x, \mu}\\
Z_{x, \mu} Z_{x', \nu} &= Z_{x', \nu} Z_{x, \mu}\\
X_{x, \mu}^2 &= Z_{x, \mu}^2 = \mathbbm{1}_{x, \mu}\;.   
\end{split}
\end{equation}
For each vertex $x \in \Gamma_{L}$, we define the star operator as, see Fig. \ref{fig:star}:
\begin{equation}\label{star}
A_{x} := X_{x, 1} X_{x, 2} X_{x-e_{1},1} X_{x-e_{2},2}\;.
\end{equation}
The operator $A_{x}$ implements the following transformation on the $Z$ operators:
\begin{equation}
A_{x}^{*} Z_{(x',\nu)} A_{x} = \left\{ \begin{array}{cc} -Z_{(x',\nu)} & \text{if $(x',\nu) \in \{(x, 1), (x, 2), (x-e_{1},1), (x-e_{2},2)\}$} \\ Z_{(x',\nu)} & \text{otherwise.} \end{array} \right.
\end{equation}
That is, the $A_{x}$ operators flip the signs of the $Z$ fields associated with the bonds touching the vertex $x$. 
\begin{figure}
\centering
\begin{tikzpicture}[scale=0.8]
\draw (0,0) grid (4,4);
\draw[->-] (0,4) -- (4,4);
\draw[->-] (0,0) -- (4,0);
\draw[->>-] (0,0) -- (0,4);
\draw[->>-] (4,0) -- (4,4);
\draw[red,fill=red] (1.5,2) circle (.09 cm);
\draw[red,fill=red] (2.5,2) circle (.09 cm);
\draw[red,fill=red] (2,1.5) circle (.09 cm);
\draw[red,fill=red] (2,2.5) circle (.09 cm);
\draw[red, line width = 0.07 cm] (2,1) -- (2,3);
\draw[red, line width = 0.07 cm] (1,2) -- (3,2);
\node[above right] (a) at (2,2) {\scalebox{0.8}{$x$}};
\node[below] (b) at (2,-0.5) {\scalebox{1.6}{$A_x$}};
\end{tikzpicture}
\caption{Graphical representation of the Star Operator.}
    \label{fig:star}
\end{figure}
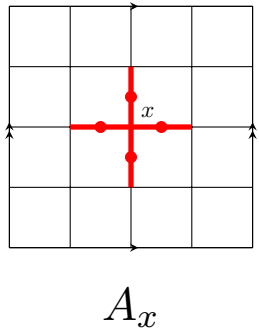
It is easy to check the following properties:
\begin{enumerate}
\item $[A_{x},  A_{x'}] = 0$
\item $A_{x}^2=\mathbbm{1}$
\item $\prod_{x \in \Gamma_{L}} A_{x} =\mathbbm{1}$.
\end{enumerate}
The last property can be viewed as a lattice version of Stokes' theorem ($\Gamma_{L}$ has no boundary). From these properties, we can define the family of projectors:
\begin{equation}
\frac{1 + A_{x}}{2}\;,\qquad \frac{1 - A_{x}}{2}\qquad \forall x\in \Gamma_{L}\;,
\end{equation}
which project over the subspace of $\mathbb{C}^{2^{4}}$ associated with the eigenvalue $+1$, resp. $-1$, of $A_{x}$. Physically, we shall say that the operator $A_{x}$ counts the number of charges associated with the vertex $x$: the eigenvalue $-1$ is associated with one charge present at $x$, while the eigenvalue $+1$ is associated with zero charges present at $x$; this point will be further discussed below.

Since the operators $A_{x}$ can be simultaneously diagonalized, the Hilbert space $\mathcal{H}^{\text{g}}_{L}$ splits into the direct sum of subspaces labelled by their eigenvalues; these subspaces are also called superselection sectors. The physical Hilbert space of the gauge sector is defined as the superselection sector with no background charges. It is:
\begin{equation}\label{eq:physH}
\mathcal{H}^{\text{phys}}_{L} := \big\{ \psi \in \mathcal{H}^{\text{g}}_{L} \, \big|\, A_{x} \psi = \psi\;,\quad \text{for all $x\in \Gamma_{L}$} \big\}\;;
\end{equation}
equivalently, we can represent the physical Hilbert space as:
\begin{equation}\label{eq:physpure}
\mathcal{H}^{\text{phys}}_{L} = \bigotimes_{x\in \Gamma_{L}} \left( \frac{1 + A_{x}}{2} \right) \mathcal{H}^{\text{g}}_{L}\;.
\end{equation}
In physical terms, we say that any vector in $\mathcal{H}^{\text{phys}}_{L}$ satisfies the vacuum Gauss' law. To understand this terminology, we think $X_{x,\mu}$ as the exponential of an electric field, $X_{x,\mu} = e^{i \pi E_{x,\mu}}$, with $E_{x,\mu}$ an operator with eigenvalues $0, 1$. Let us denote by $n_{x} = 0,1$ the number of fermionic charges sitting at $x$. The Gauss' law for the classical bond electric field reads $\text{d}_{x,1} E_{x,1} + \text{d}_{x,2} E_{x,2} = n_{x}$, with $\text{d}_{x,\mu}$ the discrete derivative, $\text{d}_{x,\mu} f(x) = f(x + e_{\mu}) - f(x)$ and $n_{x}$ the number of charges sitting at $x$. Thus, we say that the quantum state $\psi$ satisfies the Gauss' law if:
\begin{equation}\label{eq:gauss0}
X_{x,1} X_{x-e_{1},1}^{*} X_{x,2} X_{x-e_{2},2}^{*} \psi = e^{i \pi n_{x}} \psi\;.
\end{equation}
Using that $X_{x,\mu}^{2} = 1$, we see that vectors in (\ref{eq:physH}) satisfy (\ref{eq:gauss0}) with $n_{x} = 0$ for all $x$, {\it i.e.} no background charges. To complete the analogy with electromagnetism, we can think of $Z_{x,\mu}$ as the exponential of a magnetic vector potential, $Z_{x,\mu} = e^{i \pi A_{x,\mu}}$, with $A_{x,\mu}$ an operator with eigenvalues $0,1$; the properties (\ref{eq:XZ}) can be viewed as a consequence of the canonical commutation relations for the lattice electromagnetic field, with the normalization convention that the only nontrivial commutator is equal to $1/\pi$. In the following we shall refer to $X_{x,\mu}$ as the electric field operator, and to $Z_{x,\mu}$ as the vector potential operator.
\begin{remark}
Observe that, due to the property $\prod_{x \in \Gamma_{L}} A_{x} =\mathbbm{1}$, only $|\Gamma_{L}|-1$ of the constraints in the definition (\ref{eq:physH}) are independent: the physical Hilbert space has dimension $2^{|E(\Gamma_{L})|-|\Gamma_{L}|+1} = 2^{|\Gamma_{L}|-1}$.
\end{remark}
Next, we introduce the notion of physical, or gauge-invariant, pure-gauge observables.
\begin{definition}[Physical pure-gauge observables.] The physical pure-gauge observables are given by the elements of the algebra generated by the $X$ and $Z$ operators that commute with all $A_{x}$ operators.
 \end{definition}
 \begin{example}\label{ex:op}
\leavevmode
\begin{enumerate}
\item The simplest examples of gauge invariant operators are provided by polynomials in the electric field operator only. Let us make a few relevant examples (we refer to \cite{Fra} for a more extensive discussion).
\begin{enumerate}
\item The electric field operator $X_{x, \mu}$ on each edge $(x, \mu)$.
\item Let $\Gamma_{L}^{*}$ be the dual lattice of $\Gamma_{L}$. Elements of the dual lattice are given by the centers of the plaquettes of $\Gamma_{L}$, denoted by $x^{*}$. Choosing an open path $\mathcal{C}^*$ in the dual lattice $\Gamma_{L}^*$ ending in the plaquettes $x^*$ and $x'^*$, consider:
\begin{equation}
\tau^X_{(x^* ,x'^*)}:=   \prod_{(x, \mu)\cap \mathcal{C}^* \neq \emptyset} X_{x, \mu}\;;
\end{equation}
this operator corresponds to the product of electric field operator on the edges threaded by the path $\mathcal{C}^*$, and it is called the monopole pair creation operator.
\item Choosing a closed path $\mathcal{C}^*$ in the dual lattice $\Gamma_{L}^*$, consider:
\begin{equation}\label{thooft}
W^*_{\mathcal{C}^*}:=   \prod_{(x, \mu)\cap \mathcal{C}^* \neq \emptyset} X_{x,\mu}\;;
\end{equation}
this operator is called the ’t Hooft magnetic loop operator. In the physical Hilbert space, deformations of $\mathcal{C}^*$ by contractible cycles do not change the operator $ W^*_{\mathcal{C}^*}$.
\end{enumerate}
\item Let $\mathcal{C}$ be a closed path in $\Gamma_{L}$. The Wilson loop operator (or holonomy) around $\mathcal{C}$ is:
\begin{equation*}
W_{\mathcal{C}}:= \prod_{(x, \mu) \in \mathcal{C}} Z_{x, \mu}\;.
\end{equation*}
This operator measures the $\mathbb{Z}_2$ magnetic flux piercing the path $\mathcal{C}$ and thus, for contractible paths, the number (mod $2$) of monopoles in the interior of the path. If $\Lambda$ is a plaquette and $\mathcal{C} = \partial \Lambda$, we shall set $W_{\mathcal{C}} = B_{\Lambda}$. More generally, if $\Omega = \cup_{i} \Lambda_i$ and $\mathcal{C} = \partial \Omega$, 
\begin{equation*}
W_{\mathcal{C}} = \prod_{i} B_{\Lambda_i}\;.
\end{equation*}
In particular, if $\Omega = \Gamma_{L}$, since the $\partial \Gamma_{L} = \emptyset$ we also have:
\begin{equation}\label{stokesb}
\prod_{\Lambda \in \Omega} B_{\Lambda} = \mathbbm{1}\;.
\end{equation}
\end{enumerate}
\end{example}

\begin{figure}
\centering
\begin{tikzpicture}[scale=0.8]
\draw (0,0) grid (4,4);
\draw[->-] (0,4) -- (4,4);
\draw[->-] (0,0) -- (4,0);
\draw[->>-] (0,0) -- (0,4);
\draw[->>-] (4,0) -- (4,4);
\draw[red,fill=red] (1.5,2) circle (.09 cm);
\draw[line width = 0.07 cm] (1.5,1.5) node[cross,blue] {};
\draw[line width = 0.07 cm] (1.5,2.5) node[cross,blue] {};
\draw[dashed, thick] (1.5,1.5) -- (1.5,2.5);
\draw[thick, red, line width = 0.07 cm] (1,2) -- (2,2);
\node[below] (a) at (2,-0.5) {\scalebox{1.6}{$X_{(x,1)}$}};
\begin{scope}[xshift=5.5 cm]
\draw (0,0) grid (4,4);
\draw[->-] (0,4) -- (4,4);
\draw[->-] (0,0) -- (4,0);
\draw[->>-] (0,0) -- (0,4);
\draw[->>-] (4,0) -- (4,4);
\draw[red,fill=red] (1.5,0) circle (.09 cm);
\draw[red,fill=red] (1.5,1) circle (.09 cm);
\draw[red,fill=red] (1.5,2) circle (.09 cm);
\draw[red,fill=red] (1.5,3) circle (.09 cm);
\draw[red,fill=red] (1.5,4) circle (.09 cm);
\draw[dashed, thick] (1.5,0) -- (1.5,4);
\node[right] (a) at (1.45,2.5) {\scalebox{0.8}{$\mathcal{C}^*$}};
\draw[thick, red, line width = 0.07 cm] (1,0) -- (2,0);
\draw[thick, red, line width = 0.07 cm] (1,1) -- (2,1);
\draw[thick, red, line width = 0.07 cm] (1,2) -- (2,2);
\draw[thick, red, line width = 0.07 cm] (1,3) -- (2,3);
\draw[thick, red, line width = 0.07 cm] (1,4) -- (2,4);
\node[below] (a) at (2,-0.5) {\scalebox{1.6}{$W^{*}_{\mathcal{C}^*}$}};
\end{scope}
\begin{scope}[xshift=11cm]
\draw (0,0) grid (4,4);
\draw[->-] (0,4) -- (4,4);
\draw[->-] (0,0) -- (4,0);
\draw[->>-] (0,0) -- (0,4);
\draw[->>-] (4,0) -- (4,4);
\draw[red,fill=red] (0.5,3) circle (.09 cm);
\draw[red,fill=red] (1,2.5) circle (.09 cm);
\draw[red,fill=red] (2,2.5) circle (.09 cm);
\draw[red,fill=red] (2.5,2) circle (.09 cm);
\draw[dashed, thick] (0.5, 3.5) -- (0.5,2.5) -- (2.5,2.5) -- (2.5,1.5);
\draw[thick, red, line width = 0.07 cm] (0,3) -- (1,3);
\draw[thick, red, line width = 0.07 cm] (1,3) -- (1,2);
\draw[thick, red, line width = 0.07 cm] (2,3) -- (2,2);
\draw[thick, red, line width = 0.07 cm] (2,3) -- (2,3);
\draw[thick, red, line width = 0.07 cm] (2,2) -- (3,2);
\node[above] (b) at (0.5,3.55) {\scalebox{0.8}{$x^*$}};
\node[below] (b) at (2.5,1.47) {\scalebox{0.8}{${x'}^*$}};
\draw[line width = 0.05 cm] (0.5,3.5) node[cross,blue] {};
\draw[line width = 0.05 cm] (2.5,1.5) node[cross,blue] {};
\node[below] (a) at (2,-0.3) {\scalebox{1.6}{$\tau^X_{x^*,{x'}^*}$}};
\end{scope}
\begin{scope}[yshift=-6cm]
\draw (0,0) grid (4,4);
\draw[->-] (0,4) -- (4,4);
\draw[->-] (0,0) -- (4,0);
\draw[->>-] (0,0) -- (0,4);
\draw[->>-] (4,0) -- (4,4);
\draw[blue,fill=blue] (1.5,2) circle (.09 cm);
\draw[blue,fill=blue] (1.5,3) circle (.09 cm);
\draw[blue,fill=blue] (2,2.5) circle (.09 cm);
\draw[blue,fill=blue] (1,2.5) circle (.09 cm);
\draw[thick, blue, line width =0.07 cm] (1,2)--(2,2)--(2,3)--(1,3)--(1,2);
\node (a) at (1.5,2.5)  {\scalebox{1.6}{$\Lambda$}};
\node[below] (a) at (2,-0.5) {\scalebox{1.6}{$B_{\Lambda}$}};
\end{scope}
\begin{scope}[yshift=-6cm, xshift = 5.5cm]
\draw (0,0) grid (4,4);
\draw[->-] (0,4) -- (4,4);
\draw[->-] (0,0) -- (4,0);
\draw[->>-] (0,0) -- (0,4);
\draw[->>-] (4,0) -- (4,4);
\draw[blue,fill=blue] (1.5,1) circle (.09 cm);
\draw[blue,fill=blue] (2.5,1) circle (.09 cm);
\draw[blue,fill=blue] (3,1.5) circle (.09 cm);
\draw[blue,fill=blue] (2.5,2) circle (.09 cm);
\draw[blue,fill=blue] (2,2.5) circle (.09 cm);
\draw[blue,fill=blue] (1.5,3) circle (.09 cm);
\draw[blue,fill=blue] (1,2.5) circle (.09 cm);
\draw[blue,fill=blue] (1,1.5) circle (.09 cm);
\draw[thick, blue, line width =0.07 cm] (1,1)--(3,1)--(3,2)--(2,2)--(2,3)--(1,3)--(1,1);
\node (a) at (1.5,1.5)  {\scalebox{1.6}{$\Omega$}};
\node[below] (b) at (1.8,1) {\scalebox{0.8}{$\partial\Omega = \mathcal{C}$}};
\node[below] (a) at (2,-0.5) {\scalebox{1.6}{$W_{\mathcal{C}}$}};
\end{scope}
\begin{scope}[yshift=-6cm, xshift = 11cm]
\draw (0,0) grid (4,4);
\draw[->-] (0,4) -- (4,4);
\draw[->-] (0,0) -- (4,0);
\draw[->>-] (0,0) -- (0,4);
\draw[->>-] (4,0) -- (4,4);
\draw[blue,fill=blue] (1,0.5) circle (.09 cm);
\draw[blue,fill=blue] (1,1.5) circle (.09 cm);
\draw[blue,fill=blue] (1,2.5) circle (.09 cm);
\draw[blue,fill=blue] (1,3.5) circle (.09 cm);
\draw[thick, blue, line width =0.07 cm] (1,0)--(1,4);
\node[right] (a) at (1,2.5)  {\scalebox{0.8}{$\mathcal{C}$}};
\node[below] (a) at (2,-0.5) {\scalebox{1.6}{$W_{\mathcal{C}}$}};
\end{scope}
\end{tikzpicture}
\caption{Graphical representation of operators introduced in Example \ref{ex:op}.}
\end{figure}
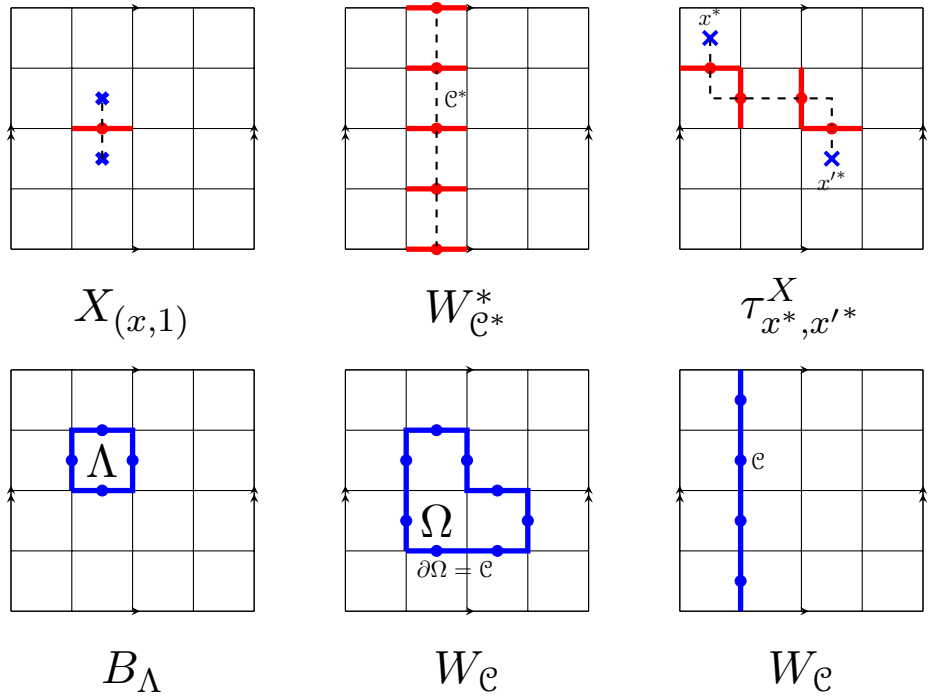
Next, we introduce the Hamiltonian for the gauge sector of the $\mathbb{Z}_{2}$-lattice gauge theory. The pure gauge Hamiltonian is:
\begin{equation}\label{eq:Hgaugeg}
H_{g}(\varepsilon) := - \sum_{\Lambda \in P(\Gamma_{L})} B_{\Lambda} -\varepsilon \sum_{(x,\mu) \in E(\Gamma_{L})} X_{x,\mu}\;.
\end{equation}
The first term favors configurations with $B_{\Lambda} = 1$ (absence of magnetic monopoles), while the second term introduces the possibility for monopole creation/annihilation. 

In this work, we will focus on the case $\varepsilon=0$, where the dynamics generated by $H_{g}(0)$ preserves the number of monopoles in the system. The ground state of $H_{g}(0)$ is described by the monopole-free condition $B_{\Lambda} = 1$ for all $\Lambda \in P(\Gamma_{L})$. Enforcing this condition (and recalling that not all these conditions are independent, due to \eqref{stokesb}), the number of ground states of the pure gauge Hamiltonian for $\varepsilon = 0$ on the physical Hilbert space is:
\footnote{In the left-hand side of (\ref{eq:Euler}): at the numerator, we have the total number of spin configurations, with no constraints; at the denominator, we take into account the presence of the constraints $A_{x} \psi = \psi$ for all $x\in \Gamma_{L}$, which are not all independent due to $\prod_{x\in \Gamma_{L}} A_{x} = 1$; and the presence of the constraints $B_{\Lambda} = 1$ for all $\Lambda \in P(\Gamma_{L})$, which are not all independent, due to $\prod_{\Lambda \in P(\Gamma_{L})} B_{\Lambda} = 1$.}
\begin{equation}\label{eq:Euler}
\frac{2^{|E(\Gamma_{L})|}}{2^{|\Gamma_{L}| - 1} 2^{|P(\Gamma_{L})| - 1}} = 2^{-\chi(\Gamma_{L}) +2} = 4  
\end{equation}
where $\chi(\Gamma_{L})$ is the Euler characteristic of $\Gamma_{L}$. Each ground state is characterized by a choice of $\mathbb{Z}_2$-fluxes along the non contractible loops of $\Gamma_{L}$.
\subsection{Coupling with fermionic matter}
We will now describe the coupling of the gauge fields with fermionic matter, in a second-quantized formalism. The Hilbert space of the fermionic degrees of freedom is the Fock space,
\begin{equation}
\mathcal{F}_{L} = \mathbb{C} \oplus \bigoplus_{n\geq 1} \ell^{2}(\Gamma_{L})^{\wedge n}\simeq \mathbb{C}^{2^{|\Gamma_{L}|}}\;.
\end{equation}
For any $x\in \Gamma_{L}$, we associate the usual fermionic creation and annihilation operators $a^{+}_{x}$ and $a^{-}_{x}$, which satisfy the canonical anticommutation relations:
\begin{equation}
\{ a^{+}_{x}, a^{-}_{y} \} = \delta_{x,y}\;,\qquad  \{ a^{+}_{x}, a^{+}_{y} \} = \{ a^{-}_{x}, a^{-}_{y} \}= 0\;,
\end{equation}
and with the understanding that $a^{+}_{x} = (a^{-}_{x})^{*}$. The algebra $\mathcal{A}$ of fermionic observables is given by the (self-adjoint) polynomials in the creation and annihilation operators. The simplest example is the number operator,
\begin{equation}
N = \sum_{x\in \Gamma_{L}} a^{+}_{x} a^{-}_{x}\;.
\end{equation}
Given a fermionic observable $\mathcal{O}$, we define the parity automorphism $\mathcal{P}$ as:
\begin{equation}
\mathcal{P}(\mathcal{O}) = (-1)^{N} \mathcal{O} (-1)^{N}\;.
\end{equation}  
Being an involution, the eigenvalues of $\mathcal{P}$ are $\pm 1$. We denote by $\mathcal{A}_{+}$ the set of polynomials which are even under $\mathcal{P}$ (eigenvalue $+1$) and by $\mathcal{A}_{-}$ the algebra of the polynomials that are odd under $\mathcal{P}$ (eigenvalue $-1$). In the following, we shall always consider physical observables that belong to the even subalgebra $\mathcal{A}_{+}$. In other words, all the physical observables we shall consider in the following satisfy the global symmetry:
\begin{equation}\label{eq:thetaO}
\mathcal{O} = \mathcal{P} (\mathcal{O})\;.
\end{equation}
We shall now proceed to gauge this symmetry, and to introduce the lattice gauge theory. To this end, we consider the total Hilbert space, for matter coupled to gauge fields,
\begin{equation}
\mathcal{H}_{L} = \mathcal{F}_{L} \otimes \mathcal{H}^{\text{g}}_{L}\;.
\end{equation} 
In order to introduce the physical Hilbert space for the combined system, we will consider states satisfying the Gauss' law in presence of dynamical matter fields. To this end, let us define the $\mathbb{Z}_{2}$-charge operator as:
\begin{equation}
Q_{x} = A_{x} (-1)^{n_{x}}\;,
\end{equation}
where $n_{x} = a^{+}_{x} a^{-}_{x}$. The $\mathbb{Z}_{2}$-charge operator satisfies the following properties, in extension to the properties introduced in Section \ref{sec:puregauge}:
\begin{enumerate}
\item $[Q_{x}, Q_{x'}] = 0$ for all $x,x'$;
\item $Q_{x}^{2} = \mathbbm{1}$;
\item  $\prod_{x \in \Gamma_{L}} Q_{x} = (-1)^{N}$.
\end{enumerate}
The spectrum of $Q_{x}$ consists of $\pm 1$. By the discussion after (\ref{eq:physpure}), we say that a state $\psi$ in the Hilbert space $\mathcal{H}_{L}$ satisfies the Gauss' law if $Q_{x} \psi = \psi$ for all $x\in \Gamma_{L}$. Thus, we define the physical Hilbert space for the coupled system as:
\begin{equation}\label{eq:Htotproj}
\mathcal{H}^{\text{phys}}_{L} := \bigotimes_{x\in \Gamma_{L}} \left( \frac{1 + Q_{x}}{2} \right)\mathcal{H}_{L}\;.
\end{equation}
Next, we extend the notion of physical, or gauge-invariant, observables, in presence of matter fields.
\begin{definition}[Physical observables]\label{def:gi}  The physical observables are given by self-adjoint polynomials in the $X, Z$ operators and in the fermionic operators $a^{\pm}$, that commute with all $Q_{x}$ operators. We shall say that a physical observable $\mathcal{O}$ is magnetic if it belongs to the subalgebra generated by the vector potential operators and by the fermionic operators.
\end{definition}
An example of gauge-invariant observable is given by the product of creation and annihilation operators, dressed by a $\mathbb{Z}_{2}$-Wilson line. That is:
\begin{equation}
W_{x,x'} = a^{+}_{x} \Big(\prod_{(z,\mu) \in \mathcal{C}} Z_{(z,\mu)}\Big) a^{-}_{x'}\;,
\end{equation}
where $\mathcal{C}$ is a path connecting $x$ to $x'$; see Fig. \ref{fig:path}.
\subsection{The Hamiltonian and the Gibbs state}
We now have all the ingredients to introduce the Hamiltonian of the full system, and to define its grand-canonical Gibbs state. 
\begin{figure}
\centering
\begin{tikzpicture}[scale=0.8]
\draw (0,0) grid (4,4);
\draw[->-] (0,4) -- (4,4);
\draw[->-] (0,0) -- (4,0);
\draw[->>-] (0,0) -- (0,4);
\draw[->>-] (4,0) -- (4,4);
\draw[blue,fill=blue] (1,2.5) circle (.09 cm);
\draw[blue,fill=blue] (1.5,2) circle (.09 cm);
\draw[blue,fill=blue] (2.5,2) circle (.09 cm);
\draw[blue,fill=blue] (3,1.5) circle (.09 cm);
\draw[blue, line width = 0.07 cm] (1,3) -- (1,2) -- (3,2) -- (3,1);
\draw[black,fill=black] (1,3) circle (.09 cm);
\draw[black,fill=black] (3,1) circle (.09 cm);
\node[above] (d) at (2.5,2.1) {\scalebox{0.8}{$\mathcal{C}$}};
\node[above left] (a) at (1,3) {\scalebox{0.8}{$x$}};
\node[below right] (c) at (3,1) {\scalebox{0.8}{$x'$}};
\node[below] (b) at (2,-0.5) {\scalebox{1.6}{$W_{x,x'}$}};
\end{tikzpicture}
\caption{Graphical representation of $W_{x,x'}$.}
    \label{fig:path}
\end{figure}
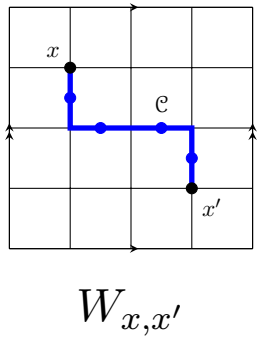

The Hamiltonian is:
\begin{equation}
H(\varepsilon) := H_{g}(\varepsilon) + H_{f}\;,
\end{equation}
where $H_{f}$ describes the fermionic sector of the Hamiltonian, and it is given by the magnetic lattice Laplacian in second quantization:
\begin{equation}\label{eq:Hmatter}
H_{f} := -t \sum_{x\in \Gamma_{L}} \sum_{\mu = 1,2} a^{+}_{x} Z_{x,\mu} a^{-}_{x+e_{\mu}} + \text{h.c.}\;.
\end{equation}
The constant $t>0$ is the hopping parameter of the model. Since we set to $1$ the prefactor of the magnetic term in (\ref{eq:Hgaugeg}), $t$ can be thought as the gauge coupling of the system. Observe that $H(\varepsilon)$ is a gauge-invariant observable, in the sense of Definition \ref{def:gi}. 

Let us introduce the short-hand notation $\mathcal{H}_{L} \equiv \mathcal{H}^{\text{tot}; \text{phys}}_{L}$ for the total, physical Hilbert space. Given a physical observable $\mathcal{O}$, the Gibbs state of the model is:
\begin{equation}\label{eq:gibbs}
\langle \mathcal{O} \rangle^{\varepsilon}_{\beta,\mu,L} := \frac{\Tr_{\mathcal{H}^{\text{phys}}_{L}} \mathcal{O} e^{-\beta (H(\varepsilon) - \mu N)}}{\Tr_{\mathcal{H}^{\text{phys}}_{L}} e^{-\beta (H(\varepsilon) - \mu N)}}\;,
\end{equation}
where $\mu$ is the chemical potential of the system. In view of (\ref{eq:Htotproj}), we can rewrite the expectation value in (\ref{eq:gibbs}) as:
\begin{equation}
\langle \mathcal{O} \rangle^{\varepsilon}_{\beta,\mu,L} = \frac{\Tr_{\mathcal{H}_{L}} \prod_{x\in \Gamma_{L}} \left( \frac{1 + Q_{x}}{2} \right) \mathcal{O} e^{-\beta (H(\varepsilon) - \mu N)}}{\Tr_{\mathcal{H}_{L}} \prod_{x\in \Gamma_{L}} \left( \frac{1 + Q_{x}}{2} \right) e^{-\beta (H(\varepsilon) - \mu N)}}\;.
\end{equation}
We will be interested in the properties of the Gibbs state of the system, at low temperatures. To this end, it is convenient to preliminarily explore the spectral properties of the matter Hamiltonian. We shall say that two set of eigenvalues ${\bm \sigma} = \{\sigma_{x,\mu}\}$ and ${\bm \sigma}' = \{\sigma'_{x,\mu}\}$ of the vector potential operators ${\bm Z} = \{Z_{x,\mu}\}$ are gauge equivalent if they have the same fluxes around all plaquettes of $\Gamma_{L}$ and around all non-contractible cycles of $\Gamma_{L}$.  Also, we shall denote by $H_{f}({\bm \sigma})$ the fermionic Hamiltonian (\ref{eq:Hmatter}) after replacing the operators $\{Z_{\mu,x}\}$ with their eigenvalues. The next proposition is a direct adaptation of \cite{LL}, to the case of periodic boundary conditions. It shows that gauge equivalent configurations give rise to the same spectrum of $H_{f}({\bm \sigma})$.
\begin{proposition}[Gauge equivalent configurations are isospectral]\label{prp:gaugeequiv} If $\{\sigma_{x,\mu}\}$ and $\{\sigma'_{x,\mu}\}$ are gauge equivalent, there exists a unitary operator $U: \mathcal{F}_{L} \to \mathcal{F}_{L}$ such that:
\begin{equation}
U H_{f}(\bm{\sigma}) U^{*} = H_{f}(\bm{\sigma}')\;.
\end{equation} 
\end{proposition}
\begin{proof} We repeat the argument of \cite{LL}, taking into account fluxes around non-contractible cycles, which are possible due to the periodic boundary conditions. Let us define:
\begin{equation}
w(x,\mu) := \sigma_{x,\mu} \sigma'_{x,\mu}\;.
\end{equation}
The product of $w(x,\mu)$ along any curve $\mathcal{C}$ on $E(\Gamma_{L})$, joining two vertices $x$ and $x'$, only depends on the endpoints, and does not depend on the path joining them. To see this, let $\mathcal{C}'$ be another path joining $x$ to $x'$. Then, the composition $\mathcal{C} \circ \mathcal{C}' =: \widetilde{\mathcal{C}}$ is a cycle. The product:
\begin{equation}\label{eq:prodC}
\prod_{(x,\mu) \in \widetilde{\mathcal{C}}} w(x,\mu)
\end{equation}
can be decomposed as a product over contractible and non-contractible cycles (simply using that $w(x,\mu)^{2} = 1$). The product over all such cycles is $1$, by gauge equivalence;  thus, (\ref{eq:prodC}) is equal to one.

Consider the points $x$ and $x + e_{\mu}$, forming the edge $(x,\mu)$. Pick an arbitrary point $x_{0} \in \Gamma_{L}$. Consider the curve $\mathcal{C}$, connecting $x_{0}$ to $x$, and the curve $\mathcal{C}'$, connecting $x + e_{\mu}$ to $x_{0}$. Thus, $\mathcal{C} \circ (x,\mu) \circ \mathcal{C}'$ forms a loop, and hence, by the previous discussion:
\begin{equation}
\Big( \prod_{(x',\mu') \in \mathcal{C}} w(x',\mu') \Big) w(x,\mu) \Big( \prod_{(x'',\mu'') \in \mathcal{C}'} w(x'',\mu'') \Big) = 1\;.
\end{equation}
Thus, if we call:
\begin{equation}
(-1)^{\phi_{x}} := \Big( \prod_{(x',\mu') \in \mathcal{C}} w(x',\mu') \Big)\;,\qquad (-1)^{\phi_{x+e_{\mu}}} = \Big( \prod_{(x'',\mu'') \in \mathcal{C}'} w(x'',\mu'') \Big)\;, 
\end{equation}
we obtain:
\begin{equation}\label{eq:phix}
1 = (-1)^{\phi_{x}} \sigma_{x,\mu} \sigma'_{x,\mu} (-1)^{\phi_{x+e_{\mu}}} \Rightarrow \sigma'_{x,\mu}  =  (-1)^{\phi_{x}} \sigma_{x,\mu} (-1)^{\phi_{x+e_{\mu}}}\;.
\end{equation}
Consider now the unitary operator $U$ on the Fock space $\mathcal{F}_{L}$:
\begin{equation}
U = (-1)^{\sum_{x} \phi_{x} n_{x}}\;,
\end{equation}
with $\{\phi_{x}\}$ as in (\ref{eq:phix}) (recall that $(-1)^{\phi_{x}} \in \{\pm 1\}$). The operator $U$ is unitary, and it implements the phase transformation:
\begin{equation}\label{eq:gaugea}
U^{*} a^{\pm}_{x} U = (-1)^{\phi_{x}} a^{\pm}_{x}\;.
\end{equation}
Therefore, combining (\ref{eq:gaugea}) and (\ref{eq:phix}), we get:
\begin{equation}
U^{*} H_{f}(\bm{\sigma}) U = H_{f}(\bm{\sigma}')\;,
\end{equation}
which concludes the proof of the proposition.
\end{proof}
\begin{remark} Thus, the spectrum of $H_{f}(\bm{\sigma})$ depends on the configuration $\{\sigma_{x,\mu}\}$ only via the magnetic fluxes though the plaquettes and through the non-contractible cycles.
\end{remark}
From now on, we will consider the model in the absence of electric fields, $\varepsilon = 0$. In this situation, the gauge fields are effectively represented by classical spins $\pm 1$, attached to the bonds of the lattice model. We will set $H = H(\varepsilon)|_{\varepsilon = 0}$, and we will denote by $H({\bm \sigma})$ the Hamiltonian $H$ evaluated on a given gauge field configuration:
\begin{equation}\label{hami}
H({\bm \sigma}) = - \sum_{\Lambda \in P(\Gamma_{L})} \Big[\prod_{(x,\mu) \in \partial \Lambda} \sigma_{x,\mu}\Big] + H_{f}({\bm \sigma})\;.
\end{equation}
We will also denote $\langle \cdot \rangle_{\beta, \mu, L} := \langle \cdot \rangle_{\beta, \mu, L}^{\varepsilon}|_{\varepsilon = 0}$ the Gibbs state in the absence of electric fields. We will consider observables in the algebra generated by the vector potential operators and by the creation/annihilation operators, {\it i.e.} magnetic observables in the sense of Definition \ref{def:gi}; as it will be discussed later in Remark \ref{rem:nonma}, this is not a loss of generality. Given an observable $\mathcal{O}$ of this type, we shall denote $\mathcal{O}({\bm \sigma}) = \langle {\bm \sigma}, \mathcal{O} {\bm \sigma}\rangle$ its projection on the fermionic sector, where $|{\bm \sigma} \rangle$ is an eigenstate of the vector potential operators ${\bm Z} = \{Z_{x,\mu}\}$ with eigenvalues ${\bm \sigma} = \{\sigma_{x,\mu}\}$, with $(x,\mu) \in E(\Gamma_{L})$. 

For $\varepsilon = 0$, the Gauss' law constraint in the definition of the physical Hilbert space (\ref{eq:Htotproj}) can be resolved explicitly. This is the content of the next proposition.
\begin{proposition}[Resolution of the gauge constraint]\label{prp:reso} Let $\varepsilon = 0$, and let $\mathcal{O}$ be a magnetic gauge invariant observable. Then:
\begin{equation}
\langle \mathcal{O} \rangle_{\beta, \mu, L} = \frac{ \sum_{{\bm \sigma}} \Tr_{\mathcal{F}_{L}}  (1 + (-1)^{N}) \mathcal{O}(\bm{\sigma}) e^{-\beta (H(\bm{\sigma}) - \mu N)} }{ \sum_{{\bm \sigma}} \Tr_{\mathcal{F}_{L}}  (1 + (-1)^{N})  e^{-\beta (H(\bm{\sigma}) - \mu N)} }\;.
\end{equation}
In particular, the partition function of the system can be rewritten as:
\begin{equation}
Z_{\beta,\mu,L} = \sum_{{\bm \sigma}} \Tr_{\mathcal{F}_{L}} \frac{(1 + (-1)^{N})}{2^{|\Gamma_{L}|}} e^{-\beta (H(\bm{\sigma}) - \mu N)}\;.
\end{equation}
\end{proposition}
\begin{proof} Let us start from the partition function. We have, recalling (\ref{eq:Htotproj}), parametrizing the trace over the fermionic Fock space using the occupation number basis, whose elements we denote by $|{\bm n} \rangle = |\{n_{x}\}_{x\in \Gamma_{L}}\rangle$, and the trace over the gauge fields using the basis formed by the eigenstates of $\{Z_{x,\mu}\}$, whose elements are denoted by $|{\bm \sigma} \rangle = | \{ \sigma_{x,\mu} \}_{(x,\mu) \in E(\Gamma_{L})} \rangle$:
\begin{equation}
\begin{split}
Z_{\beta,\mu,L} &= \Tr_{\mathcal{H}_{L}} \left(\prod_{x \in V(\Gamma)} \left(\frac{1+Q_{x}}{2}\right) e^{-\beta (H - \mu N)}\right)\\
&= \sum_{ {\bm \sigma}, {\bm n}} \Big\langle {\bm n}\;; {\bm \sigma} \, \Big| \prod_{x\in \Gamma_{L}} \left( \frac{1 + Q_{x}}{2} \right) e^{-\beta (H - \mu N)} \Big| {\bm n} \;; {\bm \sigma}  \Big\rangle\;,
\end{split}
\end{equation}
where $\Big| {\bm n}\;; {\bm \sigma}  \Big\rangle= \Big| {\bm n} \Big\rangle \otimes \Big| {\bm \sigma}  \Big\rangle$. Thus, since all $\{ Q_{x} \}$ commute with all $a^{+}_{x}a^{-}_{x}$, and since $H, N$ commute with all $Z_{x,\mu}$, we have:
\begin{equation}\label{eq:zsum}
Z_{\beta,\mu,L} =  \sum_{{\bm \sigma}, {\bm n}} \Big\langle {\bm n} \Big| e^{-\beta (H(\bm{\sigma}) - \mu N)} \Big| {\bm n} \Big\rangle \Big\langle {\bm \sigma} \Big| \prod_{x\in \Gamma_{L}} \left( \frac{1 + A_{x} (-1)^{n_{x}}}{2} \right) \Big| {\bm \sigma} \Big\rangle\;.
\end{equation}
Similarly, we also have:
\begin{equation}\label{eq:Osum}
\begin{split}
\langle \mathcal{O} \rangle_{\beta,\mu,L} &= \frac{1}{Z_{\beta,\mu,L}} \sum_{{\bm \sigma}, {\bm n}} \Big\langle {\bm n} \Big| O(\bm{\sigma}) e^{-\beta (H(\bm{\sigma}) - \mu N)} \Big| {\bm n} \Big\rangle \\&\quad \cdot  \Big\langle {\bm \sigma} \Big| \prod_{x\in \Gamma_{L}} \left( \frac{1 + A_{x} (-1)^{n_{x}}}{2} \right) \Big| {\bm \sigma} \Big\rangle\;.
\end{split}
\end{equation}
Let us analyse the last term in the expressions (\ref{eq:zsum}), (\ref{eq:Osum}). Observe that the operator $A_{x}$ swaps eigenstates of $Z_{x,\mu}$. Expanding the products over $x\in \Gamma_{L}$, the only terms with a non-vanishing scalar product are $1$ and $\prod_{x\in \Gamma_{L}} A_{x}(-1)^{n_{x}} =(-1)^{N}$ (here we used that all $X_{x,\mu}$ operators appear twice, and that $X_{x,\mu}^{2} = 1$). Therefore,
\begin{equation}
\Big\langle {\bm \sigma} \Big| \prod_{x\in \Gamma_{L}} \left( \frac{1 + A_{x} (-1)^{n_{x}}}{2} \right) \Big| {\bm \sigma} \Big\rangle = \frac{1 + (-1)^{N}}{2^{|\Gamma_{L}|}}\;. 
\end{equation}
Inserting this expression in (\ref{eq:zsum}), (\ref{eq:Osum}), the final claim follows.
\end{proof}
\begin{remark}
\leavevmode
\begin{itemize}
\item[(i)] Observe that, by Proposition \ref{prp:gaugeequiv}, the sums over the gauge field configurations reduce to sums over inequivalent configurations, that can be parametrized as sums over choices of magnetic fluxes through plaquettes and non-contractible loops, up to multiplication by the volume of gauge orbits.
\item[(ii)] In particular, we can represent the partition function as:
\begin{equation}
Z_{\beta,\mu,L} = \sum_{\Phi} \sum_{a,b} Z_{\beta,\mu,L}(\Phi,a,b)\;,
\end{equation}
where: $\Phi$ is a configuration of fluxes through the plaquettes; $a$ and $b$ are the fluxes associated with two inequivalent, non-contractible loops of the torus, horizontal resp. vertical; and:
\begin{equation}
\begin{split}
Z_{\beta,\mu,L}(\Phi,a,b) &= \sum^{*}_{{\bm \sigma}} \Tr_{\mathcal{F}_{L}} \frac{(1 + (-1)^{N})}{2^{|\Gamma_{L}|}} e^{-\beta (H(\bm{\sigma}) - \mu N)} \\
&= \frac{1}{2} \Tr_{\mathcal{F}_{L}} (1 + (-1)^{N}) e^{-\beta (H(\Phi, a, b) - \mu N)}\;.
\end{split}
\end{equation}
The asterisk in the first line means that the gauge field configurations are compatible with the fluxes $\Phi, a, b$; $H(\Phi, a, b) \equiv H(\bm{\sigma}(\Phi, a, b))$ is the Hamiltonian evaluated over a gauge field configuration compatible with the flux configuration $\Phi, a, b$; and in the last step we used that the number of gauge field configurations compatible with a given flux configuration is:\footnote{At the exponent, the factor $-(|P(\Gamma_{L})| - 1)$ is due to the plaquettes' flux constraint, while the factor $-2$ is due to the choice of the fluxes on the non-contractible loops. }
\begin{equation}
2^{|E(\Gamma_{L})| - (|P(\Gamma_{L})| - 1) - 2} = 2^{|\Gamma_{L}| - 1}\;,
\end{equation}
recall (\ref{eq:Euler}).
\end{itemize}
\end{remark}
\begin{remark}[Notation for fluxes]\label{rem:not} In what follows, with a slight abuse of notation, we will identify the flux through a plaquette with the product of the spins around the same plaquette. That is, we shall think $\Phi = (\Phi_{\Lambda})_{\Lambda \in P(\Gamma_{L})}$ with $\Phi_{\Lambda} = \pm 1$. The value $1$ corresponds to $0$-flux, while the value $-1$ corresponds to the $\pi$-flux. Similarly, we shall denote by $a,b$ the product of the spins along two inequivalent non-contractible loops: $a,b = 1, -1$ correspond to flux $0$ or $\pi$ through the corresponding loops, respectively. 
\end{remark}
\begin{remark}[Non-magnetic observables.]\label{rem:nonma} In Proposition \ref{prp:reso}, we apparently considered only a subclass of physical observables in the sense of Definition \ref{def:gi}. It turns out however that, in the absence of electric field operators in the Hamiltonian ($\varepsilon = 0$), these exhaust all the physical observables with non-zero expectation values. In other words, we claim that the expectation value of a physical observable depending on the $X$ fields is either zero, or it can be represented as the expectation value of a magnetic observable. 

To see this, consider a monomial $\mathcal{O}$ in the $X$ operators, with coefficient taking values in the algebra generated by $Z, a^{\pm}$ (the general claim follows by linearity). Let us denote by $A \subset E(\Gamma_{L})$ the set of bonds occupied by the $X$ operators appearing in the monomial.  Consider:
\begin{equation}\label{eq:X}
\Big\langle {\bm \sigma} \Big| \Big[\prod_{(x,\mu) \in A} X_{x,\mu}\Big] \prod_{x\in \Gamma_{L}} \left( \frac{1 + A_{x} (-1)^{n_{x}}}{2} \right) \Big| {\bm \sigma} \Big\rangle\;.
\end{equation}
We claim that, in order for (\ref{eq:X}) to be non-zero, 
\begin{equation}\label{eq:ass}
\prod_{(x,\mu) \in A} X_{x,\mu} = \prod_{x\in B} A_{x}\;,
\end{equation}
for some $B\subset \Gamma_{L}$. To see this, we proceed as follows. The scalar product (\ref{eq:X}) can be expanded into a sum of terms proportional to, for arbitrary $C\subseteq \Gamma_{L}$:
\begin{equation}\label{eq:C}
\Big\langle {\bm \sigma}\Big| \Big[\prod_{(x,\mu) \in A} X_{x,\mu}\Big] \Big[ \prod_{x\in C} A_{x}\Big] \Big| {\bm \sigma} \rangle\;. 
\end{equation}
In order for this quantity to be non-zero, the argument of the scalar product in (\ref{eq:C}) has to be the identity, due to the fact that $X$ swaps eigenstates of $Z$: every bond in $E(\Gamma_{L})$ has to be occupied by either zero or by an even number of $X$ operators. Thus,
\begin{equation}
\Big[\prod_{(x,\mu) \in A} X_{x,\mu}\Big] \Big[ \prod_{x\in C} A_{x}\Big]  = \mathbbm{1}\Rightarrow \Big[\prod_{(x,\mu) \in A} X_{x,\mu}\Big] = \Big[ \prod_{x\in C} A_{x}\Big]\;,
\end{equation}
where we used that $A_{x}^{2} = \mathbbm{1}$. This proves (\ref{eq:ass}). Next, if (\ref{eq:ass}) holds true, we can use the  projection over states satisfying the Gauss' law to write:
\begin{equation}\label{eq:gaussX}
\Big[ \prod_{x\in B} A_{x}\Big] \Big[\prod_{x \in V(\Gamma)} \left(\frac{1+Q_{x}}{2}\right)\Big] = \prod_{x\in B} (-1)^{n_{x}} \Big[ \prod_{x \in V(\Gamma)} \left(\frac{1+Q_{x}}{2}\right)\Big] \;;
\end{equation}
thus, Eq. (\ref{eq:gaussX}) shows that the expectation value of $\mathcal{O}$ can be rewritten as the expectation value of a magnetic gauge invariant operator, after replacing all $A_{x}$ operators in the observable $\mathcal{O}$ by the corresponding fermionic operators $(-1)^{n_{x}}$.
\end{remark}
\begin{remark}\label{rem:odd}
Let:
\begin{equation}\label{eq:Zpm}
Z_{\beta, \mu, \Gamma}^+ ({\bm \sigma}) := \Tr_{\mathcal{F}_{L}} e^{-\beta (H(\bm{\sigma}) - \mu N)}\;,\qquad Z_{\beta, \mu, \Gamma}^- ({\bm \sigma}) := \Tr_{\mathcal{F}_{L}} (-1)^{N} e^{-\beta (H(\bm{\sigma}) - \mu N)}\;.
\end{equation}
We shall call these objects the even and the odd partition functions. Observe that the odd partition function can be viewed as a partition function for a modified Gibbs state, with complex chemical potential $\mu + i\pi/\beta$. The presence of this imaginary shift in the chemical potential introduces an odd version of the KMS identity. Let:
\begin{equation}
\gamma_{t}(\mathcal{O}) := e^{t(H - \mu N)} \mathcal{O} e^{-t(H-\mu N)}. 
\end{equation}
be the imaginary-time evolution of $\mathcal{O}$, for $t \in [0;\beta)$. Then:
\begin{equation}
\begin{split}
&\frac{\Tr_{\mathcal{F}_{L}} (-1)^{N} \gamma_{t}(a^{-}_{x}) a^{+}_{y} e^{-\beta (H(\bm{\sigma}) - \mu N)} }{Z_{\beta, \mu, \Gamma}^-} \\
&\qquad = \frac{\Tr_{\mathcal{F}_{L}} \gamma_{t}(a^{-}_{x}) (-1)^{N - 1}  a^{+}_{y} e^{-\beta (H(\bm{\sigma}) - \mu N)} }{Z_{\beta, \mu, \Gamma}^-} \\
&\qquad= -\frac{\Tr_{\mathcal{F}_{L}} (-1)^{N}  a^{+}_{y} \gamma_{t-\beta}(a^{-}_{x}) e^{-\beta (H(\bm{\sigma}) - \mu N)} }{Z_{\beta, \mu, \Gamma}^-}\;.
\end{split}
\end{equation}
This relation will be used, later on, to introduce periodic or antiperiodic time extension of the time-ordered correlation functions, for the modified Gibbs states.
\end{remark}
\subsection{Results}
\subsubsection{Stability of the $\pi$-flux phase}
The next theorem is our main result. From now on, we shall choose $\mu = 0$, corresponding to the half-filling condition. We shall also use the short-hand notations:
\begin{equation}
\langle \cdot \rangle_{\beta,L} \equiv \langle \cdot \rangle_{\beta,0,L}\;,\qquad Z_{\beta,L} \equiv Z_{\beta, 0, L}\;.
\end{equation}
We shall denote by $\mathcal{O}$ a general, magnetic gauge-invariant observable. We say that a configuration ${\bm \sigma}$ is associated with the $\pi$-flux phase if, for all $\Lambda\in P(\Gamma_{L})$, we have $\prod_{(x,\mu) \in \partial \Lambda} \sigma_{x,\mu} = -1$: that is, the magnetic flux piercing the plaquette is equal to $\pi$. Physically, we can think that a magnetic monopole sits in the middle of each lattice plaquelle. Also, we shall denote by $\mathcal{O}(\pi, a,b)$ the observable evaluated on a representative configuration realizing the $\pi$-flux phase, with fluxes $(a,b) \in \mathbb{Z}_{2}^{2}$ in correspondence with two inequivalent non-contractible loops of the torus; $a$ is the flux associated with the circle $x_{2} = 0$, while $b$ is the flux associated with the circle $x_{1} = 0$. The specific choice of the representative of the $\pi$-flux phase is not important: by Proposition \ref{prp:gaugeequiv}, changing the representative amounts to a unitary conjugation of the observable, and it will be irrelevant in what follows.
\begin{theorem}[Stability of the $\pi$-flux phase]\label{thm:main} Let $L = 4\ell$, $\ell \in \mathbb{N}$. There exists $t_{0}>0$ independent of $\beta$ and $L$ such that for $t>t_{0}$ the following is true. Let $f_{\beta,L}$ be the free energy of the gauge theory, and let $f_{\beta,L}(\pi)$ the free energy of the $\pi$-flux phase,
\begin{equation}
f_{\beta,L} := -\frac{1}{\beta L^{2}} \log Z_{\beta,L}\;,\qquad f_{\beta,L}(\pi) := -\frac{1}{\beta L^{2}} \log \Big( \sum_{a,b} Z_{\beta,L}(\pi,a,b) \Big)\;. 
\end{equation}
Then:
\begin{equation}\label{eq:fen}
0\leq \beta L^{2} (f_{\beta,L}(\pi) - f_{\beta,L}) \leq \log\Big( ( 1 + e^{-\beta (\Delta_{\beta,L} - 1)} )^{L^{2}} + (1 - e^{-\beta (\Delta_{\beta,L} - 1)})^{L^{2}}\Big) + \log 2\;,
\end{equation}
where, for $\beta, L\gg 1$, and for a suitable universal constant $\kappa>0$:
\begin{equation}
\Delta_{\beta,L} = \Delta_{\infty} + o(1)\;,\qquad \Delta_{\infty} = \kappa  t\;.
\end{equation}
Furthermore, for any magnetic gauge invariant observable $\mathcal{O}$, and for $\beta, L$ large enough: 
\begin{equation}\label{eq:main}
\Big| \langle \mathcal{O} \rangle_{\beta, L} - \langle \mathcal{O} \rangle^{\pi}_{\beta, L}\Big|\leq K_{\mathcal{O}} ((1 + e^{- \beta (\Delta_{\beta,L}-1)})^{L^{2}} - 1)
\end{equation}
where the constant $K_{\mathcal{O}}$ only depends on $\max_{\Phi, a, b} \| \mathcal{O}(\Phi, a, b) \|$ and:
\begin{equation}\label{eq:main2}
\langle \mathcal{O} \rangle^{\pi}_{\beta, L} := \frac{\sum_{(a,b) \in \mathbb{Z}_{2}^{2}} \Tr_{\mathcal{F}_{L}} e^{-\beta H_{f}(\pi, a,b)} \mathcal{O}(\pi, a,b) (1 + (-1)^{N})}{\sum_{(a,b) \in \mathbb{Z}_{2}^{2}}\Tr_{\mathcal{F}_{L}} e^{-\beta H_{f}(\pi, a,b)} (1 + (-1)^{N})}\;.
\end{equation}
Also, let:
\begin{equation}
E_{0;L}(\pi, a, b) := - \lim_{\beta \to \infty} \frac{1}{\beta} \log Z_{\beta,L}(\pi, a, b)
\end{equation}
be the ground state energy of $H_{f}(\pi, a,b)$. Then,
\begin{equation}
E_{0;L}(\pi, a, b) \geq E_{0;L}(\pi, -1, -1)\;.
\end{equation} 
Furthermore, the ground state of the $\pi$-flux phase is almost degenerate, in the sense that, for any $0<\alpha<1/3$:
\begin{equation}\label{eq:deg}
\Big| E_{0;L}(\pi, a, b) - E_{0;L}(\pi, -1, -1) \Big| \leq \frac{C_{\alpha}}{L^{1 - 3\alpha}}\;.
\end{equation}
\end{theorem}
The proof of Theorem \ref{thm:main} is based on reflection positivity methods. It relies on the key result of Lieb \cite{Lieb}, that proved that the optimum, energy minimizing magnetic flux for a half-filled band of electrons hopping on a planar, bipartite graph is $\pi$ per square plaquette. A few remarks are in order.
\begin{remark}\label{rem:main}
\leavevmode
\begin{itemize}
\item[(i)] The quantity $\Delta_{\beta,L}$ can be thought as a lower bound on the monopole's mass: it quantifies the energetic penalty due to the removal of one monopole from the system. The number $\kappa$ is given by the integral of a suitable trigonometric function, that will appear in the proof of the theorem. Its numerical evaluation with the software {\tt Mathematica} gives $\kappa \simeq 0.181$.
\item[(ii)] The error terms in (\ref{eq:main}) are exponentially small for $\beta$ large, however non uniformly in $L$. Thus, the result allows to compare the Gibbs states of the lattice gauge theory and of the $\pi$-flux phase for $L\gg 1$ and $\beta \gtrsim \log L$.
\item[(iii)] The result of Lieb \cite{Lieb} implies that the smallest energy of $H_{\text{f}}({\bm \sigma})$ is attained in correspondence of configurations realizing the $\pi$-flux phase. The result of \cite{Lieb} is actually much more general than this, since it does not require the background gauge fields to be $\mathbb{Z}_{2}$ valued, and it also applied to certain non-abelian gauge groups. See also \cite{NM} for a further generalization of \cite{Lieb}. Here we prove that the ground state energy of Hamiltonians $H_{\text{f}}({\bm \sigma})$ not associated with the $\pi$-flux phase are separated by the ground state energy of the $\pi$-flux phase by a positive gap. Furthermore, we provide a lower bound for this gap, that increases with the number of monopoles' removals, that is plaquettes with zero flux. The proof is based on reflection positivity, as \cite{Lieb}. More precisely, to prove the bound for the gap, we rely on the chessboard estimate, see \cite{Tasaki} for a pedagogical review. The chessboard estimate allows to bound the gap from below, in terms of the ground state energy of a ``chessboard configuration'' of plaquettes with $0$ and $\pi$ fluxes, which we explicitly compute.
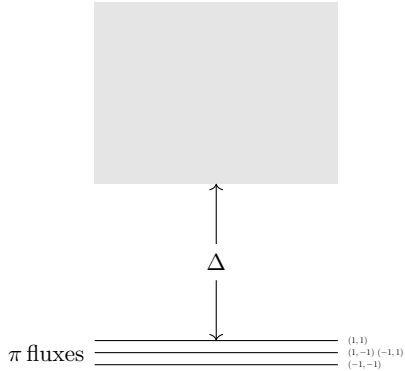
\begin{figure}[H]
\centering
\begin{tikzpicture}[scale=0.8]
\draw[black] (0,0) -- (4,0);
\draw[black] (0,0.2) -- (4,0.2);
\draw[black] (0,0.4) -- (4,0.4);
\fill [gray, opacity=0.2] (0,3) rectangle (4,6);
\draw[black, ->]  (2,1.4) -- (2,0.4);
\draw[black, ->]  (2,2) -- (2,3);
\node[left] (e) at (0,0.2) {\scalebox{0.8}{$\pi\,\text{fluxes}$}};
\node[] (d) at (2,1.7) {\scalebox{0.8}{$\Delta$}};
\node[right] (a) at (4,0.4) {\scalebox{0.3}{$(1,1)$}};
\node[right] (b) at (4,0.2) {\scalebox{0.3}{$(1,-1)\,\,(-1,1)$}};
\node[right] (c) at (4,0) {\scalebox{0.3}{$(-1,-1)$}};
\end{tikzpicture}
\caption{Schematical representation of the ground state energies of the fermionic Hamiltonians $H_{f}({\bm \sigma})$, for $L = 4\ell$. The grey box contains the ground state energies of $H_{f}({\bm \sigma})$ with at least one zero-flux plaquette. A numerical computation shows that $E_{0;L}(\pi, -1, -1) < E_{0;L}(\pi, 1, -1) < E_{0;L}(\pi, 1, 1)$ (and $E_{0;L}(\pi, 1, -1) = E_{0;L}(\pi, -1, 1)$). The difference of the $\pi$-flux ground state energies vanishes as $L\to \infty$.}
\end{figure}
\item[(iv)] The fermionic ground state energy is minimized by the $\pi$-flux phase also for $L = 2(2\ell + 1)$. In this case, however, the optimal choice of $a,b$ is $(a,b) = (1,1)$. The restriction to $L = 4\ell$ is technical, and it comes from the argument used to prove a lower bound for the monopole's mass.
\end{itemize}
\end{remark}
\begin{remark}[On the $\pi$-flux phase]\label{rem:pf}
There are only four inequivalent $\pi$ flux phases, labelled by the fluxes $(a,b)$ corresponding to non-contractible loops of the torus. As discussed in Section \ref{sec:piflux}, for the four $\pi$-flux phases it is possible to choose gauge field configurations ${\bm \sigma}$ that preserve translation invariance, paying the price of considering a fundamental cell formed by two lattice sites; we refer the reader to Figs. \ref{fig:pi1}, \ref{fig:pifluxes} for graphical representations. To be more precise, the $\pi$-flux phases can be described in terms of translation-invariant fermionic Hamiltonians, with periodic/antiperiodic boundary conditions, according to the value of $(a,b)$. The choice $a = 1$ corresponds to periodic boundary condition on the horizontal direction, while $a=-1$ correspond to antiperiodic boundary conditions; the same holds for $b$, for the vertical direction. 
\end{remark}
\begin{remark}[Decay of fermionic correlations] As the proof of Theorem \ref{thm:main} shows, all Euclidean (that is, imaginary time) ground state correlation functions of the gauge theory can be computed explicitly, via the fermionic Wick's rule. In particular, since the fermionic Hamiltonian with the $\pi$-flux background has no spectral gap, it turns out that Euclidean correlation functions decay algebraically, at a non-integrable rate in space-time. Let us denote by $x$ the location of the fundamental cell of the $\pi$-flux phase, with elements labelled by $i = A,B$; see Fig. \ref{fig:pi1}. Let us denote by $a^{\pm}_{x,i}$ the fermionic creation/annihilation operators, associated with the creation/annihilation of a particle in the fundamental cell labelled by $x$, and occupying the lattice site $i= A,B$ in the cell. Then, we have, for $|x-y| \neq 0$ and $t>0$:
\begin{equation}\label{eq:dd}
\lim_{L\to \infty}\lim_{\beta \to \infty}\langle e^{tH} a^{+}_{x,i}a^{-}_{x,i} e^{-tH}\;; a^{+}_{y,j}a^{-}_{y,j} \rangle_{\beta, L} = -\frak{g}_{ij}((x,t);(y,0)) \frak{g}_{ji}((y,0);(x,t))\;,
\end{equation}
where $\frak{g}(\cdot ; \cdot)$ is:
\begin{equation}
\frak{g}((x,t);(y,s)) = \int_{[0,2\pi] \times [0,\pi]} \frac{dk}{(2\pi)^{2}} \int_{\mathbb{R}} \frac{d\omega}{2\pi}\, \frac{e^{i \omega (t-s) + i k (x-y)}}{-i\omega + h(k)}\;,
\end{equation}
with Bloch Hamiltonian:
\begin{equation}
h(k) = -t \begin{pmatrix} - e^{ik_{1}} - e^{-ik_{1}} & 1 + e^{2ik_{2}}  \\ 1 + e^{-2ik_{2}} & e^{ik_{1}} + e^{-ik_{1}} \end{pmatrix}\;.
\end{equation}
The left-hand side of (\ref{eq:dd}) is the Euclidean and time-ordered density-density correlation function, whose decay properties are ultimately determined by the spectrum of $h(k)$. The eigenvalues of $h(k)$, also called energy bands, are given by:
\begin{equation}
e_{\pm}(k) = \pm 2t \sqrt{1 + \frac{1}{2}\cos(2k_{1}) + \frac{1}{2} \cos(2k_{2})}\;.
\end{equation}
\begin{figure}
\centering
\includegraphics[trim={0cm 3cm 15cm 3cm}, clip, width=.7\textwidth]{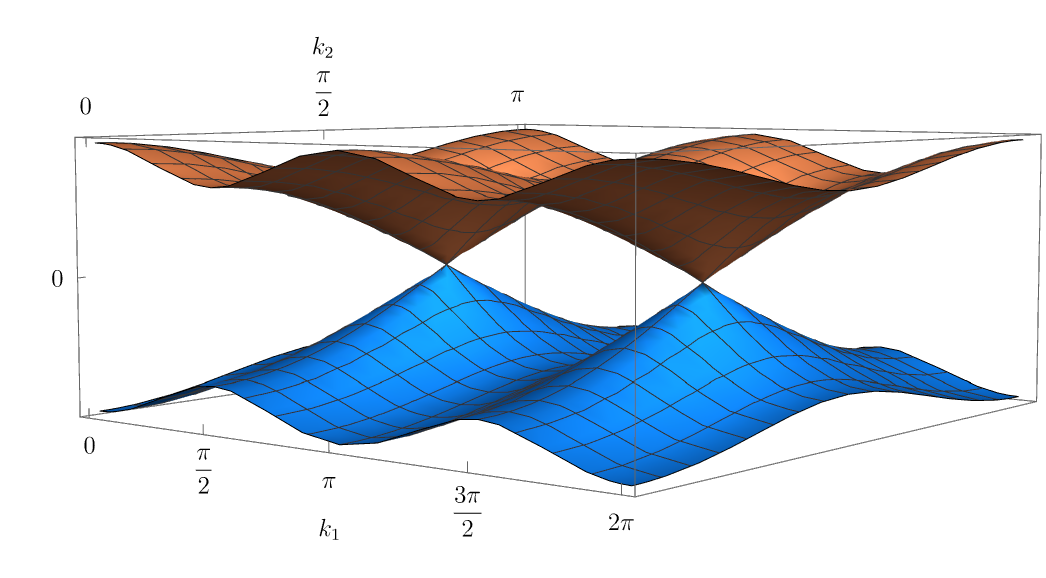}
\caption{Energy bands associated with the $\pi$-flux phase.}
    \label{fig:disp}
\end{figure}
In particular, $e_{\pm}(k) = 0$ if and only if $k  = (\pi/2, \pi/2) =: k^{-}_{F}$ or $k = (3\pi/2, \pi/2) =: k_{F}^{+}$. In correspondence with these points, the energy bands display conical intersections, see Fig. \ref{fig:disp}: for $\alpha = \pm$,
\begin{equation}
e_{\pm}(k^{\alpha}_{F} + q) = q\cdot \nabla_{q} e_{\pm}(k_{F}) + O(|q|^{2}) = \mp 2t |q| + O(|q|^{2})\;,\qquad \text{for $|q| \ll 1$.}
\end{equation}
This implies that the function $\frak{g}({\bm x};{\bm y})$, with ${\bm x} = (x,t)$ and ${\bm y} = (y,s)$, decays as $|{\bm x} - {\bm y}|^{-2}$, and hence that the ground-state density-density correlation function (\ref{eq:dd}) decays as $|{\bm x} - {\bm y}|^{-4}$, which is non-integrable in $2+1$ dimensions.
\end{remark}
\subsubsection{Susceptibility}
Theorem \ref{thm:main} opens the way to explicitly compute physically relevant quantities for the lattice gauge theory. In particular, it allows to compute the magnetic susceptibility. This quantity is used to probe the emergence of a semimetallic phase, and it has been investigated numerically in \cite{gazit}. Let us recall its definition.

We shall couple the system to an external, time-independent vector potential, which generates a static magnetic field. The vector potential only couples to the fermionic Hamiltonian, and it does so via the Peierls' substitution. Denoting by $H_{f}$ the Fock space Hamiltonian associated with a single-particle, nearest-neighbour Schroedinger operator $h$, we define:
\begin{equation}\label{eq:peierls}
H_f (A) := \sum_{x,y \in \Gamma_{L}} a^{+}_{x} h(x;y) e^{i q \int_{x\to y} d\ell \cdot A(\ell)} a^{-}_{y} + \text{h.c.}\;,
\end{equation}
with $q$ the electric charge, and where the integral is over the bond connecting $x$ to $y$. The external potential is defined in the continuum, over $\mathbb{R}^{2} / L \mathbb{R}^{2}$. 

By Theorem \ref{thm:main}, we know that the ground state of the lattice gauge theory is represented by a $\pi$-flux phase; this corresponds to a periodic arrangement of spins on the lattice bonds, as anticipated in Remark \ref{rem:main} and as discussed later in Section \ref{sec:piflux}. Thus, we find it convenient to rewrite the Hamiltonian $H_{f}(A)$ using the fundamental cell that accommodates for the periodicity of the $\pi$-flux configurations, see Fig. \ref{fig:pi1}. Let $\Gamma_{L}^{\text{red}}$ be the sublattice of $\Gamma_{L}$ generated by the basis vectors $e_{1}$ and $2 e_{2}$. We have:
\begin{equation}\label{hamipei}
    H_{f}(A) = \sum_{\substack{x, x' \in \Gamma_{L}^{\text{red}}}}\sum_{i,j\in I}  h_{ij}(x,x') e^{i q \int_0^1 A((x'+r_j)(1-s)+s(x+r_i)) \cdot \xi ds} a^+_{x,i} a^{-}_{x', j}
\end{equation}
where $I = (A, B)$ collects the labels in the fundamental cell, and $h_{ij}(x,x')$ are the matrix elements of the single-particle Hamiltonian, labelled by the points $x,x'$ of the reduced lattice $\Gamma_{L}^{\text{red}}$ and by the points in the fundamental cell. We denote by $r_{A} = e_{2}$ and $r_{B} = 0$ the relative coordinates in the fundamental cell. The displacement vector $\xi$ is directed along the bonds of the lattice, and it is $\xi = x + r_{i} - x' - r_{j}$ (and it can take values $\pm e_{1}, \pm e_{2}$). Next, we represent the vector potential in terms of its Fourier coefficients,
\begin{equation}
A_{\mu}(x+r_i) = \frac{1}{L^{2}} \sum_{p \in \frac{2\pi}{L} (\mathbb{Z}^{2} / L \mathbb{Z}^{2})} \hat{A}_{\mu}(p) e^{i p \cdot (x+r_i)}\;,\qquad \mu =1,2,\qquad x\in \Gamma_{L}^{\text{red}}\;,
\end{equation}
where here and below the sum runs over the independent momenta compatible with $L$-periodicity in space. The momentum-space current operator is then:
\begin{equation}
\hat J^{(A)}_{\mu}(p) := - L^{2}  \frac{\partial H_{f}(A)}{\partial \hat A_{\mu}(-p)}\;;
\end{equation}
by expading the Hamiltonian (\ref{hamipei}) in $A$, we get:
\begin{equation}
\begin{split}
    H_{f}(A) &= H_{f} - \frac{1}{L^{2}} \sum_{p} \sum_{\mu = 1,2} \hat J_{\mu}(p) \hat A_\mu(-p)\\&\quad  - \frac{1}{2}  \frac{1}{L^{4}} \sum_{p, q} \sum_{\mu,\nu = 1,2} \hat A_\mu(-p) \hat K_{\mu,\nu}(p, q) \hat A_\nu(-q) + O(A^3)\;,
    \end{split}
\end{equation}
and hence:
\begin{equation}
     \hat J_\mu^{(A)}(p) =   \hat J_\mu(p) + \frac{1}{L^{2}}\sum_{q} \sum_{\mu,\nu=1,2} \hat K_{\mu,\nu}(p, q) \hat A_\nu(-q) + O(A^{2})\;,
\end{equation}
where:
\begin{equation}\label{eq:JK}
\begin{split}
&\hat J_\mu(p) \\
&:=-\frac{i q}{2} \sum_{\substack{x, x' \in \Gamma_{L}^{\text{red}}\\i,j \in I}} e^{-i p \cdot (x'+ r_j) } \Big(h_{ij}(x,x') a^+_{x,i} a^{-}_{x', j} -h_{ji}(x',x) a^{+}_{x',j} a^{-}_{x, i} \Big)  \eta_{\xi}(p)  \xi_{\mu}\\
&\hat K_{\mu,\nu}(p,q) \\
&:= -\frac{q^2}{2}  \sum_{\substack{x, x' \in \Gamma_{L}^{\text{red}}\\i,j \in I}} e^{-i (p + q) \cdot (x'+ r_j) } \Big(h_{ij}(x,x') a^{+}_{x,i} a^{-}_{x', j} +h_{ji}(x',x) a^{+}_{x',j} a^{-}_{x, i} \Big)   \eta_{\xi}(p)  \eta_{\xi}(q) \xi_{\mu}\xi_{\nu}
\end{split}
\end{equation}
with:
\begin{equation}\label{eq:etadef}
\eta_{\xi}(p) = \left\{ \begin{array}{cc} \frac{e^{-i p \cdot \xi}-1}{-i p \cdot \xi} & p \cdot \xi \neq 0 \\ 1 & \text{otherwise.} \end{array} \right.
\end{equation}
We are now ready to compute the average of the current operator in presence of the external gauge field, at first order in the perturbation. Let $H(A)$ be the total Hamiltonian of the gauge theory, coupled to the external gauge field. We have, by Duhamel formula:
\begin{equation}\label{eq1}
\begin{split}
&\frac{\Tr(e^{-\beta H(A)} \hat J^{(A)}_1(p))}{\Tr(e^{-\beta H(A)})} - \frac{\Tr(e^{-\beta H} \hat J_1(p))}{\Tr(e^{-\beta H})} \\
&\quad = \frac{1}{L^{2}} \sum_{\nu =1,2}\sum_{q}\bigg[\int_0^{\beta} ds \left \langle \hat J_{\nu}(q,s); \hat J_1(p) \right \rangle_{\beta, L} + \left \langle \hat K_{1,\nu}(p, q) \right \rangle_{\beta, L} \bigg] \hat A_{\nu}(q) + O(A^{2})\;,
\end{split}
\end{equation}
where $\mathcal{O}(s) = e^{s H} \mathcal{O} e^{-s H}$ is the imaginary-time evolution of the operator $\mathcal{O}$. By Theorem \ref{thm:main}, as $\beta \to \infty$ the Gibbs state average reduces to the linear combination of averages on the four $\pi$-flux phases. 
%
%
Since all $\pi$-flux configurations are translation-invariant, we can use momentum conservation to write:
\begin{equation}\label{eq:susclim}
\begin{aligned}
&\lim_{\beta \to \infty} 
\frac{1}{L^{2}} \sum_{q}\bigg[\int_0^{\beta} ds \left \langle \hat J_{\nu}(q,s); \hat J_1(p) \right \rangle_{\beta, L} + \left \langle \hat K_{1,\nu}(p, q) \right \rangle_{\beta, L} \bigg] \hat A_{\nu}(q) \\
&\quad = \lim_{\beta \to \infty}  \frac{1}{L^{2}} \sum_{q} \bigg[ \int_0^{\beta} ds \left \langle  \hat J_{\nu}(-p,s); \hat J_1(p) \right \rangle_{\beta, L} + \left \langle \hat K_{1,\nu}(p, -p) \right \rangle_{\beta, L} \bigg] \delta_{q,-p} \hat A_{\nu}(q)\;, 
\end{aligned}
\end{equation}
with $\delta_{\cdot,\cdot}$ the periodic Kronecker symbol (of period $\pi$ in both momentum components). Let us consider a vector potential $\hat A(p) = (\hat A_{1}(p), 0)$, with $\hat A_{1}(p)$ nonzero for $p = (0, p_{2})$. By Maxwell's laws, the associated magnetic field is $\hat B(p) = -ip_{2} \hat A_{1}(p) + ip_{1} \hat A_{2}(p)$. We define the magnetic susceptibility from the variation of (\ref{eq:susclim}) at first order in the.vector potential. That is:
\begin{equation}\label{eq:chifin}
\chi(p_{2}) := \frac{1}{p_{2}^{2}}\lim_{L\to \infty} \lim_{\beta \to \infty} \frac{1}{L^{2}}\bigg[ \int_0^{\beta} ds \left \langle  \hat J_1(-p,s); \hat J_1(p) \right \rangle_{\beta, L} + \left \langle \hat K_{1,1}(p, -p) \right \rangle_{\beta, L} \bigg]\;. 
\end{equation}
The factor $1/p_{2}^{2}$ takes into account the fact that we are interested in the variation with respect of the magnetic field $\hat B((0,p_{2})) = -ip_{2} \hat A_{1}((0,p_{2}))$, and the fact that we are actually measuring the variation of the magnetization: for a (classical) conserved current $\hat J(p)$, we can write $\hat J_{1}(p) = ip_{2} \hat M(p)$ and $\hat J_{2}(p) = -ip_{1} \hat M(p)$, where $\hat M(p)$ is the Fourier transform of the local magnetic moment \cite{gazit, Ando}. The next result provides the explicit form of the magnetic susceptibility, at small external momenta.
\begin{proposition}[Ground state susceptibility]\label{prp:susc} Under the same assumption of Theorem \ref{thm:main}, the following is true:
\begin{equation}\label{eq:susc}
\chi(p_{2}) = -\frac{q^{2} v}{8 |p_{2}|} + o\bigg( \frac{1}{|p_{2}|} \bigg)\;, 
\end{equation}
where $v = 2t$ is the Fermi velocity of the fermionic sector in the $\pi$-flux phase.
\end{proposition}
\begin{remark} The expression (\ref{eq:susc}) is the same obtained for non-interacting $2d$ massless Dirac fermions and for graphene \cite{Ando}. The result (\ref{eq:susc}) puts on rigorous grounds the numerical evidence of \cite{gazit}, for the $\mathbb{Z}_{2}$ lattice gauge theory. 
\end{remark}
The proof of Proposition \ref{prp:susc} follows from an explicit calculation of the transport coefficient on the $\pi$-flux phase, which correspond to a semimetallic system. The final expression for the susceptibility is completely determined by the Dirac-like excitation spectrum of the $\pi$-flux phase, and lattice conservation laws.

The rest of the paper is devoted to the proof of Theorem \ref{thm:main} and of Proposition \ref{prp:susc}.

\section{Reflection positivity and chessboard estimates}\label{sec:RP}
In this section we shall introduce the techniques needed in order to prove our main result, Theorem \ref{thm:main}. The key tool is the notion of reflection positivity for lattice fermionic systems. Specifically, we shall adapt the argument of Lieb \cite{Lieb} for the proof of the $\pi$-flux conjecture, to our case. Then, in order to prove the stability of the $\pi$-flux phase in presence of dynamical gauge fields, we will have to prove that the energetic contribution of $0$-fluxes in the square lattice is extensive in the number of such fluxes. To achieve this, we shall use the chessboard estimate (see \cite{Tasaki} for a review) to show that the energetic insertion of $0$-fluxes, or monopoles' removals, can be bounded below proportionally to the energy of a suitable periodic configuration of fluxes, which can be explicitly evaluated.

\subsection{Reflection positivity and the $\pi$-flux phase}\label{sec:RPpi}

Let us denote by $P$ a hyperplane cutting perpendicularly the torus $\Gamma_{L}$ in two halves. The cut can be represented on a plane as in Fig. \ref{fig:cut}.

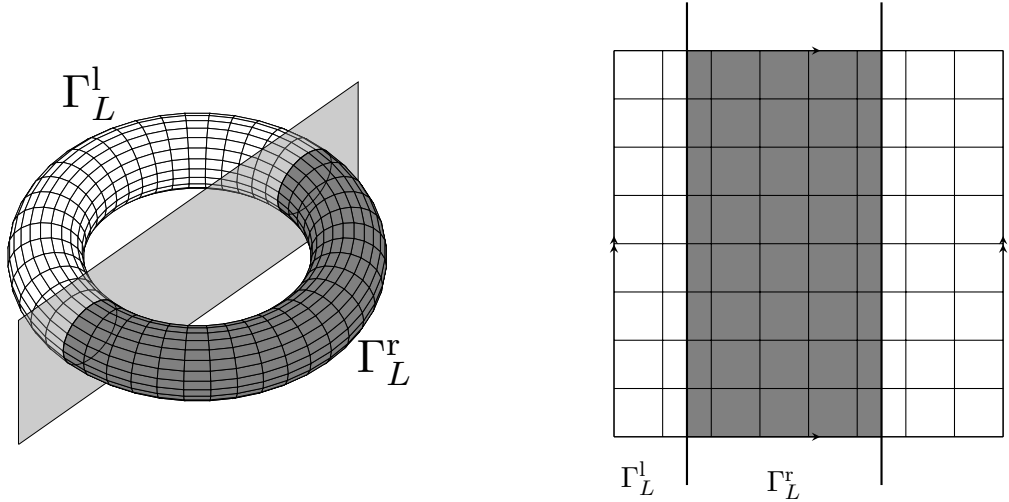
\begin{figure}
\begin{tikzpicture}[scale=1.6]
\begin{axis}[axis equal image,
        axis line style={draw=none},
        tick style={draw=none},
        xmax=18,ymax=22,zmax=7,xmin=-15,ymin=-20,zmin=-7,
        ticks=none,
        clip bounding box=upper bound,
        colormap/blackwhite,
        samples=23,
        view={45}{45}]
        \begin{pgfonlayer}{background layer}
            \addplot3[domain=0:360,
                    y domain=90:270,
                    surf,
                    z buffer=sort,
                    color = white!30,
                    mesh/ordering=y varies,
                    shader=flat, 
                    draw=black,
                    line width=0.1pt
                ]
                ({(12 + 3 * cos(x)) * cos(y)} ,
                {(12 + 3 * cos(x)) * sin(y)},
                {3 * sin(x)});
        \end{pgfonlayer} 
        \begin{pgfonlayer}{main}

            \draw[fill=gray, fill opacity=0.4, draw=black, line width=0.1pt] (0,18,-7) -- (0,18,7) -- (0,-20,7)  -- (0,-20,-7) -- (0,18,-7);
        \end{pgfonlayer}

        \begin{pgfonlayer}{foreground layer}
            \addplot3[domain=0:360,y 
                domain=-90:90,
                surf,
                z buffer=sort,
                color=gray, 
                mesh/ordering=y varies,
                shader=flat, 
                draw=black,
                line width=0.1pt]
                ({(12 + 3 * cos(x)) * cos(y)} ,
                {(12 + 3 * cos(x)) * sin(y)},
                {3 * sin(x)});
            \addplot3[line width = 0.1pt,domain=0:360,surf, z buffer=sort, black,samples=71, samples y=1]
                (0,
                {(12 + 3 * cos(x)) * sin(-90)},
                {3 * sin(x)});
            \addplot3[line width=0.1pt,domain=0:360,surf, z buffer=sort, black,samples=71, samples y=1]
                (0,
                {(12 + 3 * cos(x)) * sin(90)},
                {3 * sin(x)});

            \node[] at (21,0,3) {\scalebox{1}{$\Gamma_L^\text{r}$}};
            \node[] at (-12,0,10) {\scalebox{1}{$\Gamma_L^\text{l}$}};
        \end{pgfonlayer}
    \end{axis}
\begin{scope}[xshift =6cm,yshift=1cm, scale =0.4]
\draw (0,0) grid (8,8);
\draw[->-] (0,8) -- (8,8);
\draw[->-] (0,0) -- (8,0);
\draw[->>-] (0,0) -- (0,8);
\draw[->>-] (8,0) -- (8,8);
\draw[thick] (1.5,-1) -- (1.5, 9);
\draw[thick] (5.5,-1) -- (5.5, 9);
 \draw[fill=black, fill opacity=0.5, draw=black] (1.5,0) -- (5.5,0) -- (5.5,8)  -- (1.5,8) -- (1.5,0);
\node[below] (a) at (3.5,-0.5) {\scalebox{1}{$\Gamma_L^{\text{r}}$}};
\node[below] (b) at (0.5,-0.3) {\scalebox{1}{$\Gamma_L^\text{l}$}};
\end{scope}
\end{tikzpicture}
 \caption{Graphical representation of the cut torus.}
    \label{fig:cut}
\end{figure}

Let us denote by $\Gamma^{\text{l}}_{L}$ and $\Gamma_{L}^{\text{r}}$ the left and right portion of the lattice $\Gamma_{L}$, with respect to the cut introduced by the hyperplane. Correspondingly, we rewrite the matter Hamiltonian as:
\begin{equation}\label{eq:matterLR}
H_{f}(\bm{\sigma}) = H^{\text{r}}_{f}(\bm{\sigma}) + H_{f}^{\text{l}}(\bm{\sigma}) + V_{\text{int}}(\bm{\sigma})\;,
\end{equation}
where $H_{f}^{\sharp}$ collects fermionic monomials that are parametrized by space coordinates fully contained in $\Gamma_{L}^{\sharp}$, with $\sharp = \text{l},{\text{r}}$, while $V_{\text{int}}$ takes into account the hopping terms that connect the two halves of the cut torus. 

Without loss of generality, we can suppose that the sign of the hoppings of $V_{\text{int}}(\bm{\sigma})$ is $-1$, that is $\sigma_{x,1} = 1$ for all bonds intersecting the cut. In fact, suppose that $\sigma_{x,1} = -1$. Consider the ``star'', formed by the lattice bonds centered at $x$. Changing signs to all gauge fields associated with bonds touching $x$ does not change the fluxes enclosed by the plaquettes and by the non-contractible cycles, and hence produces a gauge equivalent configuration $\bm{\sigma}'$. Thus, since gauge equivalent configurations correspond to unitarily equivalent matter Hamiltonians (recall Proposition \ref{prp:gaugeequiv}), from now on we will assume that:
\begin{equation}\label{eq:Vint}
V_{\text{int}}(\bm{\sigma}) = -t \sum_{\substack{x:\\
x\in \Gamma^{\text{l}}_{L},\\ x + e_{1} \in \Gamma_{L}^{\text{r}}}} (a^{+}_{x} a^{-}_{x + e_{1}} + a^{+}_{x+e_{1}} a^{-}_{x})\;
\end{equation}
Next, let us define the operator that implements the reflection across the hyperplane cutting the torus.  To this end, let $r(x)$ be the geometrical reflection of $x$ across the hyperplane $P$. Let $\mathcal{R}$ be the unitary operator implementing the geometric reflection on the fermionic algebra,
\begin{equation}
\mathcal{R}^{*} a^{\pm}_{x} \mathcal{R} = a^{\pm}_{r(x)}\;.
\end{equation}
Let $\tau$ be the unitary operator implementing the particle-hole transformation,
\begin{equation}
\tau^{*} a^{\pm}_{x} \tau = a^{\mp}_{x}\;.
\end{equation}
\begin{definition}[Reflection operator] 
The reflection operator $\Theta$ is the antilinear, unitary operator acting on the fermionic algebra as:
\begin{equation}
\Theta(\mathcal{O}) = \overline{\tau^{*} \mathcal{R}^{*} \mathcal{O} \mathcal{R} \tau}
\end{equation}
where the complex conjugation acts on the coefficients of the fermionic monomials.
\end{definition}
\begin{remark} This is the notion of reflection operator of \cite{Lieb}. In the case of $\mathbb{Z}_{2}$ gauge theory, the complex conjugation could actually be dropped (all terms in the Hamiltonian are real in the sense that the first quantized hamiltonian is real).
\end{remark}

We are now ready to define the notion of reflection symmetric Hamiltonian.
\begin{definition}[Reflection symmetric Hamiltonian] Let $H_{f}(\bm{\sigma})$ be as in (\ref{eq:matterLR}), (\ref{eq:Vint}). The Hamiltonian $H_{f}(\bm{\sigma})$ is called reflection symmetric if:
\begin{equation}\label{eq:symm}
H_{f}^{\textnormal{l}}(\bm{\sigma}) = \Theta(H_{f}^{\textnormal{r}}(\bm{\sigma}))\;.
\end{equation}
\end{definition}
\begin{remark} The identity (\ref{eq:symm}) is obviously not true for all possible gauge field configurations. It can be understood as a symmetry condition on the configuration $\bm{\sigma}$, with respect to the cut induced by $P$.
\end{remark}
The next result states key estimates for the partition functions (\ref{eq:Zpm}) at $\mu = 0$. This choice of chemical potential defines the half-filling condition. It is a direct adaptation of the argument of \cite{Lieb}; the adaptation allows to take into account the presence of the parity operator $(-1)^{N}$ in the definition of the odd partition function. Let $Z^{\pm}_{\beta, L}(H^{\text{l}}, H^{\text{r}})$ be the even and odd partition functions at $\mu = 0$ of the Hamiltonian $H = H^{\text{l}} + H^{\text{r}} + V_{\text{int}}$, with $H^{\text{l}}, H^{\text{r}}$ in the left, resp. right fermionic algebras, and $V_{\text{int}}$ as in (\ref{eq:Vint}).
\begin{lemma}[Lieb]\label{lem:lieb} For any $\beta \geq 0$ and for $L = 2\ell$, the following inequality holds true:
\begin{equation}\label{eq:RP}
Z^{\pm}_{\beta,L}(H^{\textnormal{l}}, H^{\textnormal{r}})^{2} \leq Z^{\pm}_{\beta,L}(H^{\textnormal{l}}, \Theta(H^{\textnormal{l}})) Z^{\pm}_{\beta,L}(H^{\textnormal{r}}, \Theta(H^{\textnormal{r}}))\;. 
\end{equation}
\end{lemma}
\begin{proof} The case with $+$ has been proved by Lieb \cite{Lieb}. Let us repeat the argument for the case with $-$. As in \cite{Lieb}, the starting point is Trotter's formula:
\begin{equation}\label{eq:trotter}
e^{-\beta (H^{\text{l}} + H^{\text{r}}+ V_{\text{int}})} = \lim_{M\to \infty}\left[  \left( 1 - \frac{\beta V_{\text{int}}}{M} \right) e^{-\frac{\beta}{M} H^{\text{l}}} e^{-\frac{\beta}{M} H^{\text{r}} } \right]^{M}\;.
\end{equation}
Observe that, thanks to the gauge choice leading to (\ref{eq:Vint}), the operator $-V_{\text{int}}$ is given by the linear combination of hopping terms of the form $a^{+}_{x} a^{-}_{y}$, $-a^{-}_{x} a^{+}_{y}$ with $x\in \Gamma_{L}^{\text{l}}$ and $y\in \Gamma_{L}^{\text{r}}$ with {\it positive} coefficients. Next, for fixed $M$, we expand the product in (\ref{eq:trotter}); we get a linear combination of expressions of the form:
\begin{equation}
X_{\alpha} = b_{1;\alpha} e^{-\frac{\beta}{M} H^{\text{l}} } e^{-\frac{\beta}{M} H^{\text{r}}} b_{2;\alpha} e^{-\frac{\beta}{M} H^{\text{l}} } e^{-\frac{\beta}{M} H^{\text{r}} } b_{3;\alpha} \cdots b_{M;\alpha} e^{-\frac{\beta}{M} H^{\text{l}} } e^{-\frac{\beta}{M} H^{\text{r}} }\;.
\end{equation}
where $b_{i;\alpha}$ is either the identity, or $(\beta/M) a^{+}_{x} a^{-}_{x+e_{1}}$, or $-(\beta/M) a^{-}_{x} a^{+}_{x+e_{1}}$, with $x\in \Gamma_{L}^{\text{l}}$ and $x+e_{1}\in \Gamma_{L}^{\text{r}}$. Lieb's idea is to move all the left operators to the left of the right operators, without changing the relative orders of the operators in the left or right subalgebras. Observing that $\Tr X_{\alpha} \neq 0$ only if the number of $a^{+}_{l} a^{-}_{r}$ monomials equals the number of $a^{-}_{l} a^{+}_{r}$, by particle number conservation, proceeding as in \cite{Lieb} we see that the overall sign picked in the rearrangement of $X_{\alpha}$ is $+1$. 

Therefore, we succeeded in rewriting $X_{\alpha}$ as $X_{\alpha}^{\text{l}} X_{\alpha}^{\text{r}}$, for suitable monomials $X_{\alpha}^{\sharp}$ in the $\sharp$ subalgebras. In particular,
\begin{equation}
\Tr\, ( (-1)^{N} X_{\alpha} ) = \Tr\, ( (-1)^{N_{\text{l}}} X_{\alpha}^{\text{l}} (-1)^{N_{\text{r}}} X_{\alpha}^{\text{r}} )\;,
\end{equation} 
where we used that the right number operator commutes with all left operators. We can further manipulate the trace, to rewrite it as:
\begin{equation}\label{eq:lr0}
\Tr\, ( (-1)^{N_{\text{l}}} X_{\alpha}^{\text{l}} (-1)^{N_{\text{r}}} X_{\alpha}^{\text{r}} ) = \Big( \Tr_{\text{l}}\, (-1)^{N_{\text{l}}} X_{\alpha}^{\text{l}}\Big) \Big(\Tr_{{\text{r}}} (-1)^{N_{\text{r}}} X_{\alpha}^{\text{r}} \Big)\;,
\end{equation}
where $\Tr_{\sharp}$ denotes the trace over the Fock space associated with $\Gamma_{L}^{\sharp}$. Furthermore,
\begin{equation}\label{eq:lr}
\begin{split}
&\Big( \Tr_{\text{l}}\,  (-1)^{N_{\text{l}}} X_{\alpha}^{\text{l}}\Big) \Big(\Tr_{{\text{r}}} (-1)^{N_{\text{r}}} X_{\alpha}^{\text{r}} )\Big) \\
&\qquad = 2^{-|\Gamma_{L}^{\text{l}}| - |\Gamma_{L}^{\text{r}}|} \Big( \Tr\,  (-1)^{N_{\text{l}}} X_{\alpha}^l\Big) \Big(\Tr\, (-1)^{N_{\text{r}}} X_{\alpha}^{\text{r}} \Big) \\
&\qquad = 2^{-|\Gamma_{L}|} \Big( \Tr\, (-1)^{N_{\text{l}}} X_{\alpha}^{\text{l}}\Big) \Big(\Tr\, (-1)^{N_{\text{r}}} X_{\alpha}^{\text{r}} \Big)\;.
\end{split}
\end{equation}
Therefore, we have:
\begin{equation}\label{eq:bdZ}
\begin{split}
&\Tr\, (-1)^{N} \left[  \left( 1 - \frac{\beta V_{\text{int}}}{M} \right) e^{-\frac{\beta}{M} H^{\text{l}} } e^{-\frac{\beta}{M} H^{\text{r}} } \right]^{M} \\
&\qquad = 2^{-|\Gamma_{L}|} \sum_{\alpha} \Big( \Tr\, (-1)^{N_{\text{l}}} X_{\alpha}^{\text{l}}\Big) \Big(\Tr\, (-1)^{N_{\text{r}}} X_{\alpha}^{\text{r}} \Big) \\
&\qquad \leq 2^{-|\Gamma_{L}|} \Big(\sum_{\alpha} \Big| \Tr\, (-1)^{N_{\text{l}}} X_{\alpha}^{\text{l}} \Big|^{2}\Big)^{1/2} \Big(\sum_{\alpha} \Big| \Tr\, (-1)^{N_{\text{r}}} X_{\alpha}^{\text{r}} \Big|^{2}\Big)^{1/2}\;,
\end{split}
\end{equation}
where the last step follows by Cauchy-Schwarz inequality for sums. Next,
\begin{equation}
\begin{split}
\sum_{\alpha} \Big| \Tr\, (-1)^{N_{\text{l}}} X_{\alpha}^{\text{l}} \Big|^{2} &= \sum_{\alpha} \Big(\Tr\, (-1)^{N_{\text{l}}} X_{\alpha}^{\text{l}}\Big) \Big(\overline{\Tr\, (-1)^{N_{\text{l}}} X_{\alpha}^{\text{l}}} \Big) \\
&= \sum_{\alpha} \Big(\Tr\, (-1)^{N_{\text{l}}} X_{\alpha}^l\Big) \Big(\Tr\, \overline{(-1)^{N_{\text{l}}} X_{\alpha}^{\text{l}}}\Big) \\
&= \sum_{\alpha} \Big(\Tr\, (-1)^{N_{\text{l}}} X_{\alpha}^{\text{l}}\Big) \Big(\Tr\, \Theta((-1)^{N_{\text{l}}} X_{\alpha}^{\text{l}})\Big) \\
& = \sum_{\alpha} \Big(\Tr\, (-1)^{N_{\text{l}}} X_{\alpha}^{\text{l}}\Big) \Big(\Tr\, (-1)^{-N_{\text{r}} + |\Gamma_{L}^{\text{r}}|} \Theta(X_{\alpha}^{\text{l}})\Big)\;,
\end{split}
\end{equation}
where we used that $\Theta(N_{l}) = \sum_{x\in \Gamma_{L}^{\text{r}}} a^{-}_{x} a^{+}_{x} = - N_{\text{r}} + |\Gamma_{L}^{\text{r}}|$. Using that $L = 2\ell$ and hence that $|\Gamma_L^{\text{r}}|$ is even, we obtain, applying (\ref{eq:lr}), (\ref{eq:lr0}) backwards:
\begin{equation}\label{eq:CS}
\begin{split}
\sum_{\alpha} \Big| \Tr\, (-1)^{N} X_{\alpha}^{\text{l}} \Big|^{2} &= 2^{|\Gamma_{L}|} \Tr\, (-1)^{N} X_{\alpha}^{\text{l}} \Theta(X_{\alpha}^{\text{l}})\;.
\end{split}
\end{equation}
Of course, the same identity can be proved for operators in the right algebra. In conclusion, plugging (\ref{eq:CS}) in (\ref{eq:bdZ}) we get:
\begin{equation}
\begin{split}
&\Tr\, (-1)^{N} \left[  \left( 1 - \frac{\beta V_{\text{int}}}{M} \right) e^{-\frac{\beta}{M} H^{\text{l}} } e^{-\frac{\beta}{M} H^{\text{r}} } \right]^{M} \\
&\quad \leq \Big(\sum_{\alpha} \Tr\, (-1)^{N} X_{\alpha}^l \Theta(X_{\alpha}^{\text{l}})\Big)^{1/2} \Big(\sum_{\alpha} \Tr\, (-1)^{N} X_{\alpha}^{\text{r}} \Theta(X_{\alpha}^{\text{r}})\Big)^{1/2} \\
&\quad = \Big(\Tr\, (-1)^{N} \left[  \left( 1 - \frac{\beta V_{\text{int}}}{M} \right) e^{-\frac{\beta}{M} H^{\text{l}} } e^{-\frac{\beta}{M} \Theta(H^{\text{l}}) } \right]^{M}\Big)^{1/2} \\
&\qquad \cdot \Big(\Tr\, (-1)^{N} \left[  \left( 1 - \frac{\beta V_{\text{int}}}{M} \right) e^{-\frac{\beta}{M} H^{\text{r}} } e^{-\frac{\beta}{M} \Theta(H^{\text{r}}) } \right]^{M}\Big)^{1/2}\;. 
\end{split}
\end{equation}
Taking the $M\to \infty$ limit, the claim follows.
\end{proof}
Let us now discuss the key consequence of the bound (\ref{eq:RP}), discovered in \cite{Lieb}. Let:
\begin{equation}
Z^{\pm}(\bm{\sigma}) := Z^{\pm}_{\beta,L}(H_{f}(\bm{\sigma}))\;.
\end{equation}
As implied by Proposition \ref{prp:gaugeequiv}, the partition function only depends on the fluxes enclosed by the plaquettes, and by the non-contractible loops. In particular, without loss of generality we can assume that (\ref{eq:Vint}) holds. Let us denote by:
\begin{equation}\label{eq:maxZ}
Z^{\pm}_{*} := \max_{\text{fluxes}} Z^{\pm}(\bm{\sigma})\;,
\end{equation}
where the maximization is over the fluxes over all plaquettes and over the non-contractible cycles. The next result proves the optimality of the $\pi$-flux phase.
\begin{theorem}[Optimality of the $\pi$-flux phase]\label{thm:piflux} Under the same assumptions of Lemma \ref{lem:lieb}, the following is true:
\begin{equation}
Z^{\pm}_{*} = Z^{\pm}({\bm \sigma}_{\pi})\;,
\end{equation}
where ${\bm \sigma}_{\pi}$ is a gauge field configuration realizing the $\pi$-flux phase. 
\end{theorem}
\begin{proof}
Let us denote by $\bm{\sigma}^{\pm}_{*}$ a configuration that attains the maximum (\ref{eq:maxZ}), and let:
\begin{equation}
H^{\pm}_{*} := H_{f}(\bm{\sigma}^{\pm}_{*})\;.
\end{equation} 
Clearly, we can rewrite:
\begin{equation}
H^{\pm}_{*} = (H^{\pm}_{*})^{\text{l}} + (H^{\pm}_{*})^{\text{r}} + V_{\text{int}}\;;
\end{equation}
thus, from the bound (\ref{eq:RP}) we have:
\begin{equation}\label{eq:RPstar}
(Z^{\pm}_{*})^{2} \leq Z^{\pm}_{\beta,L}\big((H^{\pm}_{*})^{\text{l}}; \Theta((H^{\pm}_{*})^{\text{l}})\big) Z^{\pm}_{\beta,L}\big((H^{\pm}_{*})^{\text{r}}; \Theta((H^{\pm}_{*})^{\text{r}})\big)\;. 
\end{equation}
Observe that $\Theta((H^{\pm}_{*})^{\text{l,r}})$ is nothing but $(H^{\pm}_{f}(\bm{\sigma}))^{\text{r,l}}$ evaluated over the {\it opposite} gauge field configuration\footnote{That is, the configuration with reversed signs.} entering in $(H^{\pm}_{*})^{\text{l,r}}$. By definition of $Z^{\pm}_{*}$, Eq. (\ref{eq:RPstar}) holds with equality, and in particular:
\begin{equation}\label{eq:uff}
Z^{\pm}_{*} = Z^{\pm}_{\beta,L}((H^{\pm}_{*})^{{\text{l}}}; \Theta((H^{\pm}_{*})^{{\text{l}}}))\;,\qquad Z^{\pm}_{*} = Z^{\pm}_{\beta,L}((H^{\pm}_{*})^{{\text{r}}}; \Theta((H^{\pm}_{*})^{{\text{r}}}))\;.
\end{equation}
Let us denote by $\tilde{\bm{\sigma}}^{\pm}_{*}$ the gauge field configuration obtained setting to $+1$ all spins intersecting the hyperplane $P$, and such that $\tilde{\sigma}^{\pm}_{*,b} = \sigma^{\pm}_{*,b}$ for bonds $b$ on the left of the cut, and $\tilde \sigma^{\pm}_{*,b} = -\sigma^{\pm}_{*,r(b)}$ for bonds $b$ on the right of the cut, with $r(\cdot)$ the geometric reflection across the cut. Thus, the first of Eq. (\ref{eq:uff}) reads:
\begin{equation}\label{eq:RP1}
Z^{\pm}_{*} = Z^{\pm}_{\beta,L}(H_{f}(\tilde{{\bm \sigma}}^{\pm}_{*}))\;.
\end{equation}
Observe that, by construction, the products of the four spins associated with the elementary plaquettes intersected by the cut are equal to $-1$: that is, these plaquettes are threaded by a $\pi$ flux.  We can now iterate the process, starting from Eq. (\ref{eq:RP1}). We can cut the torus in two halves choosing a hyperplane $P$' different from $P$, and use Lemma \ref{lem:lieb} to prove an upper bound for the right-hand side of (\ref{eq:RP1}). One finds out that $Z^{\pm}_{*}$ is equal to the partition function of a matter sector parametrized by a gauge field configuration with at least {\it three} columns of plaquettes with flux equal to $\pi$. The process can be iterated, and one has:
\begin{equation}\label{eq:finpi}
Z^{\pm}_{*} = Z^{\pm}_{\beta,L}(H_{f}({\bm \sigma}_{\pi}))\;,
\end{equation}
where ${\bm \sigma}_{\pi}$ is a configuration associated with the $\pi$ flux phase.
\end{proof}
\begin{remark}\label{rem:bdry} 
\leavevmode
\begin{enumerate}
\item The gauge field configurations produced by the above iteration has flux $0$ or $\pi$ through the non-contractible cycles, depending on whether $\ell$ is odd or even. In fact, the process of cutting and reflecting does not change the flux along non-contractible cycles parallel to the cut, and sets to $\prod_{i=1}^{\ell-1} (\sigma_i \cdot -\sigma_i) = (-1)^{\ell -1}$ the product of the spins along non-contractible cycles orthogonal to the cut. The claim follows from repeating the argument starting from (\ref{eq:finpi}), using now planes orthogonal to $P$.
\item In the limit $\beta\to 0$, the optimality of the $\pi$-flux phase for the odd partition function $Z_{\beta,L}^{-}({\bf \sigma})$ proved in Theorem \ref{thm:piflux} implies that $|\text{det}\,H({\sigma})| \leq |\text{det}\, H({\bm \sigma}_{\pi})|$, with ${\bm \sigma}_{\pi}$ the $\pi$-flux gauge field configuration produced in the proof of the theorem. This recovers a result of Lieb and Loss \cite{LL}, which was obtained without using reflection positivity.
\end{enumerate}
\end{remark}
To prove the stability of the $\pi$-flux phase against gauge field fluctuations, we have to quantify the energetic penalty due to the removal of magnetic monopoles, that is the insertion of plaquettes with magnetic flux equal to $0$. In fact, the energy of the gauge field associated with a spin configuration ${\bm \sigma}$ can be written as:
\begin{equation}\label{eq:hg}
\begin{split}
H_{\text{g}}({\bm \sigma}) &= - \sum_{\Lambda \in P(\Gamma_{L})} \Big[\prod_{(x,\mu) \in \partial \Lambda} \sigma_{x,\mu}\Big] \\
&= H_{\text{g}}({\bm \sigma}_{\pi}) - k n({\bm \sigma})\;,
\end{split}
\end{equation}
where $n({\bm \sigma})$ is the number of plaquettes with flux equal to $0$. Eq. (\ref{eq:hg}) is the gauge field energy of the $\pi$-flux phase after removing $k$ monopoles. We would like to prove that the coupling with the fermions introduces an energetic penalty for removing monopoles. This can also be done via reflection positivity; we conclude this section by giving a sketch of the argument. A stronger statement will be given in Section \ref{sec:chess}, using the chessboard estimate.

Let $L = 4\ell$. Suppose we start with a configuration of fluxes $\Phi$ different from the $\pi$-flux phase. Thus, there exist at least $1$ plaquette with flux $0$ (actually $2$ by (\ref{stokesb})). The idea is to apply a sequence of reflections to reduce the initial configuration to a periodic one, whose energy can be explicitly computed using Fourier analysis.
\begin{figure}
\centering
\includegraphics[trim={0cm 6cm 0cm 5cm}, clip, width=\textwidth]{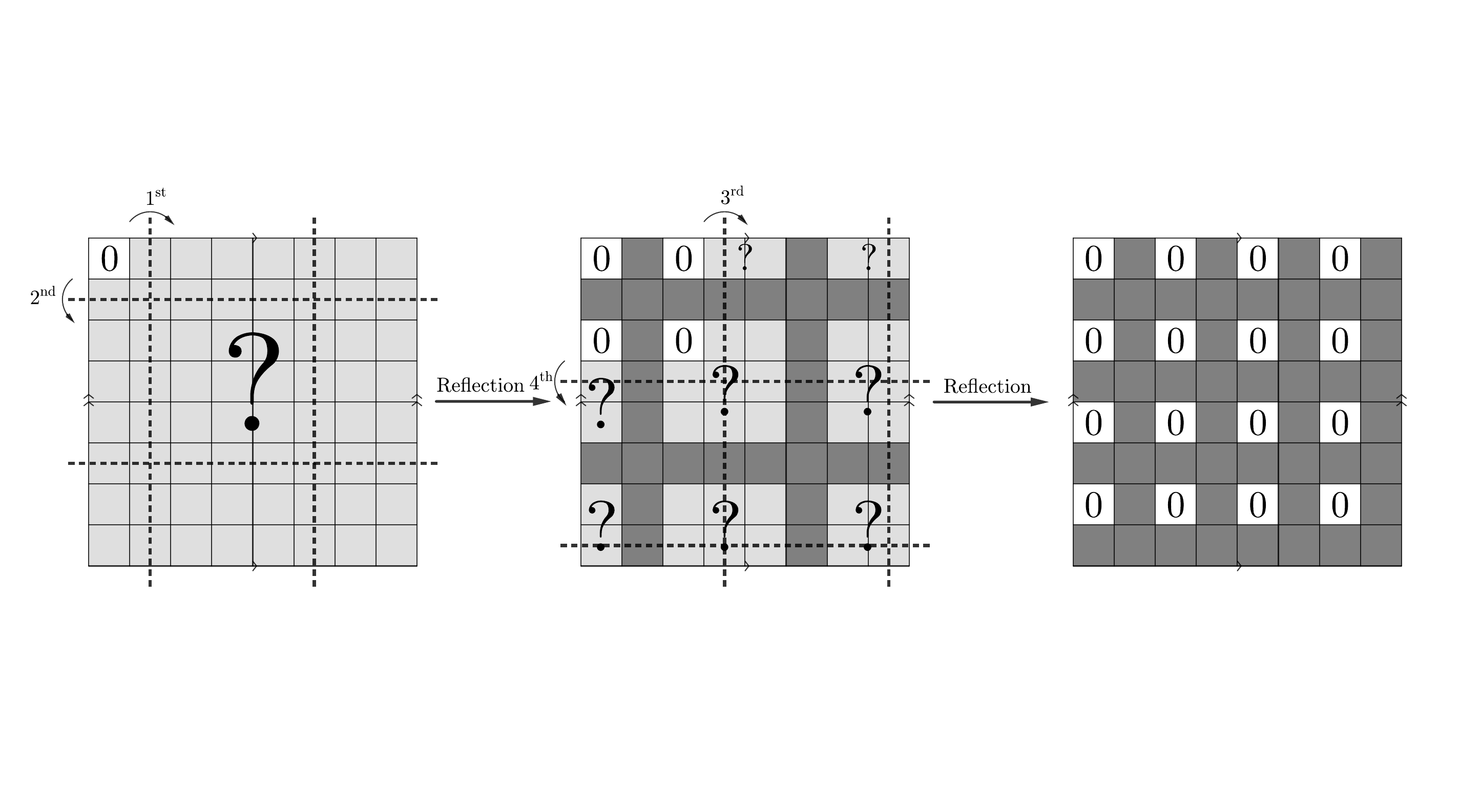}
\caption{Algorithm used to prove the bound for the gap. Dark plaquettes have flux $\pi$.}
\label{specgap}
\end{figure}
Without loss of generality, we can assume that the monopole free face is in the top left corner in the fundamental domain as in Fig. \ref{specgap}. We choose the first reflection hyperplane $P$ to cut the second column (and also the $(2+2\ell)$-th). Upon reflection of the configuration $\Phi$, we obtain a configuration similar to the one shown in Fig. \ref{specgap} and another one whose even/odd partition function can be estimated from above by partition function corresponding to the $\pi$-flux phase with $(a,b) = (-1,-1)$ (recall remark \ref{rem:pf}).

We then choose the horizontal hyperplane cutting the second row. The reflection produces a $3\times 3$ square with corner plaquette with $0$-flux configuration (which we simply call $a_2=3$) and another configuration which we once again estimate with the $\pi$-flux phase (see the second panel in Fig. \ref{specgap}). Thus, after two rounds of reflections, we get:
\begin{equation}
\log Z^{\pm}_{\beta, L}(\Phi)  \leq \bigg(\frac{1}{2} + \frac{1}{4}\bigg) \log Z^{\pm}_{\beta, L}(\pi) + \frac{1}{4} \log Z^{\pm}_{\beta, L}(a_2)\;;
\end{equation}
where we set $Z^{\pm}_{\beta, L}(\pi) = \max_{a,b} Z^{\pm}_{\beta, L}(\pi, a, b)$, and we do not keep track the fluxes on the non-contractible cycles for the first and the last partition function, since they play no role in the argument. We can iterate this sequence of vertical and horizontal reflections to produce squares of length $a_{2m}$ (where $m$ is the number of pairs of reflections) such that:
\begin{equation}
\begin{cases}
a_{2m} = 2 a_{2m -2} +1\\
a_2 = 3\;.
\end{cases}
\end{equation}
The solution of the recursion is $a_{2m} = 2^{m+1}-1$. Iterating the reflection procedure, we obtain:
\begin{equation}\begin{aligned}
\log Z^{\pm}_{\beta, L}(\Phi)  &\leq  \bigg(\sum_{k=1}^{2m} \frac{\log Z^{\pm}_{\beta, L}(\pi) }{2^k}\bigg) + \frac{\log Z^{\pm}_{\beta, L}(a_{2m})}{2^{2m}} \\
&= \bigg(\frac{1-\big(\frac{1}{2}\big)^{2m+1}}{1-\frac{1}{2}} -1 \bigg)  \log Z^{\pm}_{\beta, L}(\pi)   + \frac{\log Z^{\pm}_{\beta, L}(a_{2m})}{2^{2m}} \\
&= \bigg( 1- \frac{1}{2^{2m}}\bigg)  \log Z^{\pm}_{\beta, L}(\pi)  + \frac{\log Z^{\pm}_{\beta, L}(a_{2m})}{2^{2m}}\;. \\
\end{aligned}
\end{equation}
The iteration stops for $m$ equal to the smallest natural number such that:
\begin{equation}
a_{2m} + 1 \geq 4\ell \Rightarrow m = \ceil{\log_2 2\ell}\;.
\end{equation}
Thus:
\begin{equation}\label{eq:refl}
\begin{aligned}
\log Z^{\pm}_{\beta, L}(\Phi) &\leq \bigg( 1 - \frac{1}{2^{2m}} \bigg)  \log Z^{\pm}_{\beta, L}(\pi)  + \frac{1}{2^{2m}}\log Z^{\pm}_{\beta, L}(a_{2m}) \\
&\leq   \bigg( 1 - \frac{1}{2^{2\ceil{\log_2 2\ell}}} \bigg) \log Z^{\pm}_{\beta, L}(\pi)  + \frac{\log Z^{\pm}_{\beta, L}(\Phi^*)}{2^{2\ceil{\log_2 2\ell}}} 
\end{aligned}
\end{equation}
where $\log(Z^{\pm}_{\beta, L}(\Phi^*))$ is the even/odd partition in the background shown in the last panel in Fig. \ref{specgap}. Thus, using (\ref{eq:refl}), we get:
\begin{equation}
\begin{aligned}
\log Z^{\pm}_{\beta, L}(\Phi) - \log Z^{\pm}_{\beta, L}(\pi)& \leq  \frac{1}{2^{2\ceil{\log_2 2\ell}}} \bigg(\log Z^{\pm}_{\beta, L}(\Phi^*) -\log Z^{\pm}_{\beta, L}(\pi)\bigg)\\
&\leq  \frac{1}{16 \ell^2} \bigg(\log Z^{\pm}_{\beta, L}(\Phi^*) -\log Z^{\pm}_{\beta, L}(\pi)\bigg)\\
\end{aligned}
\end{equation}
that is:\footnote{We used that, in general, $\ceil{\log_2 2\ell} \leq \log_2 2\ell + 1$. In the special case in which $\ell$ is a power of $2$, that is if $\ceil{\log_2 2\ell} = \log_2 2\ell$, the right-hand side in (\ref{gap}) is multiplied by an extra factor $4$.}
\begin{equation}\label{gap}
-\frac{1}{\beta}(\log Z^{\pm}_{\beta, L}(\Phi) - \log Z^{\pm}_{\beta, L}(\pi))\geq \frac{1}{\beta L^{2}} \log\frac{Z^{\pm}_{\beta, L}(\pi)}{Z^{\pm}_{\beta, L}(\Phi^*)}\;.
\end{equation}
As $\beta\to \infty$, the left-hand side of (\ref{gap}) is equal to the difference of two fermionic ground state energies, associated with different flux configurations. Instead, the right-hand side of (\ref{gap}) can be explicitly computed, and we will prove that it is bounded below uniformly in $\beta$ and $L$; see Proposition \ref{prp:compD} below. 

The above argument allows to quantity the energetic penalty of removing one monopole. However, the weakness of the argument is that is provides the same lower bound regardless of how many monopoles have been removed. In order to prove the stability of the $\pi$-flux phase against fluctuations of the gauge field, that is to contrast the decrease in $k$ in (\ref{eq:hg}), we need to show that the energetic increase of the fermionic ground state energy due the removal of $k>0$ monopoles grows linearly with $k$. This will be done in Section \ref{sec:chess}, by means of another beautiful consequence of reflection positivity, the chessboard estimate. Intuitively speaking, the chessboard estimate can be understood as a refinement of the above argument which takes into account both reflected configurations at every step.

\subsection{Analysis of the $\pi$-flux phases}\label{sec:piflux}

The argument of the previous section shows that the even and odd partition functions are maximized by a gauge field configuration with flux $\pi$ in all plaquettes, and with flux $0$ or $\pi$ along the non-contractible cycles, depending on whether $L/2 = \ell$ is odd or even, respectively. These configurations are of course far from being unique, due to the fact that all equivalent gauge field configurations produce the same partition function, recall Proposition \ref{prp:gaugeequiv}. 

One can produce different $\pi$-flux phases, choosing different fluxes on the non-contractible cycles: we can think of the $\pi$-flux phases as equivalence classes of gauge field configurations, parametrized by four integers, labelling the product of the spins over two inequivalent, non-contractible loops of the torus. The argument of Lieb shows that the partition function is maximized in correspondence with one particular $\pi$-flux phase; it leaves open the possibility that the other $\pi$-flux phases might have the same energy. This is the question that we will address in the present section, via an explicit computation of the matter partition functions associated with the four inequivalent $\pi$-flux phases. 

Let us introduce the ground state energies associated with the even and odd partition function as:
\begin{equation}\label{eq:gstate}
E^{\pm}_{0;L}(\bm{\sigma}) := - \lim_{\beta \to \infty} \frac{1}{\beta} \log Z^{\pm}_{\beta,L}(\bm{\sigma})\;.
\end{equation}
The reader might be worried about taking the log of $Z^{-}_{\beta,L}$, due to the lack of positivity of the argument of the trace. However, it is not difficult to see that the odd partition function is non-negative at half-filling. In fact, particle-hole symmetry at the level of the first quantized Hamiltonian reads $h(\bm{\sigma}) = - h(\bm{\sigma})^{T}$. That is, there exists a unitary operator $U$ such that $h(\bm{\sigma}) = - U^{*} h(\bm{\sigma}) U$. This implies that the spectrum of $h(\bm{\sigma})$ is symmetric around zero, and that the multiplicities of the eigenvalues $\lambda, -\lambda$ are the same. Let $(\lambda_{i})_{i = 1,\ldots, L^{2}}$ be the eigenvalues of $h(\bm{\sigma})$. The odd partition function is:
\begin{equation}
Z^{\pm}_{\beta,L}(\bm{\sigma}) = \prod_{i=1}^{L^{2}} (1 - e^{-\beta \lambda_{i}})\;.
\end{equation}
The existence of zero modes implies that the odd partition function is zero, and in particular non-negative. Suppose now that there are no zero modes. Using the equality of the multiplicities of the eigenvalues of $h(\bm{\sigma})$, we can rewrite the odd partition function as:
\begin{equation}\label{eq:Z-2}
Z^{-}_{\beta,L}(\bm{\sigma}) = \prod_{i: \lambda_{i} < 0} (1 - e^{-\beta \lambda_{i}}) (1 - e^{\beta \lambda_{i}})\;.
\end{equation}
The function $(1 - e^{-\beta \lambda_{i}}) (1 - e^{\beta \lambda_{i}})$ is non-positive. However, the right-hand side of (\ref{eq:Z-2}) is non-negative, since it involves the product of $L^{2} / 2$ terms, which is even for $L$ even. This shows that $Z^{-}_{\beta,L}(\bm{\sigma}) > 0$, which implies that $E^{-}_{0;L}(\bm{\sigma})$ in (\ref{eq:gstate}) is real.

A corollary of Theorem \ref{thm:piflux} is that:
\begin{equation}
E^{\pm}_{0;L}(\bm{\sigma}) \geq E^{\pm}_{0;L}(\bm{\sigma}_{\pi})\;,
\end{equation}
where ${\bm \sigma}_{\pi}$ is a configuration with flux equal to $\pi$ in every plaquette. In this section we shall consider all such phases; these are parametrized by $(a,b) \in \mathbb{Z}_{2}^{2}$, which correspond to the fluxes along two non-contractible cycles of the torus (recall Remark \ref{rem:not}). We shall denote by $E^{\pm}_{0;L}(a,b)$ the corresponding energies. The next proposition proves the energetic closeness of these phases. It proves, in particular, the statement (\ref{eq:deg}) in Theorem \ref{thm:main}.
\begin{proposition}[Almost-degeneracy of the $\pi$-flux phases]\label{prp:deg} Let $L = 2\ell$. For $\ell = 2n$:
\begin{equation}
Z^{-}_{\beta,L}(\pi; 1,1) = 0\;,\qquad Z^{-}_{\beta,L}(\pi; a,b) > 0\quad \text{for $(a,b) \neq (1,1)$.}
\end{equation}
Furthermore, for $(\varepsilon, a, b) \neq (-1,1,1)$, and $0<\alpha < 1/3$:
\begin{equation}\label{eq:endiff1}
| E^{\varepsilon}_{0;L}(\pi;a,b) - E^{\varepsilon}_{0;L}(\pi;1,1) | \leq \frac{C_{\alpha}}{L^{1 - 3\alpha}}\;.
\end{equation}
For $\ell = 2n+1$:
\begin{equation}
Z^{-}_{\beta,L}(\pi; -1,-1) = 0\;,\qquad Z^{-}_{\beta,L}(\pi; a,b) > 0\quad \text{for $(a,b) \neq (-1,-1)$.}
\end{equation}
Furthermore, for $(\varepsilon, a, b) \neq (-1,-1,-1)$, and $0<\alpha < 1/3$:
\begin{equation}\label{eq:endiff2}
| E^{\varepsilon}_{0;L}(\pi;a,b) - E^{\varepsilon}_{0;L}(\pi;-1,-1) | \leq \frac{C_{\alpha}}{L^{1 - 3\alpha}}\;.
\end{equation}
\end{proposition}
\begin{remark}
\leavevmode
\begin{itemize}
\item[(i)] We stress that bounds (\ref{eq:endiff1}) and (\ref{eq:endiff2}) estimate the difference of estensive quantities. 
\item[(ii)] From reflection positivity, we know that:
\begin{enumerate}
\item  $E^{\varepsilon}_{0;L}(-1,-1) \leq E^{\varepsilon}_{0;L}(a,b)$ for $\ell$ even
\item  $E^{\varepsilon}_{0;L}(1,1) \leq E^{\varepsilon}_{0;L}(a,b)$ for $\ell$ odd
\end{enumerate}
 recall Remark \ref{rem:bdry}. The estimates (\ref{eq:endiff1}), (\ref{eq:endiff2}) do not rule out possible degeneracies; however, they show that the energy variation due to the change of fluxes on non-contractible cycles is bounded by a vanishing quantity as $\ell \to \infty$.
\end{itemize}
\end{remark}
\begin{proof} We shall only discuss the case $L/2 = 2n$, as the case $L/2 = 2n+1$ is completely analogous. The proof of Proposition \ref{prp:deg} is based on the Bloch diagonalization of the Hamiltonian associated with the four inequivalent $\pi$-fluxes. Let us start from the $\pi$-flux phase parametrized by $(a,b) = (1,1)$. A particularly simple gauge field configuration associated with this phase is represented in Fig. \ref{fig:pi1}. Clearly, the corresponding Hamiltonian is not translationally invariant with respect to all lattice translations; however, it becomes translation invariant after having introduced the appropriate fundamental cell. We introduce a fundamental cell formed by two lattice sites, labelled by $A, B$, as in Fig. \ref{fig:pi1}. The position of the cell is defined by the coordinate of the $B$-lattice site. We shall denote by $\Gamma_{L}^{\text{red}}$ the two-dimensional lattice formed by the positions of the fundamental cells. The number of lattice sites contained in $\Gamma_{L}^{\text{red}}$ is $L^{2} / 2$. Observe that the lattice spacing between nearest-neighbours fundamental cells in $\Gamma_{L}^{\text{red}}$ is $1$ in the horizontal direction, and $2$ in the vertical direction.

Using these coordinates, the Hamiltonian of the matter sector is:
\begin{equation}\label{eq:H11}
H_{f}(\pi; 1,1) = -t \sum_{x \in \Gamma_{L}^{\text{red}}} ( a^{+}_{A,x} a^{-}_{B,x+2e_2} - a^{+}_{A,x} a^{-}_{A,x+e_{1}}  + a^{+}_{B,x} a^{-}_{A,x} + a^{+}_{B,x} a^{-}_{B, x+e_{1}} ) + \text{h.c.}
\end{equation}

\begin{figure}
\centering
\begin{tikzpicture}[scale=0.6]
\begin{scope}[xshift=-1cm, yshift = 2cm]
\draw (0,0) grid (8,8);
\draw[->-] (0,8) -- (8,8);
\draw[->-] (0,0) -- (8,0);
\draw[->>-] (0,0) -- (0,8);
\draw[->>-] (8,0) -- (8,8);
\draw[line width = 0.05 cm] (0,0) -- (8,0);
\draw[line width = 0.05 cm] (0,2) -- (8,2);
\draw[line width = 0.05 cm] (0,4) -- (8,4);
\draw[line width = 0.05 cm] (0,6) -- (8,6);
\draw[line width = 0.05 cm] (0,8) -- (8,8);
\draw[dotted] (0.5,0) -- (0.5,8);
\draw[dotted] (1.5,0) -- (1.5,8);
\draw[dotted] (2.5,0) -- (2.5,8);
\draw[dotted] (3.5,0) -- (3.5,8);
\draw[dotted] (4.5,0) -- (4.5,8);
\draw[dotted] (5.5,0) -- (5.5,8);
\draw[dotted] (6.5,0) -- (6.5,8);
\draw[dotted] (7.5,0) -- (7.5,8);
\draw[dotted] (0,0.5) -- (8,0.5);
\draw[dotted] (0,2.5) -- (8,2.5);
\draw[dotted] (0,4.5) -- (8,4.5);
\draw[dotted] (0,6.5) -- (8,6.5);
\draw[dotted, fill = gray, fill opacity = 0.2] (3.5,2.5)--(4.5,2.5) -- (4.5,4.5) -- (3.5,4.5) -- (3.5,2.5);
\end{scope}
\begin{scope}[xshift =11.5 cm, yshift=4cm]
\draw[dotted,fill = gray, fill opacity = 0.2] (0,0) -- (0,4) -- (2,4)--(2,0)--(0,0);
\draw[] (1,0) -- (1,4);
\draw[] (0,1) -- (2,1);
\draw[thick, line width = 0.05 cm] (0,3) -- (2,3);
\node[below left] (a) at (1,1) {\scalebox{0.6}{$B$}};
\node[above left] (b) at (1,3) {\scalebox{0.6}{$A$}};
\draw[line width = 0.02 cm, ->] (0,0)--(0,4);
\draw[line width = 0.02 cm, ->] (0,0)--(2,0);
\end{scope}
\end{tikzpicture}
\caption{On the left, graphical representation of the spin configurations associated with the $(1,1)$ $\pi$ flux phase. The thin bonds correspond to spin $+1$, while the solid bonds correspond to spin $-1$. On the right, definition of the fundamental cell.}
\label{fig:pi1}
\end{figure}
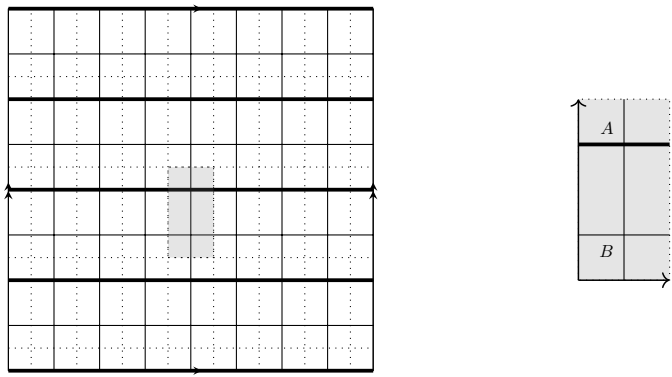

Let us introduce the Brillouin zone $B_{L}(1,1)$ as:
\begin{equation}\label{eq:BL11}
B_{L}(1,1) := \Big\{ k\in \frac{2\pi}{L}(n_{1}, n_{2}) \mid 0\leq n_{1} \leq L-1,\; 0\leq n_{2} \leq L/2 - 1 \Big\}\;,
\end{equation}
and let us define the momentum-space creation/annihilation operators as:
\begin{equation}\label{eq:afou}
\hat a^{\pm}_{\alpha, k} = \sum_{x\in \Gamma_{L}^{\text{red}}} e^{\pm ik\cdot x} a^{\pm}_{\alpha, x} \iff a^{\pm}_{\alpha, x} = \frac{1}{|\Gamma^{\text{red}}_{L}|} \sum_{k \in B_{L}(1,1)} e^{\mp ik\cdot x} \hat a^{\pm}_{\alpha, k}\;,
\end{equation}
for $\alpha \in \{A,B\}$. In terms of these operators, the Hamiltonian (\ref{eq:H11}) becomes:
\begin{equation}
H_{f}(\pi; 1,1) = \frac{1}{|\Gamma^{\text{red}}_{L}|} \sum_{k \in B_{L}(1,1)} (\hat a^{+}_{k}, h(k) \hat a^{-}_{k})\;,
\end{equation}
where $(f,g) = \sum_{\alpha} f_{\alpha} g_{\alpha}$, and the Bloch Hamiltonian $h(k)$ is:
\begin{equation}\label{eq:blochH}
h(k) = -t \begin{pmatrix} - e^{ik_{1}} - e^{-ik_{1}} & 1 + e^{2ik_{2}}  \\ 1 + e^{-2ik_{2}} & e^{ik_{1}} + e^{-ik_{1}} \end{pmatrix}\;.
\end{equation}
Its eigenvalues are:
\begin{equation}\label{eq:epi}
e_{\pm}(k) = \pm 2t \sqrt{1 + \frac{1}{2}\cos(2k_{1}) + \frac{1}{2} \cos(2k_{2})}\;.
\end{equation}
The eigenvalues are vanishing if and only if $k = (\frac{\pi}{2}, \frac{\pi}{2}), (\frac{3\pi}{2}, \frac{\pi}{2})$. The partition function is:
\begin{equation}
\begin{split}
Z^{\pm}_{\beta,L}(\pi; 1,1) &= \prod_{k\in B_{L}(1,1)} \prod_{z = \pm } (1 \pm e^{-\beta e_{z}(k)})  \\
&= \prod_{k\in B_{L}(1,1)} (2 \pm (e^{-\beta e_{-}(k)} + e^{\beta e_{-}(k)}))\;.
\end{split}
\end{equation}
From this expression, we see that $Z^{-}_{\beta,L}(\pi; 1,1) = 0$, due to the presence of zero modes at $k = (\frac{\pi}{2}, \frac{\pi}{2}), (\frac{3\pi}{2}, \frac{\pi}{2})$, which belong to $B_{L}(1,1)$ since $L/2$ is even. Instead, the ground state energy $E^{+}_{0;L}(\pi; 1,1)$, defined as in (\ref{eq:gstate}), is:
\begin{equation}
E^{+}_{0;L}(\pi; 1,1) = \sum_{k\in B_{L}(1,1)} e_{-}(k)\;.
\end{equation}
Let us now consider the ground state energy associated with the $\pi$-flux phases with $(a,b) \neq (1,1)$. The representative spin configurations associated with the remaining cases are illustrated in Fig. \ref{fig:pifluxes}.

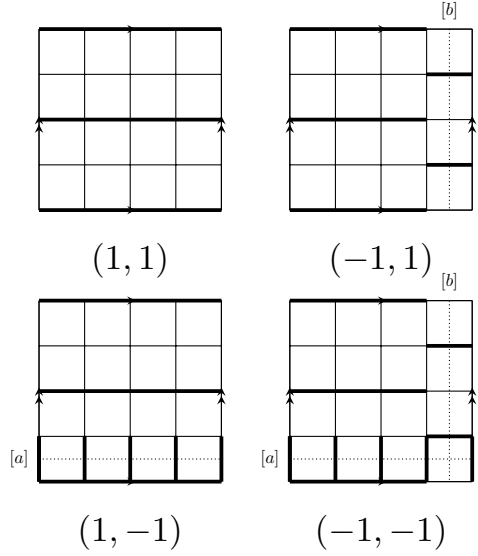
\begin{figure}
\centering
\begin{tikzpicture}[scale=0.6]
\draw (0,0) grid (4,4);
\draw[->-] (0,4) -- (4,4);
\draw[->-] (0,0) -- (4,0);
\draw[->>-] (0,0) -- (0,4);
\draw[->>-] (4,0) -- (4,4);
\draw[line width = 0.05 cm] (0,0) -- (4,0);
\draw[line width = 0.05 cm] (0,2) -- (4,2);
\draw[line width = 0.05 cm] (0,4) -- (4,4);
\node[below] (a) at (2,-0.5) {\scalebox{1.2}{$(1,1)$}};
\begin{scope}[xshift = 5.5cm]
\draw (0,0) grid (4,4);
\draw[->-] (0,4) -- (4,4);
\draw[->-] (0,0) -- (4,0);
\draw[->>-] (0,0) -- (0,4);
\draw[->>-] (4,0) -- (4,4);
\draw[line width = 0.05 cm] (0,0) -- (3,0);
\draw[line width = 0.05 cm] (0,2) -- (3,2);
\draw[line width = 0.05 cm] (0,4) -- (3,4);
\draw[line width = 0.05 cm] (3,3) -- (4,3);
\draw[line width = 0.05 cm] (3,1) -- (4,1);
\draw[densely dotted] (3.5,4)-- (3.5,0);
\node[above] (b) at (3.5,4) {\scalebox{0.6}{$[b]$}};
\node[below] (a) at (2,-0.5) {\scalebox{1.2}{$(-1,1)$}};
\end{scope}
\begin{scope}[xshift=5.5 cm, yshift=-6cm]
\draw (0,0) grid (4,4);
\draw[->-] (0,4) -- (4,4);
\draw[->-] (0,0) -- (4,0);
\draw[->>-] (0,0) -- (0,4);
\draw[->>-] (4,0) -- (4,4);
\draw[line width = 0.05 cm] (0,0) -- (3,0);
\draw[line width = 0.05 cm] (0,2) -- (3,2);
\draw[line width = 0.05 cm] (0,4) -- (3,4);
\draw[line width = 0.05 cm] (3,3) -- (4,3);
\draw[line width = 0.05 cm] (3,1) -- (4,1);
\draw[line width = 0.05 cm] (0,0) -- (0,1);
\draw[line width = 0.05 cm] (0,0) -- (0,1);
\draw[line width = 0.05 cm] (1,0) -- (1,1);
\draw[line width = 0.05 cm] (2,0) -- (2,1);
\draw[line width = 0.05 cm] (3,0) -- (3,1);
\draw[line width = 0.05 cm] (4,0) -- (4,1);
\draw[densely dotted] (0,0.5) -- (4,0.5);
\draw[densely dotted] (3.5,4) -- (3.5,0);
\node[above] (b) at (3.5,4) {\scalebox{0.6}{$[b]$}};
\node[left] (b) at (0,0.5) {\scalebox{0.6}{$[a]$}};
\node[below] (a) at (2,-0.5) {\scalebox{1.2}{$(-1,-1)$}};
\end{scope}
\begin{scope}[yshift=-6cm]
\draw (0,0) grid (4,4);
\draw[->-] (0,4) -- (4,4);
\draw[->-] (0,0) -- (4,0);
\draw[->>-] (0,0) -- (0,4);
\draw[->>-] (4,0) -- (4,4);
\draw[line width = 0.05 cm] (0,0) -- (4,0);
\draw[line width = 0.05 cm] (0,2) -- (4,2);
\draw[line width = 0.05 cm] (0,4) -- (4,4);
\draw[line width = 0.05 cm] (0,0) -- (0,1);
\draw[line width = 0.05 cm] (0,0) -- (0,1);
\draw[line width = 0.05 cm] (1,0) -- (1,1);
\draw[line width = 0.05 cm] (2,0) -- (2,1);
\draw[line width = 0.05 cm] (3,0) -- (3,1);
\draw[line width = 0.05 cm] (4,0) -- (4,1);
\draw[densely dotted] (0,0.5) -- (4,0.5);
\node[left] (b) at (0,0.5) {\scalebox{0.6}{$[a]$}};
\node[below] (a) at (2,-0.5) {\scalebox{1.2}{$(1,-1)$}};
\end{scope}
\end{tikzpicture}
\caption{All possible $\pi$-flux phases, for $L=4$.}
\label{fig:pifluxes}
\end{figure}
Consider the case $(a,b) = (-1,1)$; all the others will be discussed in a similar way. The only difference with respect to the case $(a,b) = (1,1)$ is that the hopping terms $a^{+}_{L-1, x_{2}} a^{-}_{L, x_{2}} + \text{h.c.} \equiv a^{+}_{L-1, x_{2}} a^{-}_{0, x_{2}} + \text{h.c.}$ in the $(1,1)$ case are replaced in the $(-1,1)$ case by $-a^{+}_{L-1, x_{2}} a^{-}_{L, x_{2}} + \text{h.c.} \equiv -a^{+}_{L-1, x_{2}} a^{-}_{0, x_{2}} + \text{h.c.}$. Thus, the Hamiltonian of the case $(-1,1)$ can be rewritten as the Hamiltonian of the case $(1,1)$ if one reabsorbs the minus sign imposing the new rule $a^{\pm}_{L, x_{2}} \equiv - a^{\pm}_{0, x_{2}} $, which is equivalent to enforcing anti-periodic boundary conditions in the $x_{1}$ direction. 

The Brillouin zone corresponding to anti-periodic boundary conditions in $x_{1}$ and periodic boundary conditions in $x_{2}$ is:
\begin{equation}
B_{L}(-1,1) := \Big\{ k\in \frac{2\pi}{L}\Big(n_{1} - \frac{1}{2}, n_{2}\Big) \mid 0\leq n_{1} \leq L-1,\; 0\leq n_{2} \leq L/2-1  \Big\}\;.
\end{equation}
Thanks to the $-1/2$ shift, the system does not admit any zero mode here, which means that the odd partition function is non-zero. The even and odd ground state energies are:
\begin{equation}
E^{\pm}_{0;L}(\pi; -1,1) = \sum_{k\in B_{L}(-1,1)} e_{-}(k)\;.
\end{equation}
More generally, defining, for $a,b = \pm 1$,
\begin{equation}\label{eq:bab}
B_{L}(a,b) := \Big\{ k\in \frac{2\pi}{L}\Big(n_{1} - \frac{a-1}{4}, n_{2} - \frac{b-1}{4}\Big) \mid 0\leq n_{1} \leq L-1,\; 0\leq n_{2} \leq L/2-1  \Big\}\;,
\end{equation}
the even/odd ground state energies of all $\pi$-flux phases are, for $(\varepsilon, a, b) \neq (-1,1,1)$:
\begin{equation}
\begin{split}
E^{\varepsilon}_{0;L}(\pi; a,b) &= \sum_{k\in B_{L}(a,b)} e_{-}(k) \\
&= \sum_{k\in B_{L}(1,1)} e_{-}\Big(k_1 - \frac{\pi}{2L}(a-1), k_2- \frac{\pi}{2L}(b-1) \Big)\;.
\end{split}
\end{equation}
Therefore, for these choices of labels:
\begin{equation}\label{eq:endif}
\Big| E^{\varepsilon}_{0;L}(\pi; a,b) - E^{\varepsilon}_{0;L}(\pi; 1,1)  \Big| = \Big|  \sum_{k\in B_{L}(1,1)} \Big[e_{-}\Big( k_1 - \frac{\pi}{2L}(a-1), k_2 - \frac{\pi}{2L}(b-1)  \Big) - e_{-}(k_1, k_2 ) \Big]\Big|\;;
\end{equation}
as shown in Appendix \ref{app:diff}, one has, for $0< \alpha < 1/3$:
\begin{equation}\label{eq:bden}
\Big|  \sum_{k\in B_{L}(1,1)} \Big[e_{-}\Big(k_1 - \frac{\pi}{2L}(a-1),k_2 - \frac{\pi}{2L}(b-1)  \Big) - e_{-}(k_1, k_2 ) \Big]\Big| \leq \frac{C_{\alpha}}{L^{1 - 3\alpha}}\;,
\end{equation}
which concludes the proof of (\ref{eq:endiff1}). The case $L/2 = 2n+1$ can be studied in the same way, we omit the details.
\end{proof}
It is interesting to observe that, as $L\to \infty$, the energy bands of $h(k)$ have conical intersections at energy zero, and hence describe a semimetal at half-filling. Let $k^{-}_{F}= (\frac{\pi}{2}, \frac{\pi}{2})$, $k_{F}^{+} = (\frac{3\pi}{2}, \frac{\pi}{2})$. In proximity of these points, we can linearize the dispersion relation as, for $\omega = \pm$:
\begin{equation}\label{eq:dirac}
e_{\pm}(k^{\omega}_{F} + q) = q\cdot \nabla_{q} e_{\pm}(k_{F}) + O(|q|^{2}) = \mp 2t |q| + O(|q|^{2})\;.
\end{equation}
Eq. (\ref{eq:dirac}) describes an isotropic cone, with slope $2t$. This means that the low-energy excitations of $h(k)$ are effectively described by massless Dirac fermions, with Fermi velocity $2t$. This fact will play an important role in the evaluation of the susceptibility, Proposition \ref{prp:susc}.

\subsection{Chessboard estimates}\label{sec:chess}
The goal of this section will be to prove a lower bound for the energy increase due to the removal of monopoles, that is the insertion of plaquettes with flux $0$ instead of $\pi$. We will prove a lower bound for the energy of the system that grows linearly in the number of monopoles' removal. This will be crucial to prove the stability of the $\pi$-flux phase in the lattice gauge theory.

The key technical tool used to prove this lower bound is the chessboard estimate. The next result is an adaptation of \cite[Lemma 4.5]{Tasaki}, to our setting.
\begin{lemma}[Chessboard estimate]\label{lem:chess} Let $L = 4\ell$, and let $F: \mathbb{Z}_{2}^{\Gamma_{L}} \to \mathbb{R}$ be a function, which we represent as $F(\underline{X}_{1}, \ldots, \underline{X}_{L})$ with $\underline{X}_{j} \in \mathbb{Z}_{2}^{L}$. Suppose that it satisfies the following properties:
\begin{equation}\label{eq:assce}
\begin{split}
F(\underline{X}_{1}, \ldots, \underline{X}_{L}) &= F(\underline{X}_{L}, \underline{X}_{1}, \ldots, \underline{X}_{L-1})\hspace{6cm} \text{(cyclicity)} \\
F(\underline{X}_{1}, \ldots, \underline{X}_{L}) &\geq \frac{1}{2} \Big( F(\underline{X}_{1}, \ldots, \underline{X}_{L/2-1}, \underline{-1}, \underline{X}_{L/2 - 1}, \ldots, \underline{X}_{1}, \underline{-1})\\&\quad + F(\underline{X}_{L}, \ldots, \underline{X}_{L/2 + 1}, \underline{-1}, \underline{X}_{L/2 + 1},\ldots, \underline{X}_{L}, \underline{-1})\Big)\qquad \text{(reflection bound)}
\end{split}
\end{equation}
where $\underline{-1}$ is a column vector with $L$ entries, all equal to $-1$. Then:
\begin{equation}\label{eq:CE}
F(\underline{X}_{1}, \ldots, \underline{X}_{L}) \geq \frac{1}{L/2} \sum_{j=1}^{L/2} F(\underline{X}_{2j-1}, \underline{-1}, \underline{X}_{2j-1}, \underline{-1}, \ldots, \underline{X}_{2j-1}, \underline{-1})\;.
\end{equation}
\end{lemma}
\begin{proof} Let us define the function:
\begin{equation}\label{eq:Gdef}
G(\underline{X}_{1}, \ldots, \underline{X}_{L}) := F(\underline{X}_{1}, \ldots, \underline{X}_{L}) - \frac{1}{L/2} \sum_{j=1}^{L/2} F(\underline{X}_{2j-1}, \underline{-1}, \underline{X}_{2j-1}, \underline{-1}, \ldots, \underline{X}_{2j-1}, \underline{-1})\;.
\end{equation}
We claim that this function satisfies both properties above. The cyclicity of $G$ follows from the cyclicity of $F$, and from the fact that the second term is invariant under permutations $\underline{X}_{2j-1} \to \underline{X}_{2\sigma(j)-1}$ for all $\sigma \in S_{L/2}$. Let us now check the reflection bound. Applying the reflection bound for the first term in $G$, we have:
\begin{equation}\label{eq:reflG}
\begin{split}
G(\underline{X}_{1}, \ldots, \underline{X}_{L}) &\geq \frac{1}{2} \Big( F(\underline{X}_{1}, \ldots, \underline{X}_{L/2-1}, \underline{-1}, \underline{X}_{L/2 - 1}, \ldots, \underline{X}_{1}, \underline{-1})\\&\quad + F(\underline{X}_{L}, \ldots, \underline{X}_{L/2 + 1}, \underline{-1}, \underline{X}_{L/2 + 1},\ldots, \underline{X}_{L}, \underline{-1})\Big) \\
&\quad - \frac{1}{L/2} \sum_{j=1}^{L/2} F(\underline{X}_{2j-1}, \underline{-1}, \underline{X}_{2j-1}, \underline{-1}, \ldots, \underline{X}_{2j-1}, \underline{-1})\;.
\end{split}
\end{equation}
Let:
\begin{equation}
\widetilde{F}(\underline{X}_{1}, \ldots, \underline{X}_{L}) := \frac{1}{L/2} \sum_{j=1}^{L/2} F(\underline{X}_{2j-1}, \underline{-1}, \underline{X}_{2j-1}, \underline{-1}, \ldots, \underline{X}_{2j-1}, \underline{-1})\;;
\end{equation}
from:
\begin{equation}
\begin{split}
\widetilde{F}(\underline{X}_{1}, \ldots, \underline{X}_{L}) &= \frac{1}{2} \Big( \widetilde{F}(\underline{X}_{1}, \ldots, \underline{X}_{L/2-1}, \underline{-1}, \underline{X}_{L/2 - 1}, \ldots, \underline{X}_{1}, \underline{-1})\\&\quad + \widetilde{F}(\underline{X}_{L}, \ldots, \underline{X}_{L/2 + 1}, \underline{-1}, \underline{X}_{L/2 + 1},\ldots, \underline{X}_{L}, \underline{-1})\Big)\;,
\end{split}
\end{equation}
and (\ref{eq:reflG}), we have:
\begin{equation}
\begin{split}
G(\underline{X}_{1}, \ldots, \underline{X}_{L}) &\geq \frac{1}{2} \Big( G(\underline{X}_{1}, \ldots, \underline{X}_{L/2-1}, \underline{-1}, \underline{X}_{L/2 - 1}, \ldots, \underline{X}_{1}, \underline{-1})\\&\quad + G(\underline{X}_{L}, \ldots, \underline{X}_{L/2 + 1}, \underline{-1}, \underline{X}_{L/2 + 1},\ldots, \underline{X}_{L}, \underline{-1})\Big)\;,
\end{split}
\end{equation}
which proves the reflection bound for $G$. Next, we shall prove (\ref{eq:CE}). Let us define:
\begin{equation}
G_{*} = \min_{\underline{X}_{1}, \ldots,\, \underline{X}_{L} \in \mathbb{Z}_{2}^{\Gamma_{L}}} G(\underline{X}_{1}, \ldots, \underline{X}_{L})\;.
\end{equation}
Let $\underline{X}_{1}^{*}, \ldots, \underline{X}_{L}^{*}$ be a minimizing configuration. Then, by the reflection bound:
\begin{equation}
\begin{split}
G(\underline{X}_{1}^{*}, \ldots, \underline{X}_{L}^{*}) &\geq \frac{1}{2} \Big( G(\underline{X}^{*}_{1}, \ldots, \underline{X}^{*}_{L/2-1}, \underline{-1}, \underline{X}^{*}_{L/2 - 1}, \ldots, \underline{X}^{*}_{1}, \underline{-1})\\&\quad + G(\underline{X}^{*}_{L}, \ldots, \underline{X}^{*}_{L/2 + 1}, \underline{-1}, \underline{X}^{*}_{L/2 + 1},\ldots, \underline{X}^{*}_{L}, \underline{-1})\Big)\;,
\end{split}
\end{equation}
and since $\{\underline{X}_{j}^{*}\}$ is a minimizing configuration:
\begin{equation}
G(\underline{X}_{1}^{*}, \ldots, \underline{X}_{L}^{*}) = G(\underline{X}^{*}_{1}, \ldots, \underline{X}^{*}_{L/2-1}, \underline{-1}, \underline{X}^{*}_{L/2 - 1}, \ldots, \underline{X}^{*}_{1}, \underline{-1})\;.
\end{equation}
Now, by cyclicity,
\begin{equation}
G(\underline{X}_{1}^{*}, \ldots, \underline{X}_{L}^{*}) = G(\underline{X}^{*}_{1}, \underline{-1}, \underline{X}^{*}_{1}, \ldots, \underline{X}^{*}_{L/2-1}, \underline{-1}, \underline{X}^{*}_{L/2 - 1}, \ldots, \underline{X}_{2}^{*})\;.
\end{equation}
Iterating the argument, we ultimately get:
\begin{equation}
G(\underline{X}_{1}^{*}, \ldots, \underline{X}_{L}^{*}) = G(\underline{X}_{1}^{*}, \underline{-1}, \underline{X}_{1}^{*}, \underline{-1}, \ldots, \underline{X}_{1}^{*}, \underline{-1})\;,
\end{equation}
and recalling the definition of $G$, Eq. (\ref{eq:Gdef}), we find:
\begin{equation}
G(\underline{X}_{1}^{*}, \ldots, \underline{X}_{L}^{*}) = 0\;.
\end{equation}
Therefore, since $G(\underline{X}_{1}, \ldots, \underline{X}_{L}) \geq G(\underline{X}_{1}^{*}, \ldots, \underline{X}_{L}^{*})$, we obtain:
\begin{equation}
F(\underline{X}_{1}, \ldots, \underline{X}_{L}) \geq \frac{1}{L/2} \sum_{j=1}^{L/2} F(\underline{X}_{2j-1}, \underline{-1}, \underline{X}_{2j-1}, \underline{-1}, \ldots, \underline{X}_{2j-1}, \underline{-1})\;,
\end{equation}
which concludes the proof of the lemma.
\end{proof}
We will apply this argument to the even/odd partition functions of the matter sector. To do this, let us denote by $\phi_{i,j}$ the flux associated with the plaquette $(i,j)$ with $1\leq i \leq L$; the labelling of the plaquettes coincides with the coordinates of the vertex of the plaquette sitting in the up-right corner. Recall that, see Remark \ref{rem:not}, we are identifying the flux through a plaquette with the product of the spins around the same plaquette (same for the non-contractible loops). Let us denote by $\underline{\Phi}_{j} \in \mathbb{Z}_{2}^{L}$ the collection of all magnetic fluxes with column label $j$. We define:
\begin{equation}\label{eq:Fdef}
F^{\pm}(\underline{\Phi}_{1}, \ldots, \underline{\Phi}_{L}) := - \max_{a,b}\log Z^{\pm}_{\beta,L}(\underline{\Phi}_{1}, \ldots, \underline{\Phi}_{L}; a,b)
\end{equation}
where $Z^{\pm}_{\beta,L}(\underline{\Phi}_{1}, \ldots, \underline{\Phi}_{L}; a,b)$ is the even/odd partition function associated with the given configuration of magnetic fluxes, and with fluxes $(a,b)$ on the non-contractible loops. In case zero modes are present for a given $(a,b)$, $\log Z^{-}$ is understood as $-\infty$.

Let us now check that the functions $F^{\pm}$ satisfies the hypothesis of Lemma \ref{lem:chess}, starting from cyclicity. Let $\mathcal{T}_{a}$ be the unitary operator implementing the lattice translation by $a\in \Gamma_{L}$:
\begin{equation}
\mathcal{T}_{a}^{*} a^{\pm}_{x} \mathcal{T}_{a} = a^{\pm}_{x + a}\;.
\end{equation} 
Calling $\bm{\sigma}$ a gauge field configuration associated with the configuration of fluxes $\underline{\Phi}_{L}, \underline{\Phi}_{1}, \ldots, \underline{\Phi}_{L-1}$, we have:
\begin{equation}
\begin{split}
F^{\pm}(\underline{\Phi}_{L}, \underline{\Phi}_{1}, \ldots, \underline{\Phi}_{L-2}, \underline{\Phi}_{L-1}) &= - \max_{a,b}\log Z^{\pm}_{\beta,L}(\underline{\Phi}_{L}, \ldots, \underline{\Phi}_{L-2}, \underline{\Phi}_{L-1}; a,b) \\
&= - \max_{a,b}\log \Tr\, (\pm 1)^{N} e^{-\beta H_{f}(\bm{\sigma})} \\
&= - \max_{a,b}\log \Tr \mathcal{T}_{-e_{1}}^{*} (\pm 1)^{N} e^{-\beta H_{f}(\bm{\sigma})} \mathcal{T}_{-e_{1}} \\
&= - \max_{a,b}\log \Tr\, (\pm 1)^{N} e^{-\beta H_{f}(\tilde {\bm \sigma})}\;.
\end{split}
\end{equation}
where $\tilde {\bm \sigma}$ is a gauge field configuration associated with the configuration of fluxes $\underline{\Phi}_{1}, \underline{\Phi}_{2}, \ldots, \underline{\Phi}_{L}$ (the periodicity of the torus has been used). This proves cyclicity in (\ref{eq:assce}). Concerning the reflection estimate, this simply follows from the reflection positivity bound (\ref{eq:RP}), after taking logs. We have, for suitable $(a_{i},b_{i})$, $i=1,2$:
\begin{equation}
\begin{split}
&- \log Z^{\pm}_{\beta,L}(\underline{\Phi}_{1}, \ldots, \underline{\Phi}_{L}; a,b) \\
&\qquad \geq \frac{1}{2} \Big( - \log Z^{\pm}_{\beta,L}(\underline{\Phi}_{1}, \ldots, \underline{\Phi}_{L/2-1}, \underline{-1},\underline{\Phi}_{L/2-1}, \ldots, \underline{\Phi}_{1}, \underline{-1}; a_{1},b_{1}) \\
&\quad \qquad - \log Z^{\pm}_{\beta,L}(\underline{\Phi}_{L}, \ldots, \underline{\Phi}_{L/2+1}, \underline{-1},\underline{\Phi}_{L/2+1}, \ldots, \underline{\Phi}_{L}, \underline{-1}; a_{2},b_{2})\Big) \\
&\qquad \geq \frac{1}{2} \Big( - \max_{a,b}\log Z^{\pm}_{\beta,L}(\underline{\Phi}_{1}, \ldots, \underline{\Phi}_{L/2-1}, \underline{-1},\underline{\Phi}_{L/2-1}, \ldots, \underline{\Phi}_{1}, \underline{-1}; a,b) \\
&\quad \qquad - \max_{a,b}\log Z^{\pm}_{\beta,L}(\underline{\Phi}_{L}, \ldots, \underline{\Phi}_{L/2+1}, \underline{-1},\underline{\Phi}_{L/2+1}, \ldots, \underline{\Phi}_{L}, \underline{-1}; a,b)\Big)\;.
\end{split}
\end{equation}
Minimizing the left-hand side over $a,b$, we see that the function (\ref{eq:Fdef}) satisfies the reflection bound (\ref{eq:assce}).

Thus, we are in the position to apply the bound (\ref{eq:CE}) to the even/odd free energies of the matter sector. This bound will allow us, in particular, to estimate the energetic contribution of monopoles' removal in the matter sector. To this end, let use denote by $(\phi^{*}_{i,j})$ the chessboard flux configuration,
\begin{equation}
\phi^{*}_{i,j} = \left\{ \begin{array}{cc} 1 & \text{if $i$ and $j$ are odd} \\ -1 & \text{otherwise.} \end{array} \right.
\end{equation} 

\begin{figure}
\centering
\begin{tikzpicture}[scale=0.8]
\draw (0,0) grid (8,8);
\draw[->-] (0,8) -- (8,8);
\draw[->-] (0,0) -- (8,0);
\draw[->>-] (0,0) -- (0,8);
\draw[->>-] (8,0) -- (8,8);
\node (a) at (0.5,7.5) {\scalebox{1.2}{$0$}};
\begin{scope}[xshift = 2cm]
\node (a) at (0.5,7.5) {\scalebox{1.2}{$0$}};
\end{scope}
\begin{scope}[xshift = 4cm]
\node (a) at (0.5,7.5) {\scalebox{1.2}{$0$}};
\end{scope}
\begin{scope}[xshift = 6cm]
\node (a) at (0.5,7.5) {\scalebox{1.2}{$0$}};
\end{scope}
\begin{scope}[yshift = -2cm]
\node (a) at (0.5,7.5) {\scalebox{1.2}{$0$}};
\end{scope}
\begin{scope}[xshift = 2cm, yshift = -2cm]
\node (a) at (0.5,7.5) {\scalebox{1.2}{$0$}};
\end{scope}
\begin{scope}[xshift = 4cm, yshift = -2cm]
\node (a) at (0.5,7.5) {\scalebox{1.2}{$0$}};
\end{scope}
\begin{scope}[xshift = 6cm, yshift = -2cm]
\node (a) at (0.5,7.5) {\scalebox{1.2}{$0$}};
\end{scope}
\begin{scope}[yshift = -4cm]
\node (a) at (0.5,7.5) {\scalebox{1.2}{$0$}};
\end{scope}
\begin{scope}[xshift = 2cm, yshift = -4cm]
\node (a) at (0.5,7.5) {\scalebox{1.2}{$0$}};
\end{scope}
\begin{scope}[xshift = 4cm, yshift = -4cm]
\node (a) at (0.5,7.5) {\scalebox{1.2}{$0$}};
\end{scope}
\begin{scope}[xshift = 6cm, yshift = -4cm]
\node (a) at (0.5,7.5) {\scalebox{1.2}{$0$}};
\end{scope}
\begin{scope}[yshift = -6cm]
\node (a) at (0.5,7.5) {\scalebox{1.2}{$0$}};
\end{scope}
\begin{scope}[xshift = 2cm, yshift = -6cm]
\node (a) at (0.5,7.5) {\scalebox{1.2}{$0$}};
\end{scope}
\begin{scope}[xshift = 4cm, yshift = -6cm]
\node (a) at (0.5,7.5) {\scalebox{1.2}{$0$}};
\end{scope}
\begin{scope}[xshift = 6cm, yshift = -6cm]
\node (a) at (0.5,7.5) {\scalebox{1.2}{$0$}};
\end{scope}
\draw[fill = black, fill opacity = 0.6] (1,8)--(2,8) -- (2,0) -- (1,0) -- (1,8);
\draw[fill = black, fill opacity = 0.6] (3,8)--(4,8) -- (4,0) -- (3,0) -- (3,8);
\draw[fill = black, fill opacity = 0.6] (7,8)--(8,8) -- (8,0) -- (7,0) -- (7,8);
\draw[fill = black, fill opacity = 0.6] (5,8)--(6,8) -- (6,0) -- (5,0) -- (5,8);
\draw[fill = black, fill opacity = 0.6] (0,7)--(0,6) -- (1,6) -- (1,7) -- (0,7);
\draw[fill = black, fill opacity = 0.6] (2,7)--(2,6) -- (3,6) -- (3,7) -- (2,7);
\begin{scope}[xshift = 4cm]
\draw[fill = black, fill opacity = 0.6] (0,7)--(0,6) -- (1,6) -- (1,7) -- (0,7);
\end{scope}
\begin{scope}[xshift = 4cm, yshift = -2cm]
\draw[fill = black, fill opacity = 0.6] (0,7)--(0,6) -- (1,6) -- (1,7) -- (0,7);
\end{scope}
\begin{scope}[xshift = 2cm, yshift = -2cm]
\draw[fill = black, fill opacity = 0.6] (0,7)--(0,6) -- (1,6) -- (1,7) -- (0,7);
\end{scope}
\begin{scope}[yshift = -2cm]
\draw[fill = black, fill opacity = 0.6] (0,7)--(0,6) -- (1,6) -- (1,7) -- (0,7);
\end{scope}
\begin{scope}[xshift = 6cm]
\draw[fill = black, fill opacity = 0.6] (0,7)--(0,6) -- (1,6) -- (1,7) -- (0,7);
\end{scope}
\begin{scope}[xshift = 6cm, yshift=-2cm]
\draw[fill = black, fill opacity = 0.6] (0,7)--(0,6) -- (1,6) -- (1,7) -- (0,7);
\end{scope}
\begin{scope}[xshift = 6cm, yshift=-4cm]
\draw[fill = black, fill opacity = 0.6] (0,7)--(0,6) -- (1,6) -- (1,7) -- (0,7);
\end{scope}
\begin{scope}[xshift = 2cm, yshift = -4cm]
\draw[fill = black, fill opacity = 0.6] (0,7)--(0,6) -- (1,6) -- (1,7) -- (0,7);
\end{scope}
\begin{scope}[xshift = 4cm, yshift = -4cm]
\draw[fill = black, fill opacity = 0.6] (0,7)--(0,6) -- (1,6) -- (1,7) -- (0,7);
\end{scope}
\begin{scope}[yshift = -4cm]
\draw[fill = black, fill opacity = 0.6] (0,7)--(0,6) -- (1,6) -- (1,7) -- (0,7);
\end{scope}
\begin{scope}[xshift = 6cm, yshift=-6cm]
\draw[fill = black, fill opacity = 0.6] (0,7)--(0,6) -- (1,6) -- (1,7) -- (0,7);
\end{scope}
\begin{scope}[xshift = 2cm, yshift = -6cm]
\draw[fill = black, fill opacity = 0.6] (0,7)--(0,6) -- (1,6) -- (1,7) -- (0,7);
\end{scope}
\begin{scope}[xshift = 4cm, yshift = -6cm]
\draw[fill = black, fill opacity = 0.6] (0,7)--(0,6) -- (1,6) -- (1,7) -- (0,7);
\end{scope}
\begin{scope}[yshift = -6cm]
\draw[fill = black, fill opacity = 0.6] (0,7)--(0,6) -- (1,6) -- (1,7) -- (0,7);
\end{scope}
\end{tikzpicture}
\caption{Graphical representation of the flux configuration corresponding to $\phi_{i,j}^{*}$. The white plaquettes correspond to $\phi^{*}_{i,j} = 1$ while the dark plaquettes correspond to $\phi^{*}_{i,j} = -1$.}
\label{fig:conf}
\end{figure}
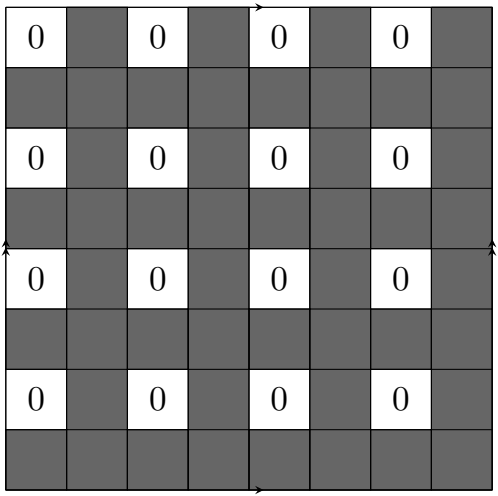

See Fig. \ref{fig:conf} for a graphical representation. The next proposition is a key consequence of the chessboard estimate.
\begin{proposition}[Removing monopoles]\label{prop:remove} Let $L = 4\ell$. Given any configuration $\Phi$ of magnetic fluxes, let $k \in \mathbb{N}$ be the number of plaquettes with magnetic flux equal to zero ({\it i.e.} $B_{\Lambda} = 1$). Then, the following bound holds true:
\begin{equation}\label{eq:mono}
- \frac{1}{\beta}\log Z^{\pm}_{\beta,L}(\Phi; a,b) \geq k \Delta^{\pm}_{\beta,L} - \frac{1}{\beta}\log Z^{\pm}_{*}
\end{equation}
with $Z^{\pm}_{*}$ the even/odd partition functions of the $\pi$-flux phase with boundary conditions $(a,b) = (-1,-1)$ and:
\begin{equation}\label{eq:delta}
\Delta^{\pm}_{\beta,L} = -\frac{1}{\beta L^{2}} \Big( \max_{a,b} \log Z^{\pm}_{\beta,L}(\Phi^{*}; a,b) - \log Z^{\pm}_{*} \Big)
\end{equation}
where $\Phi^{*} = \{\phi_{i,j}^{*}\}$ is the chessboard flux configuration (see Fig. \ref{fig:conf}).
\end{proposition}
\begin{remark} This bound is particularly useful, because the quantity $\Delta^{\pm}_{\beta,L}$ can be explicitly computed as $\beta \to \infty$, see Proposition \ref{prp:compD}. We will refer to $\Delta_{\infty} = \lim_{L\to \infty} \lim_{\beta \to \infty} \Delta^{\pm}_{\beta,L}$ as the monopole's mass (the limit is the same for even and odd partition functions).
\end{remark}
\begin{proof} The formula (\ref{eq:CE}) has been obtained via iterative reflections across vertical planes. We could have of course obtained a similar estimate, reflecting across horizontal planes. We do so for all entries of the sum in (\ref{eq:CE}). The final result is:
\begin{equation}\label{eq:CE2}
F(\underline{X}_{1}, \ldots, \underline{X}_{L}) \geq \frac{1}{(L/2)^{2}} \sum_{i,j = 1}^{L/2} F(\underline{Y}_{2j-1,2i-1}, \underline{-1}, \underline{Y}_{2j-1,2i-1}, \underline{-1}, \ldots, \underline{Y}_{2j-1,2i-1}, \underline{-1})\;,
\end{equation}
where we introduced the vectors:
\begin{equation}
\underline{Y}_{j,i}^{T} = ( (\underline{X}_{j})_{i}, -1, (\underline{X}_{j})_{i}, -1, \ldots, (\underline{X}_{j})_{i}, -1)\;.
\end{equation}
Graphically, the argument of $F$ in the right-hand side of (\ref{eq:CE2}) has a chessboard structure, as in Fig. \ref{fig:conf}, with $0$ replaced by $0,\pi$ depending on whether $(\underline{X}_{2j-1})_{2i-1} = 1,-1$, respectively. Observe that the right-hand side of (\ref{eq:CE2}) only depends on the values $(\underline{X}_{\ell})_{m}$  for odd $\ell$, $m$. This is due to the fact that, in the proof of Proposition \ref{lem:chess}, we decided to cut the system via planes that intersect plaquettes with even column label. Of course, there is nothing special about this choice; we could have decided to cut the system in correspondence with the odd columns. To take into account all possible choices, we write:
\begin{equation}
F(\underline{X}_{1}, \ldots, \underline{X}_{L}) = 4\times \frac{1}{4} F(\underline{X}_{1}, \ldots, \underline{X}_{L})
\end{equation}
and we apply the bound (\ref{eq:CE}) to the four functions $(1/4) F(\cdot)$, using cuts via orthogonal planes associated with even columns-even rows, even columns-odd rows, odd columns-even rows and odd columns-odd rows. The final result is, using cyclicity:
\begin{equation}\label{eq:chess}
F(\underline{X}_{1}, \ldots, \underline{X}_{L}) \geq \frac{1}{L^{2}} \sum_{i,j = 1}^{L} F(\underline{Y}_{i,j}, \underline{-1}, \underline{Y}_{i,j}, \underline{-1}, \ldots, \underline{Y}_{i,j}, \underline{-1})\;.
\end{equation}
Let us now use this estimate to prove (\ref{eq:mono}). Let us choose $F$ as in (\ref{eq:Fdef}). Denoting by $k$ the number of elements of $\Phi$ such that $\phi_{i,j} = 1$, Eq. (\ref{eq:chess}) implies:
\begin{equation}
\begin{split}
- \frac{1}{\beta} \max_{a,b} \log Z^{\pm}_{\beta,L}(\underline{\Phi}; a,b) &\geq -\frac{k}{\beta L^{2}} \max_{a,b} \log Z^{\pm}_{\beta,L}(\Phi^{*}; a,b)  - \frac{L^{2} - k}{\beta L^{2}} \log Z_{*}^{\pm} \\
&\geq -\frac{k}{\beta L^{2}} \Big(\max_{a,b} \log Z^{\pm}_{\beta,L}(\underline{\Phi}^{*}; a,b) - \log Z_{*}^{\pm}\Big)  - \frac{1}{\beta} \log Z_{*}^{\pm}\;,
\end{split}
\end{equation}
where $\Phi^{*}$ is the staggered magnetic flux configuration as in Fig. \ref{fig:conf}. This concludes the proof of the proposition.
\end{proof}
We conclude the section by computing the asymptotics of $\Delta_{\beta,L}^{\pm}$ as $L\to \infty, \beta \to \infty$.
\begin{proposition}[Asymptotics of the monopole's mass]\label{prp:compD} Let $L=4\ell$. For $\beta, \ell$ large enough,
\begin{equation}\label{eq:deltares}
\begin{split}
\Delta_{\beta,L}^{\pm} &= \frac{t}{4\pi^2} \int_0^{2\pi} \int_0^{2 \pi} dk_{1} dk_{2}\, \sqrt{1+\frac{1}{2} \cos(k_{1}) +\frac{1}{2}\cos(k_{2})} \\& \quad - \frac{t}{8\pi^2} \int_0^{2\pi} \int_0^{2 \pi} dk_{1} dk_{2}\, \sqrt{1+\frac{1}{2} \sqrt{1+\cos^2(k_{1}) +\cos^2(k_{2})}}\\
        &\quad - \frac{t}{8\pi^2} \int_0^{2\pi} \int_0^{2 \pi} dk_{1} dk_{2}\, \sqrt{1-\frac{1}{2} \sqrt{1+\cos^2(k_{1}) +\cos^2(k_{2})}} + o(1)
        \end{split}
\end{equation}
\end{proposition}
\begin{remark}
Calling $\Delta_{\infty}$ the sum of the three integrals in (\ref{eq:deltares}), numerical evaluation shows that:
\begin{equation}
\Delta_{\infty} \simeq 0.181 t\;.
\end{equation}
\end{remark}
\begin{proof} Let us consider the Hamiltonian associated with a gauge field configuration corresponding to the periodic flux configuration $\Phi^{*}$ depicted in Fig. \ref{fig:conf}. Consider the case $(a,b) = (1,1)$; the final result will not depend on the choice of $a,b$. A configuration associated with this flux arrangement is depicted in Fig. \ref{fig:conf2}, together with the corresponding  fundamental cell of the fermionic Hamiltonian. The fundamental cell contains eight lattice points, labelled as in Fig. \ref{fig:conf2}. We denote by $\widetilde{\Gamma}^{\text{red}}_{L}$ the lattice of the centers of the fundamental cells, corresponding with the position of the $a$ sites on the original lattice. We have:

\begin{figure}
\centering
\begin{tikzpicture}[scale=0.8]
\draw (0,0) grid (8,8);
\draw[->-] (0,8) -- (8,8);
\draw[->-] (0,0) -- (8,0);
\draw[->>-] (0,0) -- (0,8);
\draw[->>-] (8,0) -- (8,8);
\draw[line width = 0.07 cm] (0,0) -- (8,0);
\draw[line width = 0.07 cm] (0,2) -- (8,2);
\draw[line width = 0.07 cm] (0,4) -- (8,4);
\draw[line width = 0.07 cm] (0,6) -- (8,6);
\draw[line width = 0.07 cm] (0,8) -- (8,8);
\draw[line width = 0.07 cm] (1,8) -- (1,7);
\draw[line width = 0.07 cm] (2,8) -- (2,7);
\draw[line width = 0.07 cm] (5,8) -- (5,7);
\draw[line width = 0.07 cm] (6,8) -- (6,7);
\draw[line width = 0.07 cm] (1,6) -- (1,5);
\draw[line width = 0.07 cm] (2,6) -- (2,5);
\draw[line width = 0.07 cm] (5,6) -- (5,5);
\draw[line width = 0.07 cm] (6,6) -- (6,5);
\draw[line width = 0.07 cm] (1,4) -- (1,3);
\draw[line width = 0.07 cm] (2,4) -- (2,3);
\draw[line width = 0.07 cm] (5,4) -- (5,3);
\draw[line width = 0.07 cm] (6,4) -- (6,3);
\draw[line width = 0.07 cm] (1,2) -- (1,1);
\draw[line width = 0.07 cm] (2,2) -- (2,1);
\draw[line width = 0.07 cm] (5,2) -- (5,1);
\draw[line width = 0.07 cm] (6,2) -- (6,1);
\draw[line width = 0.07 cm] (5,8) -- (5,7);
\draw[line width = 0.07 cm] (6,8) -- (6,7);
\draw[dotted, thick] (3.5,0) -- (3.5,8);
\draw[dotted, thick] (7.5,0) -- (7.5,8);
\draw[dotted, thick] (0,0.5) -- (8,0.5);
\draw[dotted, thick] (0,2.5) -- (8,2.5);
\draw[dotted, thick] (0,0.5) -- (8,0.5);
\draw[dotted, thick] (0,2.5) -- (8,2.5);
\draw[dotted, thick] (0,4.5) -- (8,4.5);
\draw[dotted, thick] (0,6.5) -- (8,6.5);
\draw[dotted, fill = gray, fill opacity = 0.1] (3.5,2.5)--(7.5,2.5) -- (7.5,4.5) -- (3.5,4.5) -- (3.5,2.5);
\begin{scope}[xshift =10 cm, yshift=2cm]
\draw[dotted,fill = gray, fill opacity = 0.1] (0,0) -- (8,0) -- (8,4)--(0,4)--(0,0);
\draw[thick] (1,0) -- (1,4);
\draw[thick] (3,0) -- (3,4);
\draw[thick] (5,0) -- (5,4);
\draw[thick] (7,0) -- (7,4);
\draw[thick] (0,1) -- (8,1);
\draw[thick, line width= 0.07 cm] (3,1) -- (3,3);
\draw[thick, line width= 0.07 cm] (5,1) -- (5,3);
\draw[thick, line width = 0.07 cm] (0,3) -- (8,3);
\node[below left] (a) at (1,1) {\scalebox{0.8}{$a$}};
\node[above left] (b) at (1,3) {\scalebox{0.8}{$A$}};
\node[below left] (a) at (3,1) {\scalebox{0.8}{$b$}};
\node[above left] (b) at (3,3) {\scalebox{0.8}{$B$}};
\node[below right] (a) at (5,1) {\scalebox{0.8}{$c$}};
\node[above right] (b) at (5,3) {\scalebox{0.8}{$C$}};
\node[below right] (a) at (7,1) {\scalebox{0.8}{$d$}};
\node[above right] (b) at (7,3) {\scalebox{0.8}{$D$}};
\draw[line width = 0.04 cm, ->] (0,0)--(0,4);
\draw[line width = 0.04 cm, ->] (0,0)--(8,0);
\end{scope}
\end{tikzpicture}
\caption{Left: gauge field configuration associated with the chessboard flux arrangement $\Phi^{*}$ in Fig. \ref{fig:conf}. Solid bonds correspond to spin $-1$, while light bonds correspond to spin $+1$. Right: fundamental cell associated with the periodic spin configuration.}
\label{fig:conf2}
\end{figure}
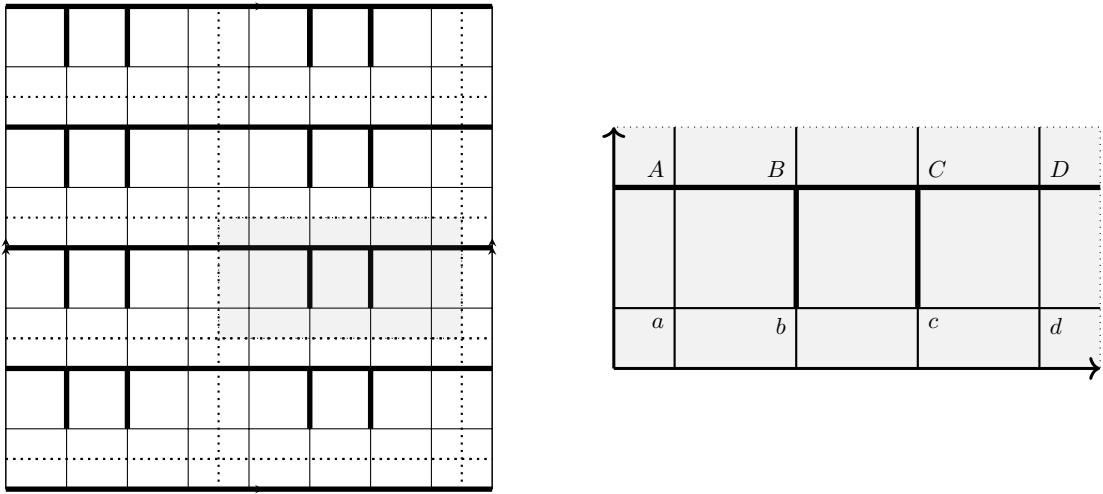
\begin{equation}
\begin{split}
&        H_{f}(\Phi^{*}; 1,1) \\
&\qquad = -t \sum_{x \in \widetilde{\Gamma}^{\text{red}}_{L}} (a^+_{a, x} a^-_{b, x} + a^+_{a, x} a^-_{A, x} + a^+_{b, x} a^-_{c, x} - a^+_{b, x} a^-_{B, x} + a^+_{c, x} a^-_{d, x} - a^+_{c, x} a^-_{C, x} \\&\qquad \quad + a^+_{d, x} a^-_{D, x}  + a^+_{d, x} a^-_{a, x + 4 e_{1} }-a^+_{A, x} a^-_{B, x} + a^+_{A, x} a^-_{a, x+2 e_{2}} - a^+_{B, x} a^-_{C, x} \\&\qquad \quad + a^+_{B, x} a^-_{b, x+2 e_{2}}  - a^+_{C, x} a^-_{D, x} + a^+_{C, x} a^-_{c, x+2 e_{2}} - a^+_{D, x} a^-_{A, x+4 e_{1}} + a^+_{D, x} a^-_{d, x+2 e_{2}}) + \text{h.c.}\;.
\end{split}        
\end{equation}
The associated Bloch Hamiltonian $h(k)$ is:
\begin{equation}\nonumber
\begin{split}
&h(k) = \\
&\scriptsize{      \begin{pmatrix}
           0 &1 &0 &e^{-4ik_1} &1+e^{-2ik_2} &0 &0 &0\\
           1 &0 &1 &0 &0 &-1+e^{-2ik_2} &0 &0\\
           0 &1 &0 &1 &0 &0 &-1+e^{-2ik_2} &0\\
           e^{4 i k_1} &0 &1 &0 &0 &0 &0 &1+e^{-2ik_2} \\
           1+e^{2ik_2} &0 &0 &0 &0 &-1 &0 &-e^{-4ik_1}\\
           0 &-1+e^{2ik_2} &0 &0 &-1 &0 &-1 &0\\
           0 &0 &-1 + e^{2ik_2} &0 &0 &-1 &0 &-1\\
           0 &0 &0 &1+e^{2ik_2} &-e^{4ik_1} &0 &-1 &0
       \end{pmatrix} }\;.
       \end{split}
\end{equation}
Its eigenvalues are doubly degenerate, and are given by:
\begin{equation}\label{eq:deltabands}
\begin{split}
       e_{1;\pm}(k) &= \pm 2t \sqrt{1+\frac{1}{2} \sqrt{1 + \cos(2k_1)^2 +\cos(2k_2)^2}}\\
        e_{2;\pm}(k) &= \pm 2t \sqrt{1-\frac{1}{2}\sqrt{1 + \cos(2k_1)^2 +\cos(2k_2)^2}}\;.
        \end{split}
\end{equation}
The quasi-momentum $k = (k_{1}, k_{2})$ is a point on the Brillouin zone,
\begin{equation}
\widetilde{B}_{L}(1,1) := \Big\{ k\in \frac{2\pi}{L}(n_{1}, n_{2})\, \mid\, 0\leq n_{1} \leq L/4-1,\; 0\leq n_{2} \leq L/2-1 \Big\}\;.
\end{equation}
More generally, for any $(a,b) \in \mathbb{Z}_{2}^{2}$, the energy bands are given by (\ref{eq:deltabands}), evaluated over the Brillouin zone:
\begin{equation}\label{eq:bstag}
\widetilde{B}_{L}(a,b) = \widetilde{B}_{L}(1,1) - \frac{\pi}{2L}((a-1), (b-1))\;.
\end{equation}
We are now ready to compute $\Delta^{\pm}_{\beta,L}$, as given by (\ref{eq:delta}). Let us denote by $e^{\pi}_{\pm}(k)$ the eigenvalues of the Bloch Hamiltonian associated with the $\pi$-flux phase, as given by (\ref{eq:epi}). In order to perform the computation, it is convenient to adapt the energy bands (\ref{eq:deltabands}) to the Brillouin zone of the $\pi$-flux phase, Eq. (\ref{eq:bab}), using that the functions in (\ref{eq:deltabands}) are periodic with period $\pi/2$ in the variable $k_{1}$. We get:
\begin{equation}
\begin{split}
&-\frac{1}{\beta L^{2}} \log \frac{Z^{\pm}_{\beta,L}(\Phi^{*}; a,b)}{Z^{\pm}_{*}} \\
&\qquad = \frac{1}{\beta L^{2}} \Big( \sum_{k\in B_{L}(-1,-1)} \log \prod_{z=\pm} (1 \pm e^{-\beta e^{\pi}_{z}(k)}) \\&\color{black}{\qquad\quad - \frac{1}{4}\sum_{k\in B_{L}(a,b)}} \log \prod_{x = 1,2} \prod_{y = \pm}(1 \pm e^{-\beta e_{x;y}(k)})^{2}\Big)\;.
\end{split}
\end{equation}
where the factor $1/4$ in the right-hand side is introduced in order to avoid overcounting, due to the fact that every quasi-momentum $k$ in (\ref{eq:bstag}) is now counted $4$ times. As $\beta, L \to \infty$ the expression in the right-hand side can be approximated by an integral, independent of the choice of the parity of the partition function, as follows:
\begin{equation}
\begin{split}
&-\frac{1}{\beta L^{2}} \log \frac{Z^{\pm}_{\beta,L}(\Phi^{*}; a,b)}{Z^{\pm}_{*}} \\
&\qquad = \frac{t}{4\pi^2} \int_0^{2\pi} \int_0^{2 \pi} dk_{1} dk_{2}\, \sqrt{1+\frac{1}{2} \cos(k_{1}) +\frac{1}{2}\cos(k_{2})} \\&\qquad \quad - \frac{t}{8\pi^2} \int_0^{2\pi} \int_0^{2 \pi} dk_{1} dk_{2}\, \sqrt{1+\frac{1}{2} \sqrt{1+\cos^2(k_{1}) +\cos^2(k_{2})}}\\
        &\qquad\quad - \frac{t}{8\pi^2} \int_0^{2\pi} \int_0^{2 \pi} dk_{1} dk_{2}\, \sqrt{1-\frac{1}{2} \sqrt{1+\cos^2(k_{1}) +\cos^2(k_{2})}} + o(1)\;.
\end{split}
\end{equation}
This concludes the proof of the proposition.
\end{proof}
\section{Proof of Theorem \ref{thm:main}}\label{sec:proofmain}
Let us start by proving the statement about the free energy, Eq. (\ref{eq:fen}). The lower bound simply follows from:
\begin{equation}
Z_{\beta,L} = \sum_{\Phi} \sum_{a,b} Z_{\beta,L}(\Phi, a, b) \geq \sum_{a,b} Z_{\beta,L}(\pi, a, b)\;.
\end{equation}
To prove the upper bound, we proceed as follows. We have:
\begin{equation}\label{eq:upper}
\begin{split}
&\beta L^{2}(f_{\beta, L}(\pi) - f_{\beta, L} ) \\
&\quad = \log\bigg(\sum_{k=0}^{L^2/2}  \frac{\sum_{a,b \in \mathbb{Z}_{2}^{2}}\sum_{n(\Phi) = 2k} e^{-\beta H_{g} (\Phi)} (Z^+_{\beta, L}(\Phi, a, b) + Z^-_{\beta, L}(\Phi, a, b))}{\sum_{a,b \in \mathbb{Z}_{2}^{2}} e^{-\beta H_{g} (\bm \pi)} (Z^+_{\beta, L}(\pi,a,b)+Z^-_{\beta, L}(\pi, a,b))}\bigg)
\end{split}
\end{equation}
where $n(\Phi)$ is the number of $0$-fluxes in the flux configuration $\Phi$; observe that this number is even, due to the constraint $\prod_{\Lambda \in P(\Gamma_{L})} B_{\Lambda} = 1$, and to the fact that the number of elementary plaquettes is even for $L$ even. By Proposition \ref{prop:remove}, we have:
\begin{equation}\label{eq:upbd}
\begin{split}
Z^+_{\beta, L}(\Phi, a, b) + Z^-_{\beta, L}(\Phi, a, b) &\leq e^{-\beta 2k \Delta^{+}_{\beta,L}} Z^{+}_{*} +  e^{-\beta 2k \Delta^{-}_{\beta,L}} Z^{-}_{*} \\
&\leq e^{-\beta 2k \Delta_{\beta,L}} (Z^{+}_{*} +  Z^{-}_{*})\;.
\end{split}
\end{equation}
with $\Delta_{\beta,L}$ the minimum between $\Delta^{+}_{\beta,L}$ and $\Delta^{-}_{\beta,L}$, and we used that $n(\Phi) = 2k$. Furthermore (recall the definition (\ref{eq:Hgaugeg})), using that:
\begin{equation}
H_{g}(\Phi) = H_{g}({\bm \pi})  - 2k = |P(\Gamma_{L})| - 2k\;,
\end{equation}
and the positivity of the odd partition function, we have:
\begin{equation}
\begin{split}
(\beta L^{2})(f_{\beta, L} - f_{\beta, L}(\pi)) &\leq \log\bigg(4\sum_{k=0}^{L^2/2} {L^{2} \choose 2k} e^{-\beta 2k (\Delta_{\beta,L} - 1)}\bigg) \\
&= \log\Big( ( 1 + e^{-\beta (\Delta_{\beta,L} - 1)} )^{L^{2}} + (1 - e^{-\beta (\Delta_{\beta,L} - 1)})^{L^{2}}\Big) + \log 2\;.
\end{split}
\end{equation}
where the inequality follows by dropping at the denominator all but the largest partition functions, and the factor $4$ arises from the sum over boundary conditions at the numerator. This proves the claim (\ref{eq:fen}). Let us now prove the statement about expectation values, Eq. (\ref{eq:main}). We rewrite the expectation value of the gauge-invariant observable $\mathcal{O}(\bm{\sigma})$ as:
\begin{equation}\label{eq:expO}
\begin{split}
\langle \mathcal{O} \rangle_{\beta, L} &= \frac{1}{Z_{\beta,L}} \sum_{\bm{\sigma}} \Tr_{\mathcal{F}_{L}}\, \Big(\frac{1 + (-1)^{N}}{2^{|\Gamma|}}\Big) \mathcal{O}(\bm{\sigma}) e^{-\beta H(\bm{\sigma})} \\
&= \frac{1}{\widetilde Z_{\beta,L}} \sum_{a,b} \sum_{\Phi}  \Tr_{\mathcal{F}_{L}}\, (1 + (-1)^{N}) \mathcal{O}(\Phi;a,b) e^{-\beta H(\Phi;a,b)}
\end{split}
\end{equation}
where: $\Phi = (\phi_{i,j})$ is the collection of fluxes in the lattice plaquettes; $\mathcal{O}(\Phi; a,b)$ is the observable evaluated over a given gauge field representative of the flux configuration $\Phi, a, b$; and:
\begin{equation}\label{eq:tildeZ}
\widetilde Z_{\beta,L} = \sum_{(a,b)} \sum_{\Phi}  \Tr_{\mathcal{F}_{L}}\, (1 + (-1)^{N}) e^{-\beta H(\Phi;a,b)}\;.
\end{equation}
Let us consider the partition function (\ref{eq:tildeZ}). From (\ref{eq:upbd}), we obtain the estimate:
\begin{equation}
\frac{\sum_{(a,b)} \sum^{*}_{\Phi}  \Tr_{\mathcal{F}_{L}}\, (1 + (-1)^{N}) e^{-\beta H(\Phi; a,b)}}{ \sum_{(a,b)} \Tr_{\mathcal{F}_{L}}\, (1 + (-1)^{N}) e^{-\beta H(\pi; a,b)} } \leq 4((1 + e^{- \beta (\Delta_{\beta,L}-1)})^{L^{2}} - 1)\;,
\end{equation}
where the asterisk denotes the constraint $n(\Phi) \geq 2$. This estimate allows to prove that:
\begin{equation}\label{eq:partpi}
\widetilde Z_{\beta,L} = \Big(\sum_{(a,b)} \Tr_{\mathcal{F}_{L}}\, (1 + (-1)^{N}) e^{-\beta H(\pi; a,b)}\Big) (1 + \mathcal{E}_{\beta,L})\;,\quad 0\leq \mathcal{E}_{\beta,L} \leq 4((1 + e^{- \beta (\Delta_{\beta,L}-1)})^{L^{2}} - 1)\;.
\end{equation}
Consider now the numerator in (\ref{eq:expO}). We rewrite it as:
\begin{equation}\label{eq:num}
\begin{split}
&\sum_{a,b} \sum_{\Phi}  \Tr_{\mathcal{F}_{L}}\, (1 + (-1)^{N}) \mathcal{O}(\Phi;a,b) e^{-\beta H(\Phi;a,b)} \\
&\qquad = \sum_{a,b} \Tr_{\mathcal{F}_{L}}\, (1 + (-1)^{N}) \mathcal{O}(\pi;a,b) e^{-\beta H(\pi;a,b)} \\
&\quad \qquad + \sum_{a,b} \sum_{\Phi}^{*}  \Tr_{\mathcal{F}_{L}}\, (1 + (-1)^{N}) \mathcal{O}(\Phi;a,b) e^{-\beta H(\Phi;a,b)} \\
&\qquad \equiv \sum_{a,b} \Tr_{\mathcal{F}_{L}}\, (1 + (-1)^{N}) \mathcal{O}(\pi;a,b) e^{-\beta H(\pi;a,b)} + \mathcal{E}_{\beta,L}^{\mathcal{O}}\;.
\end{split}
\end{equation}
Proceeding as for the partition function, we estimate the last term in (\ref{eq:num}) as:
\begin{equation}\label{eq:estE}
\begin{split}
\Big| \mathcal{E}_{\beta,L}^{\mathcal{O}} \Big| &\leq \sum_{a,b} \sum_{\Phi}^{*} \| \mathcal{O}(\Phi; a,b) \|  \Tr_{\mathcal{F}_{L}}\, (1 + (-1)^{N}) e^{-\beta H(\Phi;a,b)} \Big| \\
&\leq C_{\mathcal{O}} ((1 + e^{- \beta (\Delta_{\beta,L}-1)})^{L^{2}} - 1) e^{-\beta H_{g}(\pi)} (Z_{*}^{+} + Z_{*}^{-})\;.
\end{split}
\end{equation}
Therefore, coming back to (\ref{eq:expO}):
\begin{equation}\label{eq:oss}
\begin{split}
\langle \mathcal{O} \rangle_{\beta, L} &= \frac{1}{1 + \mathcal{E}_{\beta,L}} \cdot \frac{1}{\sum_{a,b} \Tr_{\mathcal{F}_{L}}\, (1 + (-1)^{N}) e^{-\beta H(\pi, a,b)}} \\
&\quad \cdot \Big(\sum_{a,b} \Tr_{\mathcal{F}_{L}}\, (1 + (-1)^{N}) \mathcal{O}(\pi;a,b) e^{-\beta H(\pi;a,b)} + \mathcal{E}_{\beta,L}^{\mathcal{O}}\Big)\;;
\end{split}
\end{equation} 
using the estimate (\ref{eq:estE}), we get:
\begin{equation}\label{eq:errmain}
\begin{split}
&\frac{\mathcal{E}_{\beta,L}^{\mathcal{O}}}{\sum_{a,b} \Tr_{\mathcal{F}_{L}}\, (1 + (-1)^{N}) e^{-\beta H(\pi; a,b)}}  \\
&\leq C_{\mathcal{O}} ((1 + e^{- \beta (\Delta_{\beta,L}-1)})^{L^{2}} - 1) \frac{Z_{*}^{+} + Z_{*}^{-}}{ \sum_{a,b} (Z^{+}_{\beta,L}(\pi;a,b) + Z^{-}_{\beta,L}(\pi;a,b))} \\
&\leq C_{\mathcal{O}} ((1 + e^{- \beta (\Delta_{\beta,L}-1)})^{L^{2}} - 1)\;.
\end{split}
\end{equation}
Combining the estimates (\ref{eq:errmain}), (\ref{eq:partpi}) with (\ref{eq:oss}), we have:
\begin{equation}
\langle \mathcal{O} \rangle_{\beta, L} = \langle \mathcal{O} \rangle^{\pi}_{\beta, L} + \mathcal{E}^{\text{tot}}_{\mathcal{O}; \beta,L} 
\end{equation}
with:
\begin{equation}
\begin{split}
|\mathcal{E}^{\text{tot}}_{\mathcal{O}; \beta,L}|  &\leq \frac{\mathcal{E}_{\beta,L}}{1 + \mathcal{E}_{\beta,L}} \max_{a,b} \| \mathcal{O}(\pi, a,b) \| + \frac{C_{\mathcal{O}} ((1 + e^{- \beta (\Delta_{\beta,L}-1)})^{L^{2}} - 1)}{1 + \mathcal{E}_{\beta,L}} \\
&\leq K_{\mathcal{O}} ((1 + e^{- \beta (\Delta_{\beta,L}-1)})^{L^{2}} - 1)\;.
\end{split}
\end{equation}
The energetic closeness of the $\pi$-flux phases, Eq. (\ref{eq:deg}), has been proved in Proposition \ref{prp:deg}. This concludes the proof of Theorem \ref{thm:main}. \qed
\section{Proof of Proposition \ref{prp:susc}}\label{sec:proofsusc}
\subsection{Rewriting of the susceptibility}
In order to compute the main term in the right-hand side of (\ref{eq1}) in the $L,\beta \to \infty$ limit, we shall rely on lattice conservation laws and Ward identities. Let us introduce the density operator in Fourier space:
\begin{equation}
\hat n(p) = \sum_{x\in \Gamma_{L}^{\text{red}}} \sum_{i\in I}  e^{-i p \cdot (x+r_i)} n_{x,i}\;,\qquad n_{x,i} = a^{+}_{x,i} a^{-}_{x,i}\;.
\end{equation}
Let us derive the lattice continuity equation for $\hat n(p)$, in imaginary time. We have\footnote{Recall the identity:
\begin{equation*}
    [a^{+}_{x,i}a_{x',j}, a^{+}_{y,k} a^{-}_{y,k}] = a^{+}_{x,i} a^{-}_{y,k} \delta_{x',y}\delta_{j,k} - a^{+}_{y,k} a^{-}_{x',k} \delta_{x,y}\delta_{i,k}\;.
\end{equation*}
}:
\begin{equation}\label{eq:dern}
\begin{split}
&-i q \partial_s \hat n(p,s) = -iq [H, \hat n_{p}](s) \\
& = -iq \sum_{x,x',y \in \Gamma_{L}^{\text{red}}} \sum_{i,j,k \in I} e^{-i p \cdot (y+r_k)} h_{ij}(x,x') [a^{+}_{x,i} a^{-}_{x', j}, n_{y,k}](s)\\
& = -iq \sum_{x,x' \in \Gamma_{L}^{\text{red}}} \sum_{i,j \in I} (e^{-i p \cdot (x'+r_i)}-e^{-i p \cdot (x+r_j)}) h_{ij}(x,x') (a^{+}_{x,i} a^{-}_{x', j})(s)\\ 
& = \frac{i q}{2} \sum_{x,x' \in \Gamma_{L}^{\text{red}}} \sum_{i,j \in I} (h_{ij}(x,x') a^{+}_{x,i} a^{-}_{x', j} - h_{ji}(x',x) a^{+}_{x',j} a^{-}_{x, i} )(s)  (e^{-i p \cdot (x+ r_i)}-e^{-i p \cdot (x'+ r_j)})\\       
& =-\frac{i q}{2} \sum_{x,x' \in \Gamma_{L}^{\text{red}}} \sum_{i,j \in I} \sum_{\mu = 1,2}e^{-i p \cdot (x'+ r_j)} (h_{ij}(x,x') a^{+}_{x,i} a^{-}_{x', j} - h_{ji}(x',x) a^{+}_{x',j} a^{-}_{x, i} )(s)  \eta_{\xi}(p) (-i p_{\mu} \xi_{\mu})
\end{split}
\end{equation}%

recall the definition of the function $\eta_{\xi}(p)$ in (\ref{eq:etadef}). Eq. (\ref{eq:dern}) can be recasted as, recall (\ref{eq:JK}):
\begin{equation}\label{eq:WI}
iq \partial_s \hat n(p,s)+\sum_{\mu = 1,2} i p_\mu \hat J_{\mu}(p,s) = 0\;.     
\end{equation}
Therefore,
\begin{equation}
    i q \partial_s \left \langle \hat n(-p,s); \hat J_1(p) \right \rangle_{\beta, L} = i \sum_{\mu=1,2} p_{\mu} 
    \left \langle \hat J_{\mu}(-p,s); \hat J_1(p) \right \rangle_{\beta, L}\;. 
\end{equation}
Let us use the identity (\ref{eq:WI}) to rewrite in a more convenient way the right-hand side of (\ref{eq:susclim}). Choose $p = (p_{1}, 0)$. We have:
\begin{equation}
\begin{split}
\int_0^{\beta} ds \left \langle  \hat J_1(-p,s); \hat J_1(p) \right \rangle_{\beta, L} &= \frac{q}{p_{1}}\int_0^{\beta} ds\,  \partial_s \left \langle \hat n(-p,s); \hat J_1(p) \right \rangle_{\beta, L} \\
&= -\frac{q}{p_{1}} \left \langle [ \hat n(-p), \hat J_1(p)] \right \rangle_{\beta, L}\;,
\end{split}
\end{equation}
where the last step follows from the KMS identity\footnote{It is easy to check that the KMS identity holds true for the Gibbs state (\ref{eq:gibbs}), for physical observables.}. The commutator is (for general values of $p$):
\begin{equation}
\begin{split}
&iq  [\hat n(-p), \hat J_1(p)] \\
& =\frac{q^2}{2} \sum_{x,x',y\in \Gamma^{\text{red}}_{L}} \sum_{i,j,k\in I}  e^{i p \cdot (y+r_k-(x'+r_j))} \Big(h_{ij}(x,x') [n_{y,k},a^{+}_{x,i} a^{-}_{x', j}]\\&\qquad -h_{ji}(x',x) [n_{y,k},a^{+}_{x',j} a^{-}_{x, i}] \Big) \eta_{\xi}(p) \xi_1\\
& =\frac{q^2}{2} \sum_{x,x'} \sum_{i,j}  (e^{i p \cdot (x+r_k-(x'+r_j))} -1) \Big(h_{ij}(x,x') a^{+}_{x,i} a^{-}_{x', j} +h_{ji}(x',x) a^{+}_{x',j} a^{-}_{x, i} \Big) \eta_\xi(p) \xi_1\\
& =  \sum_{\mu=1,2} i p_{\mu} \hat K_{1, \mu}(p,-p)\;,   
\end{split}
\end{equation}
recall Eq. (\ref{eq:JK}). Thus, for $p = (p_{1}, 0)$:
\begin{equation}\label{eq:K11}
\left\langle \hat K_{1,1}(p, -p) \right\rangle_{\beta,L} = \frac{q}{p_{1}} \left\langle [\hat n_{-p}, \hat J_1(p)] \right\rangle_{\beta,L} = -\int_0^{\beta} ds \left \langle  \hat J_1(-p,s); \hat J_1(p) \right \rangle_{\beta, L}\;.
\end{equation}
We are interested in the choice of external momenta $p = (0, p_{2})$. To begin, observe that, for $p = (0,p_{2})$, from Eq. (\ref{eq:JK}):
\begin{equation}\label{eq:KK}
\left \langle \hat K_{1,1}(p, -p) \right \rangle_{\beta, L} = \left \langle \hat K_{1,1}(0, 0) \right \rangle_{\beta, L}\;.
\end{equation}
As we will see in the next section, the $\beta, L\to \infty$ limit of the left-hand side of (\ref{eq:K11}) is continuous in $p$; thus, we have:
\begin{equation}\label{eq:lims}
\begin{split}
&
\lim_{L\to \infty} \lim_{\beta \to \infty} \left\langle \hat K_{1,1}((0, p_{2}), (0, -p_{2})) \right\rangle_{\beta,L} \\
&\quad = \lim_{L\to \infty} \lim_{\beta \to \infty} \left\langle \hat K_{1,1}((0, 0), (0, 0)) \right\rangle_{\beta,L} \\
&\quad = -\lim_{p_{1}\to 0} \lim_{L\to \infty} \lim_{\beta \to \infty}  \int_0^{\beta} ds \left \langle  \hat J_1((-p_{1},0),s); \hat J_1((p_{1},0)) \right \rangle_{\beta, L}\;.
\end{split}
\end{equation}
Therefore, we can express the magnetic susceptibility (\ref{eq:chifin}) as:
\begin{equation}\label{eq:suscfin}
\begin{split}
&\chi(p_{2}) \\
& = \frac{1}{p_{2}^{2}} \lim_{L\to \infty} \lim_{\beta \to \infty} \frac{1}{L^{2}}\int_0^{\beta} ds\, \Big( \left \langle  \hat J_1((0,-p_{2}),s); \hat J_1((0,p_{2})) \right \rangle_{\beta,L} - \left \langle  \hat J_1((0,0),s); \hat J_1((0,0)) \right \rangle_{\beta,L}\Big)\;.
\end{split}
\end{equation}
To obtain Eq. (\ref{eq:suscfin}), we implicitly assumed that in the last step of (\ref{eq:lims}) we can first set $p_{1} = 0$ in the current-current correlation, and then take $\beta, L\to \infty$, which will be shown later. The rewriting (\ref{eq:suscfin})  will be a convenient starting point for the explicit evaluation of the susceptibility. 
\subsection{Evaluation of $\chi(p_{2})$}
In order to evaluate the susceptibility, as given by (\ref{eq:suscfin}), we shall use Theorem \ref{thm:main}. We have:
\begin{equation}
\left \langle  \hat J_1(-p,s); \hat J_1(p) \right \rangle_{\beta,L} = \left \langle  \hat J_1(-p,s); \hat J_1(p) \right \rangle^{\pi}_{\beta,L} + \frak{e}_{\beta,L}
\end{equation}
where: $\frak{e}_{\beta,L}$ is an error term that vanishes as $L,\beta \to \infty$; the average over the $\pi$-flux phases can be written as:
\begin{equation}\label{eq:sum}
\left \langle  \hat J_1(-p,s); \hat J_1(p) \right \rangle^{\pi}_{\beta,L} = \sum_{\varepsilon = \pm} \sum_{(a,b) \in \mathbb{Z}_{2}^{2}} \left \langle  \hat J_1(-p,s); \hat J_1(p) \right \rangle_{\varepsilon,a,b} z_{\varepsilon,a,b}
\end{equation}
with, for $(\varepsilon, a, b) \in \mathbb{Z}_{2}^{3}$:
\begin{equation}
\begin{split}
\left \langle  \mathcal{O} \right \rangle_{\varepsilon,a,b} &:= \frac{\Tr_{\mathcal{F}_{L}} e^{-\beta H(\pi; a,b)} \mathcal{O}(\pi; a,b) (\varepsilon)^{N}}{ \Tr_{\mathcal{F}_{L}} e^{-\beta H(\pi; a,b)} (\varepsilon)^{N} } \\
z_{\varepsilon,a,b} &:= \frac{\Tr_{\mathcal{F}_{L}} e^{-\beta H(\pi; a,b)} (\varepsilon)^{N}}{\sum_{(a,b) \in \mathbb{Z}_{2}^{2}}\Tr_{\mathcal{F}_{L}} e^{-\beta H(\pi; a,b)} (1 + (-1)^{N})}\;.
\end{split}
\end{equation}
As we will see, taking the limit $\beta \to \infty$ and then the limit $L\to \infty$, all $\left \langle  \mathcal{O} \right \rangle_{+,a,b}$ converge to a value independent of $a,b$. In the same order of limits, also $\left \langle  \mathcal{O} \right \rangle_{-,a,b}$ converge to this value, provided no zero mode is present in the spectrum of the single-particle Hamiltonian. Instead, for choices of $a,b$ for which a zero mode is present, $\left \langle  \mathcal{O} \right \rangle_{-,a,b} z_{-,a,b} = 0$, provided the order in the fermionic creation/annihilation operators appearing in $\mathcal{O}$ is not too high. This is a consequence of the next lemma.
\begin{lemma}\label{lem:zero}
Let $\mathcal{O}$ be an operator in the fermionic algebra, and let $H$ be the second quantization of a Hamiltonian $h$. Suppose that $\mathcal{O}$ has the following form:
\begin{equation}\label{eq:lemO}
\mathcal{O} = \sum_{\substack{ i_{1}, \ldots, i_{r} \\ j_{1}, \ldots, j_{r}}} O\big(\underline{i}; \underline{j}\big) a^{+}_{i_{1}} \cdots a^{+}_{i_{r}} a^{-}_{j_{r}} \cdots a^{-}_{j_{1}}
\end{equation}
where $a_{i}$ and $a^{+}_{i}$, $i=1,\ldots, L^{2}$, are the fermionic annihilation and creation operators expressed in the basis of $h$. Suppose that $h$ has $k> 0$ zero modes, and suppose that $k>r$. Then,
\begin{equation}
\Tr\, (-1)^{N} e^{-\beta H} \mathcal{O} = 0\;.
\end{equation}
\end{lemma} 
\begin{proof}
Let us rewrite the trace as:
\begin{equation}
\sum_{\{q_{i}\}} \sum_{\{p_{j}\}} (-1)^{\sum_{i} q_{i} + \sum_{j} p_{j}} e^{-\beta \sum_{i} \varepsilon_{i} q_{i}} \langle \psi(\{q_{i}\}, \{p_{j}\}),  \mathcal{O} \psi(\{q_{i}\}, \{p_{j}\})\rangle\;,
\end{equation}
where: $\{\varepsilon_{i}\}$ are the eigenvalues of $h$; $\{q_{i}\}$ are the occupation numbers of the non-zero modes of $h$, and $\{p_{j}\}$ are the occupation numbers of the zero modes of $h$; and $\langle \psi(\{q_{i}\}, \{p_{j}\})$ is the corresponding eigenstate of $H$ (a Slater determinant). We claim that, if the number of zero modes $k$ is larger that $r$:
\begin{equation}
\sum_{\{p_{j}\}} (-1)^{\sum_{j} p_{j}} \langle \psi(\{q_{i}\}, \{p_{j}\}),  \mathcal{O} \psi(\{q_{i}\}, \{p_{j}\})\rangle = 0\;.
\end{equation}
To see this, it is enough to prove the statement for $\mathcal{O} = n_{b_{1}} \cdots n_{b_{r}}$, with $n_{b} = a^{+}_{b} a^{-}_{b}$; the general claim follows by linearity. Let $B = \{b_{1}, \ldots, b_{r}\}$. We have:
\begin{equation}
\begin{split}
&\sum_{\{p_{j}\}} (-1)^{\sum_{j} p_{j}} \langle \psi(\{q_{i}\}, \{p_{j}\}),  \mathcal{O} \psi(\{q_{i}\}, \{p_{j}\})\rangle \\
&\quad = \sum_{\{p_{j}\}} (-1)^{\sum_{j} p_{j}} \Big[\prod_{\ell\in B} q_{\ell}\Big] \Big[ \prod_{r\in B} p_{r} \Big] \\
&\quad = \Big(\sum_{\{p_{j}\}: j \notin B} (-1)^{\sum_{j \notin B} p_{j}}\Big)  \Big(\sum_{\{p_{j}\}: j \in B}\Big[\prod_{\ell\in B} q_{\ell}\Big] \Big[ \prod_{r\in B} (-1) p_{r} \Big]\Big)
\end{split}
\end{equation}
which is zero if the set of zero modes is strictly larger than $B$, true if $k>r$. This concludes the proof of the lemma.
\end{proof}
Let us now consider the special case in which the Hamiltonian $h$ is the lattice Laplacian associated with the $\pi$-flux phase, with fluxes $a,b$ on the non-contractible cycles. As discussed in Section \ref{sec:piflux}, zero modes are present if and only if $(a,b) = (1,1)$. The zero modes take place at $k = (\pi/2, \pi/2), (3\pi/2, \pi/2)$, and each has degeneracy $2$ (the $2\times 2$ Bloch Hamiltonian (\ref{eq:blochH}) vanishes for this choice of momentum). Thus, Lemma \ref{lem:zero} applies, with $r< 4$. In particular, the operators $ \hat J_1(-p,s)$, $\hat J_1(-p,s) \hat J_1(p)$ are of the form (\ref{eq:lemO}) with $r=1$ and $r=2$, respectively.

The plan will be to show that, for these operators,
\begin{equation}\label{eq:Oapp0}
\left \langle  \mathcal{O} \right \rangle_{\varepsilon,a,b} = \left \langle  \mathcal{O} \right \rangle_{\infty} + \frak{e}^{\beta,L}_{\varepsilon, a,b}
\end{equation}
with $\left \langle  \mathcal{O} \right \rangle_{\infty}$ independent of $\varepsilon, a, b$ and $\frak{e}^{\beta,L}_{\varepsilon, a,b} = o_{L}(1)$ as $\beta \to \infty$. Therefore, from (\ref{eq:sum}):
\begin{equation}\label{eq:Oapp}
\langle \mathcal{O} \rangle^{\pi}_{\beta, L} = \left \langle  \mathcal{O} \right \rangle_{\infty} + \frak{e}^{\beta,L}\;,
\end{equation}
with
\begin{equation}
\begin{split}
\Big| \frak{e}^{\beta,L}\Big| &= \Big| \sum_{(\varepsilon, a, b) \neq (-1,1,1)} \frak{e}^{\beta,L}_{\varepsilon, a,b} \frac{\Tr_{\mathcal{F}_{L}} e^{-\beta H(\pi; a,b)} (\varepsilon)^{N}}{\sum_{(a,b) \in \mathbb{Z}_{2}^{2}}\Tr_{\mathcal{F}_{L}} e^{-\beta H(\pi; a,b)} (1 + (-1)^{N})}\Big| \\
 &\leq \sum_{(\varepsilon, a, b) \neq (-1,1,1)} |\frak{e}^{\beta,L}_{\varepsilon, a,b}| = o_{L}(1)\;,
 \end{split}
\end{equation}
where we used the positivity of the odd partition functions and the vanishing of the partition function associated with $(\varepsilon, a, b) = (-1,1,1)$.

Let us now show (\ref{eq:Oapp}). To this end, we start by considering:
\begin{equation}\label{eq:intchi}
\int_0^{\beta} ds\, \left \langle  \hat J_1((0,-p_{2}),s); \hat J_1((0,p_{2})) \right \rangle_{-,a,b} \equiv \int_0^{\beta} ds\, \left \langle  \hat J_1((0,-p_{2}),s) \hat J_1((0,p_{2})) \right \rangle_{-,a,b}
\end{equation}
where the identity follows from the fact that the average of $J_1((0,-p_{2}),s)$ is zero for $p_{2} \neq 0$. In order to evaluate this quantity, it is convenient to extend periodically the argument of the integral to all $s\in \mathbb{R}$. Setting $a^{\sharp}_{x}(s) = e^{s H} a^{\sharp}_{x} e^{-s H}$, let us introduce the fermionic time ordering as, for $s_{i} \in[0,\beta)$:
\begin{equation}
{\bf T} a^{\sharp_{1}}_{x_{1}}(s_{1}) \cdots a^{\sharp_{r}}_{x_{r}}(s_{r}) = \text{sgn}(\pi) a^{\sharp_{\pi(1)}}_{x_{\pi(1)}}(s_{\pi(1)}) \cdots a^{\sharp_{\pi(r)}}_{x_{\pi(r)}}(s_{\pi(r)})\;, 
\end{equation}
where $\pi$ is the permutation such that $s_{\pi(1)} \geq s_{\pi(2)} \geq \cdots \geq s_{\pi(n)}$. Equal times ambiguities are solved via normal ordering (they are not important, since they involve a zero measure set of times, and the fermionic operators are bounded). The action of ${\bf T}$ is extended to all the fermionic algebra, by linearity. Using this definition, we have:
\begin{equation}\label{eq:wick}
(\ref{eq:intchi}) = \int_{-\beta/2}^{\beta/2} ds\, \left \langle {\bf T}  \hat J_1((0,-p_{2}),s); \hat J_1((0,p_{2})) \right \rangle_{-,a,b}
\end{equation}
where the argument of the integral is a periodic function over $\mathbb{R}$, with period $\beta$; in $[0,\beta]$, it agrees with the argument of the original integral (\ref{eq:intchi}). 

The factor $(-1)^{N}$ in the definition of $\langle \cdot \rangle_{-,a,b}$ can be taken into account introducing an imaginary chemical potential $\mu = i\pi / \beta$ in the definition of the Gibbs state. The state is quasi-free, and all correlations can be computed using the fermionic Wick's rule. To do this, let us represent the current operator in momentum space, using Eqs. (\ref{eq:afou}), with $k\in B_{L}(a,b)$. Observe that, in general, the current operator depends on the gauge field configuration; here, we are considering it evaluated on the periodic configuration that represents the $\pi$-flux phase with fluxes $(a,b)$ on the non-contractible cycles. From (\ref{eq:JK}), we have:
\begin{equation}
\hat J_{\mu}(p) = \frac{1}{|\Gamma_{L}^{\text{red}}|}\sum_{k \in B_{L}(a,b)} \hat a^{+}_{k-p,i} j_{\mu;ij}(k-p,k)\hat a^{-}_{k, j} 
\end{equation}
Let $\mu = 1$. Recalling that $(r_{i})_{1} = 0$ for $i=A,B$, we have, for $p = (0, p_{2})$:
\begin{equation}
j_{1}(k-p,k) = itq \eta_{1}(p) \begin{pmatrix} e^{-ip_2}(-e^{-i(k_{1}-p_1)} + e^{ik_{1}})  & 0 \\ 0 & e^{-i(k_{1}-p_1)} - e^{ik_{1}} \end{pmatrix}\;.
\end{equation}
The expression (\ref{eq:wick}) can be computed using the expression for the Euclidean two-point function, for $t\neq s$:
\begin{equation}\label{eq:2pt}
\langle {\bf T} \gamma_{t}(\hat a^{-}_{k,i}) \gamma_{s}(\hat a^{+}_{q,j}) \rangle_{-,a,b} = \delta_{k,q} |\Gamma^{\text{red}}_{L}| \frac{1}{\beta} \lim_{M\to \infty} \sum_{\substack{\omega \in \mathbb{M}_{\beta}^{\text{F}} \\ |\omega| \leq M}} e^{-i\omega(t-s)} \Big(\frac{1}{-i\omega + h(k) - i \pi/\beta}\Big)_{ij}
\end{equation}
where $\mathbb{M}_{\beta}^{\text{F}}$ is the set of fermionic Matsubara frequencies, $\mathbb{M}_{\beta}^{\text{F}} = \frac{2\pi}{\beta} (\mathbb{Z} + 1/2)$. Thus, applying the fermionic Wick's rule, and using that $\mathbb{M}_{\beta}^{\text{F}} + \pi/\beta = \mathbb{M}_{\beta}^{\text{B}}$, with $\mathbb{M}_{\beta}^{\text{B}}$ the set of bosonic Matsubara frequencies, we get:
\begin{equation}\label{eq:matsu}
\begin{split}
&(\ref{eq:wick}) \\
&= -\frac{1}{\beta}  \sum_{k\in B_{L}(a,b)} \sum_{\omega \in \mathbbm{M}^{\text{B}}_{\beta}} \Tr\Big(\frac{1}{-i \omega + h(k)} j_{1}(k,k-p) \frac{1}{-i \omega + h(k-p)} j_{1}(k-p,k) \Big) \\
&= -  2 |\Gamma^{\text{red}}_{L}|\int_{\mathbb{B}} \frac{d k}{(2\pi)^{2}} \int_{\mathbb{R}} \frac{d\omega}{(2\pi)}  \Tr\Big(\frac{1}{-i \omega + h(k)} j_{1}(k,k-p) \frac{1}{-i \omega + h(k-p)} j_{1}(k-p,k) \Big) + o(1)\;,
\end{split}
\end{equation}
where $\mathbb{B} = [0,2\pi] \times [0, \pi]$, and the $o(1)$ error terms are due to the approximation of the sums with integrals, which are vanishing as $\beta\to \infty$ and $L\to \infty$, irrespectively of the order.

The main term in (\ref{eq:matsu}) is independent of $a,b$; the choice of $a,b$ only affects the subleading error terms. The same computation can be reproduced for all $\langle \cdot \rangle_{\varepsilon, a, b}$, such that no zero modes are present; for $\varepsilon = +$, the two-point function is given by (\ref{eq:2pt}), omitting the term $-i\pi/\beta$ at the denominator. All these choices give rise to the same $L,\beta \to \infty$ limit. This proves (\ref{eq:Oapp0}) and hence (\ref{eq:Oapp}). 

Observe that, for $(\varepsilon, a, b) \neq (-1,1,1)$, the momentum-space two-point function is never singular at finite $\beta, L$. Also, the integral in the argument of the right-hand side of (\ref{eq:matsu}) is continuous in $p$, thanks to the almost-everywhere continuity and the absolute integrability of the integrand. Similarly, for these choices of $(\varepsilon, a, b)$, the explicit expression of the two-point function shows that the limits $p\to (0,0)$ and $L,\beta \to \infty$ of (\ref{eq:wick}) commute. This allows to check all assumptions we made to obtain the rewriting (\ref{eq:suscfin}). Thus, for $p = (0,p_{2})$, we get, using that $ |\Gamma^{\text{red}}_{L}| = L^{2}/2$, recall (\ref{eq:suscfin}):
\begin{equation}
\begin{split}
\chi(p_{2}) &= -\frac{1}{p^{2}} \int_{\mathbb{B}} \frac{d k}{(2\pi)^{2}} \int_{\mathbb{R}} \frac{d\omega}{(2\pi)} \Big[ \Tr\Big(\frac{1}{-i \omega + h(k)} j_{1}(k,k-p) \frac{1}{-i \omega + h(k-p)} j_{1}(k-p,k) \Big)\\
&\qquad - \Tr\Big(\frac{1}{-i \omega + h(k)} j_{1}(k,k) \frac{1}{-i \omega + h(k)} j_{1}(k,k) \Big)\Big] \\
&\equiv -\frac{1}{p_{2}^{2}} (A(p_{2}) - A(0))\;.
\end{split}
\end{equation}
In order to evaluate this quantity, we proceed similarly to what has been done for the conductivity of graphene, see {\it e.g.} \cite{GMPcond}. We preliminarily observe that the function:
\begin{equation}
p \mapsto \int_{\mathbb{B}} \frac{dk}{(2 \pi)^2} \Tr\Big(\frac{1}{-i \omega + h(k)} j_{1}(k,k-p) \frac{1}{-i \omega + h(k-p)} j_{1}(k-p,k) \Big)
\end{equation}
is even (use the change of variables $k+p\to k$, and the cyclicity of the trace). Therefore, if the function $A(p_{2})$ was differentiable, the limit $\lim_{p_{2}\to 0} p_{2} \chi(p_{2})$ would be zero. As for graphene, differentiability at zero momentum does not hold; nevertheless, the parity of $A(p_{2})$ can be used to simply the computation of the susceptibility. For $\delta > |p_{2}|$, we write:
\begin{equation}
\begin{split}
&A(p_{2}) \\
& =  \int_{\mathbb{B}} \int_{\mathbb{R}} \frac{d k}{(2\pi)^{2}} \frac{d\omega}{(2\pi)}  \Big[ \Tr\Big(\frac{1}{-i \omega + h(k)} j_{1}(k,k-p) \frac{1}{-i \omega + h(k-p)} j_{1}(k-p,k) \Big) \chi_{\delta}(k) +R_{1}(p) \\
& =: A_{\delta}(p_{2}) + R_{1}(p_{2})\;,
\end{split}
\end{equation}
where $\chi_{\delta}(k) = \mathbbm{1}(|k_{1} - \pi/2| \leq \delta) \mathbbm{1}(|k_{2} - \pi/2| \leq \delta) +  \mathbbm{1}(|k_{1} - 3\pi/2| \leq \delta) \mathbbm{1}(|k_{2} - \pi/2| \leq \delta)$ and $R_{1}(p_{2})$ is an error term, which takes into account the integration region associated with $k$ away from the singularity points $k_{F}^{-} = (\pi/2,\pi/2), k_{F}^{+} = (3\pi/2, \pi/2)$. It is easy to see that the function $R_{1}(p_{2})$ is differentiable in $p_{2}$ at $p_{2}=0$. Next, in $A_{\delta}(p_{2})$, we compare the integrand with its relativistic approximation. To do this, we use that, for $k = k_{F}^{\alpha} + k'$, $\alpha = \pm$ and $|k'| \leq \delta$:
\begin{equation}\label{eq:jh}
\begin{split}
j_{1}(k-p,k) &= 2tq \begin{pmatrix} \alpha e^{-ip_{2}} & 0 \\ 0 & -\alpha \end{pmatrix} + \frak{e}_{j}(k-p,k) \\
j_{1}(k,k-p) &= 2tq \begin{pmatrix} \alpha e^{ip_{2}} & 0 \\ 0 & -\alpha \end{pmatrix} + \frak{e}_{j}(k,k-p) \\
h(k_{F}^{\alpha} + k') &= -2t \begin{pmatrix} -\alpha k'_{1} & -ik'_{2} \\ ik'_{2} & \alpha k'_{1} \end{pmatrix} + \frak{e}_{h}(k)\;,
\end{split}
\end{equation}
where, for $k$ such that $|k - k^{\alpha}_{F}| \leq 2\delta$:
\begin{equation}\label{eq:extra}
\| \frak{e}_{j}(k-p,k) \| \leq C|k - k^{\alpha}_{F}|\;,\qquad \| \frak{e}_{j}(k,k-p) \| \leq C|k - k^{\alpha}_{F}|\;,\qquad \| \frak{e}_{h}(k) \|\leq C|k - k^{\alpha}_{F}|^{2}\;.
\end{equation}
In $A_{\delta}(p_{2})$, we replace $j_1$, $h$ with their leading contributions for $k$ close to either $k_{F}^{+}$ or $k_{F}^{-}$, as given by (\ref{eq:jh}). Correspondingly, we write:
\begin{equation}
A_{\delta}(p_{2})  = A^{\text{rel}}_{\delta}(p_{2}) + R_{2}(p_{2})\;;
\end{equation}
the function $A^{\text{rel}}_{\delta}(p_{2})$ takes into account the relativistic approximation of the integrand, while $R_{2}(p_{2})$ is an error term. The relativistic propagator entering in $A^{\text{rel}}_{\delta}(p_{2})$ is, for $k = k' + k_{F}^{\alpha}$, $|k'| \leq \delta$ and recalling that $v = 2t$:
\begin{equation}
\Big(\frac{1}{-i \omega + h(k)}\Big)^{\text{rel}}  = \frac{1}{\omega^{2} + v^{2} |k'|^{2}} \begin{pmatrix} i\omega + v\alpha k'_{1} & iv k'_{2} \\ -i v k'_{2} & i\omega - v\alpha k'_{1} \end{pmatrix}\;.
\end{equation}
Concerning the error terms, due to the extra factors $|k-k^{\alpha}_{F}|$ in the estimates (\ref{eq:extra}), one can check that $R_{2}(p_{2})$ is differentiable in $p$ at $p=0$. All together:
\begin{equation}
A(p_{2}) = A^{\text{rel}}_{\delta}(p_{2}) + \widetilde R(p_{2})\;,\qquad \widetilde R(p_{2}) = R_{1}(p_{2}) + R_{2}(p_{2})\;;
\end{equation}
the function $A^{\text{rel}}_{\delta}(p_{2})$ is even in $p_{2}$, and hence $\widetilde R(p_{2})$ is also even in $p_{2}$. In particular, by differentiability at zero,
\begin{equation}
\frac{1}{p_{2}} (\widetilde R(p) - \widetilde R(0)) \to 0\qquad \text{as $p_{2}\to 0$,}
\end{equation}
which implies:
\begin{equation}\label{eq:chip}
\chi(p_{2}) = -\frac{1}{p_{2}^{2}} (A^{\text{rel}}_{\delta}(p_{2})  - A^{\text{rel}}_{\delta}(0) ) + o\Big( \frac{1}{p_{2}} \Big)\;.
\end{equation}
Thus, we are left with evaluating the first term in the right-hand side of (\ref{eq:chip}), which is fully determined by the relativistic approximation at low energy. We will reduce this term to an integral that can be explicitly computed, plus error terms of the form $(1/p_{2}^{2}) (R_{j}(p_{2}) - R_{j}(0))$, for suitable $R_{j}(p_{2})$ that are even and differentiable; hence, arguing as before, they give a subleading contribution to $\chi(p_{2})$. We have, dropping all primes from now on:
\begin{equation}\label{eq:Ap2}
\begin{split}
A^{\text{rel}}_{\delta}(p_{2}) &=16 t^{2} q^{2} \int_{|k_{j}| \leq \delta} \frac{dk}{(2\pi)^{2}} \int \frac{d\omega}{(2\pi)}\, \frac{1}{(\omega^{2} + v^{2} |k|^{2}) (\omega^{2} + v^{2} |k - p|^{2})} \\
&\qquad \cdot \Big( -\omega^{2} + v^{2} k_{1}^{2} -  v^{2} k_{2}(k_{2} - p_{2}) \Big) + R_{3}(p_{2})\;,
\end{split}
\end{equation}
where $R_{3}(p_{2})$ takes into account the replacement of $e^{\pm ip_{2}}$ with $1$ in (\ref{eq:jh}), and it bounded as:
\begin{equation}
\frac{1}{p^{2}_{2}}\Big|R_{3}(p_{2}) - R_{3}(0)\Big|\leq C\left|\log p_{2}\right|\;. 
\end{equation}
Consider the main term in (\ref{eq:Ap2}). We rewrite it as:
\begin{equation}\label{eq:interm}
\begin{split}
&16 t^{2} q^{2} \int_{|k_{j}| \leq \delta} \frac{dk}{(2\pi)^{2}} \int \frac{d\omega}{(2\pi)}\, \frac{1}{(\omega^{2} + v^{2} |k|^{2}) (\omega^{2} + v^{2} |k - p|^{2})} \\
&\qquad \cdot \Big( -\omega^{2} + v^{2} k_{1}^{2} -  v^{2} k_{2}(k_{2} - p_{2}) \Big) \\
&= 4 q^{2} \int_{|k_{j}| \leq v\delta} \frac{dk}{(2\pi)^{2}} \int \frac{d\omega}{(2\pi)}\, \frac{1}{(\omega^{2} + |k|^{2}) (\omega^{2} + |k - v p|^{2})} \\
&\qquad \cdot \Big( -\omega^{2} + k_{1}^{2} -  k_{2}(k_{2} - v p_{2}) \Big) \\
&= 4 q^{2} \int_{|k_{j}| \leq v\delta} \frac{dk}{(2\pi)^{2}} \int_{|\omega| \leq v\delta} \frac{d\omega}{(2\pi)}\, \frac{1}{(\omega^{2} + |k|^{2}) (\omega^{2} + |k - v p|^{2})} \\
&\qquad \cdot \Big( -\omega^{2} + k_{1}^{2} -  k_{2}(k_{2} - v p_{2}) \Big) + R_{4}(p_{2})\;,
\end{split}
\end{equation}
where $R_{4}(p_{2})$ takes into account the contribution of $|\omega| > v \delta$, and it is smooth and even in $p_{2}$. Concerning the main term, observing that the integration domain is symmetric under exchange of $\omega$ and $k_{1}$, we have:
\begin{equation}
\begin{split}
&4 q^{2} \int_{|k_{j}| \leq v\delta} \frac{dk}{(2\pi)^{2}} \int_{|\omega| \leq v\delta} \frac{d\omega}{(2\pi)}\, \frac{-  k_{2}(k_{2} - v p_{2})}{(\omega^{2} + |k|^{2}) (\omega^{2} + |k - v p|^{2})}\\
&\quad = 4 q^{2} \int_{|\omega| \leq v\delta} \frac{d\omega}{2\pi} \int_{|k_{1}| \leq v\delta} \frac{d k_{1}}{2\pi} \int_{\mathbb{R}} \frac{d k_{2}}{2\pi}\, \frac{-  k_{2}(k_{2} - v p_{2})}{(\omega^{2} + |k|^{2}) (\omega^{2} + |k - v p|^{2})} + R_{5}(p_{2})\;,
\end{split}
\end{equation}
where $R_{5}(p_{2})$ takes into account the contribution of $|k_{2}| > v\delta$, which we added and subtracted; this error term is smooth and even in $p_{2}$. Let us come back to (\ref{eq:chip}). We have:
\begin{equation}\label{eq:fin0}
\begin{split}
\chi(p_{2}) &= \frac{4 q^{2}}{p_{2}^{2}} \int_{|\omega| \leq v\delta} \frac{d\omega}{2\pi} \int_{|k_{1}| \leq v\delta} \frac{d k_{1}}{2\pi} \int_{\mathbb{R}} \frac{d k_{2}}{2\pi}\, \Big( \frac{k_{2}(k_{2} - v p_{2})}{(\omega^{2} + |k|^{2}) (\omega^{2} + |k - v p|^{2})} - \frac{k_{2}^{2}}{(\omega^{2} + |k|^{2})^{2}} \Big) \\
&\quad - \frac{1}{p_{2}^{2}} \sum_{j=1}^{5} (R_{j}(p_{2}) - R_{j}(0))\;. 
\end{split}
\end{equation}
The integral over $k_{2}$ can be evaluated explicitly using residues; it has been done, for instance, in \cite[A.9]{GMPcond}. We have, replacing $p_{0}$ in \cite[A.9]{GMPcond} with $-vp_{2}$:
\begin{equation}
\int_{\mathbb{R}} \frac{d k_{2}}{2\pi}\, \Big( \frac{k_{2}(k_{2} - v p_{2})}{(\omega^{2} + |k|^{2}) (\omega^{2} + |k - v p|^{2})} - \frac{k_{2}^{2}}{(\omega^{2} + |k|^{2})^{2}} \Big) = -\frac{v^{2} p_{2}^{2}}{4(\omega^{2} + k_{1}^{2})^{\frac{1}{2}} (v^{2}p_{2}^{2} + 4 (\omega^{2} + k_{1}^{2}))}\;.
\end{equation}
Thus, we can rewrite the main term in (\ref{eq:fin0}) as:
\begin{equation}\label{eq:545}
\begin{split}
&\frac{4 q^{2}}{p_{2}^{2}} \int_{|\omega| \leq v\delta} \frac{d\omega}{2\pi} \int_{|k_{1}| \leq v\delta} \frac{d k_{1}}{2\pi} \int_{\mathbb{R}} \frac{d k_{2}}{2\pi}\, \Big( \frac{k_{2}(k_{2} - v p_{2})}{(\omega^{2} + |k|^{2}) (\omega^{2} + |k - v p|^{2})} - \frac{k_{2}^{2}}{(\omega^{2} + |k|^{2})^{2}} \Big) \\
&\quad = -\frac{4 q^{2}}{p_{2}^{2}}\int_{|\omega| \leq v\delta} \frac{d\omega}{2\pi} \int_{|k_{1}| \leq v\delta} \frac{d k_{1}}{2\pi}\,\frac{v^{2} p_{2}^{2}}{4(\omega^{2} + k_{1}^{2})^{\frac{1}{2}} (v^{2}p_{2}^{2} + 4 (\omega^{2} + k_{1}^{2}))} \\
&\quad = -\frac{4 v q^{2}}{|p_{2}|}\int_{|\omega| \leq \delta / |p_{2}|} \frac{d\omega}{2\pi} \int_{|k_{1}| \leq \delta / |p_{2}|} \frac{d k_{1}}{2\pi}\,\frac{1}{4(\omega^{2} + k_{1}^{2})^{\frac{1}{2}} (1 + 4 (\omega^{2} + k_{1}^{2}))} \\
&\quad = -\frac{4 v q^{2}}{|p_{2}|} \int_{\omega^{2} + k_{1}^{2} \leq \frac{\delta^{2}}{p_{2}^{2}}} \frac{d\omega}{(2\pi)} \frac{d k_{1}}{(2\pi)}\,\frac{1}{4(\omega^{2} + k_{1}^{2})^{\frac{1}{2}} (1 + 4 (\omega^{2} + k_{1}^{2}))} + E(p_{2})\;,
\end{split}
\end{equation}
where $E(p_{2})$ takes into account the change of integration domain, and it is $o(1 / |p_{2}|)$. The main term in (\ref{eq:545}) can be evaluated explicitly, using that:
\begin{equation}\label{eq:546}
\begin{split}
&\int_{\omega^{2} + k_{1}^{2} \leq \frac{\delta^{2}}{p_{2}^{2}}} \frac{d\omega}{(2\pi)} \frac{d k_{1}}{(2\pi)}\,\frac{1}{4(\omega^{2} + k_{1}^{2})^{\frac{1}{2}} (1 + 4 (\omega^{2} + k_{1}^{2}))} \\
&\quad = \frac{1}{2\pi} \int_{0}^{\delta / |p_{2}|} x dx\, \frac{1}{4x (1 + 4 x^{2})} \\
&\quad = \frac{1}{16 \pi} \int_{0}^{2\delta / |p_{2}|} dy\, \frac{1}{1 + y^{2}} = \frac{1}{16 \pi} \arctan(2\delta / |p_{2}|)\;.
\end{split}
\end{equation}
Thus, putting together (\ref{eq:fin0}), (\ref{eq:545}), (\ref{eq:546}), we have:
\begin{equation}
\chi(p_{2}) = -\frac{v q^{2}}{8 |p_{2}|} - \frac{1}{p_{2}^{2}} \sum_{j=1}^{5} (R_{j}(p_{2}) - R_{j}(0)) + \widetilde{E}(p_{2})\;,
\end{equation}
where $\widetilde{E}(p_{2})$ takes into account the replacement of the arctan with its limit, and it is $o(1 / |p_{2}|)$. Thus, since all $R_{j}(p_{2})$ are even and differentiable, we have:
\begin{equation}
\chi(p_{2}) = -\frac{v q^{2}}{8 |p_{2}|} + o\Big( \frac{1}{|p_{2}|} \Big)\;.
\end{equation}
This concludes the proof of Proposition \ref{prp:susc}.\qed
\appendix

\section{Proof of (\ref{eq:bden})}\label{app:diff}
We will treat the cases $L/2=2n$ and $L/2= 2n+1$ separately, as we expect from Lieb's theorem that the minimizers are different in the two cases (recall the Remark \ref{rem:bdry}). Let:
\begin{equation}
e_-(k|a,b) := e_{-}\Big( k_1 - \frac{\pi}{2L}(a-1), k_2 - \frac{\pi}{2L}(b-1)  \Big) \frac{1}{2t}
\end{equation}
The quasi-momentum $k$ belongs to the Brillouin zone $B_{L}(1,1)$, Eq. (\ref{eq:BL11}); that is, $k = (2\pi/L) (n_{1}, n_{2})$, with $0\leq n_{1} \leq L-1$, $0\leq n_{2} \leq L/2 - 1$. Using that $e_{-}(k_{1}, k_{2}) = e_{-}(k_{1} + \pi, k_{2})$, we can assume that $0\leq n_{1} \leq L/2-1$, and we count twice the energetic contribution of every quasi-momentum in the evaluation of the total energy. Suppose that $L/2 =2n$. The parametrization of the lower energy band as function of $n_{1}, n_{2}$ is, for $0\leq n_{i} \leq 2n-1$:
\begin{equation}
e_-(k|a,b) = - \sqrt{ 1 + \frac{1}{2} \cos\bigg( \frac{\pi n_{1}}{n} - \frac{\pi}{4n}(a-1)\bigg) + \frac{1}{2} \cos\bigg( \frac{\pi n_{2}}{n} - \frac{\pi}{4n}(b-1) \bigg) }.
\end{equation}
For these values of $(n_{1}, n_{2})$, the energy band is vanishing if and only if $(n_{1}, n_{2}) = (n,n)$. To begin, we separate the values of $(n_{1}, n_{2})$ close to $(n,n)$ from the values far from this point. That is, for $0<\alpha < 1$ to be chosen:
\begin{equation}\label{eq:split}
\sum_{n_{1}, n_{2}: 0\leq n_{i} \leq 2n-1 } e_-(k|a,b) = \sum_{n_{1}, n_{2}: |(n_{1}, n_{2}) - (n,n)| \leq n^{\alpha} } e_-(k|a,b) + \sum_{n_{1}, n_{2}: |(n_{1}, n_{2}) - (n,n)| > n^{\alpha} } e_-(k|a,b)\;,
\end{equation}
where here and in the following $|(n_{1}, n_{2}) - (n,n)|$ is the Euclidean distance defined $\text{mod}\, 2n$. Also, in what follows we will always assume that $(n_{1}, n_{2})$ are summed over the range $0\leq n_{i} \leq 2n-1$. We will show that the second term in the right-hand side of (\ref{eq:split}) is essentially independent of $(a,b)$, up to subleading terms as $n\to \infty$, while the first is small. 

Consider the first term in the right-hand side of (\ref{eq:split}). Here we use that, for the considered values of $n_{1}, n_{2}$, it holds $|e_{-}(k|a,b)| \leq C n^{\alpha - 1}$, for an $\alpha$-dependent constant. In what follows, we will not keep track of this dependence, and we will use the symbol $C$ to denote general multiplicative constants appearing in the estimates, that might change from line to line. Therefore:
\begin{equation}\label{eq:err1}
\sum_{n_{1}, n_{2}:\, |(n_{1}, n_{2}) - (n,n)| \leq n^{\alpha} } |e_-(k|a,b)| \leq C \frac{n^{3\alpha}}{n}\;,
\end{equation}
which is $o(1)$ for $\alpha < 1/3$. Consider now the second term in (\ref{eq:split}). Let us expand $e_-(k|a,b)$ around $e_-(k|1,1)$, as follows:
\begin{equation}\label{exp1}
\begin{split}
      e_-(k|a,b)^2 &= e_-(k|1,1)^2 + \bigg[\frac{1}{2} \frac{\pi}{4n} \big((a-1) \sin(\pi n_1/n) +(b-1) \sin(\pi n_2/n) \big) \\ &\quad - \frac{1}{4} \bigg(\frac{\pi}{4n}\bigg)^2 \big( (a-1)^2 \cos(\pi n_1/n) +(b-1)^2 \cos(\pi n_2/n) \big)\bigg] + O\bigg(\frac{1}{n^3}\bigg)\;.
\end{split}
\end{equation}
Observe that, for $n$ large enough, $|e_{-}(k|a,b)|^{2} \geq C (|n - n_{1}|^{2} + |n - n_{2}|^{2}) / n^{2}$, for all $a,b$, while the argument of the square brackets in (\ref{exp1}) is bounded above as $C(|n_{1} - n| + |n_{2} - n|)/n^{2}$; thus, for the considered values of $n_{1}, n_{2}$ and for $\alpha > 0$, the first term in the right-hand side of (\ref{exp1}) is much larger than the other ones.  Taking the square root of \eqref{exp1}, we have:
\begin{equation}\label{eq:B5}
\begin{split}
      e_-(k|a,b)&= e_-(k|1,1)\bigg(1 +\frac{\pi}{4n} \frac{\frac{(a-1)}{2} \sin(\pi n_1/n) +\frac{(b-1)}{2} \sin(\pi n_2/n)}{1 + \frac{1}{2}\cos(\pi n_1/n)+\frac{1}{2}\cos(\pi n_2/n)} \\ &\qquad -  \bigg(\frac{\pi}{4n}\bigg)^2 \frac{(\frac{a-1}{2})^2 \cos(\pi n_1/n) +(\frac{b-1}{2})^2 \cos(\pi n_2/n)}{1 + \frac{1}{2}\cos(\pi n_1/n)+\frac{1}{2}\cos(\pi n_2/n)} + O\bigg(\frac{1}{e^{2}_-(k|1,1) n^3}\bigg) \bigg)^{\frac{1}{2}} \\
      &=  e_-(k, |1,1) \bigg(1 + \frac{1}{2}
    \frac{\pi}{4n} \frac{\frac{(a-1)}{2} \sin(\pi n_1/n) +\frac{(b-1)}{2} \sin(\pi n_2/n)}{1 + \frac{1}{2}\cos(\pi n_1/n)+\frac{1}{2}\cos(\pi n_2/n)}\\
      &\quad - \frac{1}{2} \bigg(\frac{\pi}{4n}\bigg)^2 \frac{(\frac{a-1}{2})^2 \cos(\pi n_1/n) +(\frac{b-1}{2})^2 \cos(\pi n_2/n)}{1 + \frac{1}{2}\cos(\pi n_1/n)+\frac{1}{2}\cos(\pi n_2/n)} \\
      &\quad  -\frac{1}{8} \bigg(\frac{\pi}{4n}\bigg)^2 \bigg(\frac{\frac{(a-1)}{2} \sin(\pi n_1/n) +\frac{(b-1)}{2} \sin(\pi n_2/n)}{1 + \frac{1}{2}\cos(\pi n_1/n)+\frac{1}{2}\cos(\pi n_2/n)} \bigg)^2\bigg) + r(k|a,b)\;,
\end{split}
\end{equation}
where the error term $r(k|a,b)$ collects the Taylor remainder, and it is bounded as, for $|(n_{1}, n_{2}) - (n,n)| > n^{\alpha}$ and $0<\alpha <1/3$:
\begin{equation}
\begin{split}
|r(k|a,b)| &\leq \frac{C}{|e_-(k|1,1)| n^{3}} + \frac{C(|n_{1} - n|^{3} + |n_{2} - n|^{3})}{n^{6} |e_-(k|1,1)|^{5}} \\
&\equiv r_{1}(k|a,b) + r_{2}(k|a,b)\;.
\end{split}
\end{equation}
Consider $r_{1}(k|a,b)$. Using the lower bound for $e_-(k|1,1)$, we have:
\begin{equation}\label{eq:err2}
\begin{split}
\sum_{n_{1}, n_{2}: |(n_{1}, n_{2}) - (n,n)| > n^{\alpha} } r_{1}(k|a,b) &\leq C\sum_{n_{1}, n_{2}: |(n_{1}, n_{2}) - (n,n)| > n^{\alpha} } \frac{1}{n^{2}(|n-n_{1}| + |n - n_{2}|)} \\
&\leq \frac{C \log n}{n}\;.
\end{split}
\end{equation}
Consider now $r_{2}(k|a,b)$. We have:
\begin{equation}
\begin{split}
\sum_{n_{1}, n_{2}: |(n_{1}, n_{2}) - (n,n)| > n^{\alpha} } r_{2}(k|a,b) &\leq \sum_{n_{1}, n_{2}: |(n_{1}, n_{2}) - (n,n)| > n^{\alpha} }\frac{C(|n_{1} - n|^{3} + |n_{2} - n|^{3})}{n (|n_{1} - n|^{5} + |n_{2} - n|^{5})} \\
&\leq \frac{C}{n} \sum_{n_{1}, n_{2}: |(n_{1}, n_{2}) - (n,n)| > n^{\alpha} }\frac{1}{|n_{1} - n|^{2} + |n_{2} - n|^{2}} \\
&\leq \frac{C \log n}{n}\;.
\end{split}
\end{equation}
All together,
\begin{equation}\label{eq:err3}
\sum_{n_{1}, n_{2}: |(n_{1}, n_{2}) - (n,n)| > n^{\alpha} } |r(k|a,b)| \leq \frac{C \log n}{n}\;.
\end{equation}
Observe that all the error terms we accumulated to far, namely (\ref{eq:err1}), (\ref{eq:err3}), are $o(1)$ as $n\to \infty$. Our next task will be to show that the sum over $n_{1}, n_{2}$ of the main three terms in (\ref{eq:B5}) is also $o(1)$. We rewrite:
\begin{equation}\label{eq:endiff}
\begin{split}
&e_-(k|a,b) -   e_-(k|1,1) \\
&\quad = \frac{1}{2}
    \frac{\pi}{4n} \frac{\frac{(a-1)}{2} \sin(\pi n_1/n) +\frac{(b-1)}{2} \sin(\pi n_2/n)}{\sqrt{1 + \frac{1}{2}\cos(\pi n_1/n)+\frac{1}{2}\cos(\pi n_2/n)}}\\
&\qquad - \frac{1}{2} \bigg(\frac{\pi}{4n}\bigg)^2 \frac{(\frac{a-1}{2})^2 \cos(\pi n_1/n) +(\frac{b-1}{2})^2 \cos(\pi n_2/n)}{\sqrt{1 + \frac{1}{2}\cos(\pi n_1/n)+\frac{1}{2}\cos(\pi n_2/n)}} \\
      &\qquad -\frac{1}{8} \bigg(\frac{\pi}{4n}\bigg)^2\frac{(\frac{(a-1)}{2} \sin(\pi n_1/n) +\frac{(b-1)}{2} \sin(\pi n_2/n))^2}{(1 + \frac{1}{2}\cos(\pi n_1/n)+\frac{1}{2}\cos(\pi n_2/n))^{3/2}} + r(k|a,b)\;.
\end{split}
\end{equation}
Let us now consider the sum over all $(n_{1}, n_{2})$ such that $|n - n_{i}| > n^{\alpha}$, for $i=1,2$, of the first three terms in the right-hand side of (\ref{eq:endiff}). We start from the first term. We have:
\begin{equation}
\begin{aligned}
     \frac{\pi}{8n}  & \sum_{\substack{n_1,n_2=0,\cdots,2n-1 \\|n - n_{i}| > n^{\alpha}}} \frac{\frac{(a-1)}{2} \sin(\pi n_1/n) +\frac{(b-1)}{2} \sin(\pi n_2/n)}{\sqrt{1 + \frac{1}{2}\cos(\pi n_1/n)+\frac{1}{2}\cos(\pi n_2/n)}}\\
&= \frac{\pi(a+b-2) }{16n} \sum_{\substack{n_1,n_2=0,\cdots,2n \\|n - n_{i}| > n^{\alpha}}} \frac{\sin(\pi n_1/n)}{\sqrt{1 + \frac{1}{2}\cos(\pi n_1/n)+\frac{1}{2}\cos(\pi n_2/n)}}\\
     &=0\;,
\end{aligned}
\end{equation}
where we used that the integer $n_{1}$ belongs to the summation range if and only if $n_{1}^{*} = 2n - n_{1}$ does so, and that the summand evaluated on $n_{1}$ is minus the summand evaluated in $n_{1}^{*}$ (observe that we could freely add $n_{1} = 2n$ to the sum, since the argument of the sum is vanishing).  Consider now the second term in (\ref{eq:endiff}). We have:
\begin{equation}\label{eq:diff}
\begin{split}
        - \frac{\pi^2}{32 n^2}  &\sum_{\substack{n_1,n_2=0,\cdots,2n-1 \\|n - n_{i}| > n^{\alpha}}} \frac{(\frac{a-1}{2})^2 \cos(\pi n_1/n) +(\frac{b-1}{2})^2 \cos(\pi n_2/n)}{\sqrt{1 + \frac{1}{2}\cos(\pi n_1/n)+\frac{1}{2}\cos(\pi n_2/n)}} \\
         &\quad =-\frac{\pi^2((\frac{a-1}{2})^2+(\frac{b-1}{2})^2)}{32 n^2}  \sum_{\substack{n_1,n_2=0,\cdots,2n-1 \\|n - n_{i}| > n^{\alpha}}} \frac{\cos(\pi n_1/n)}{\sqrt{1 + \frac{1}{2}\cos(\pi n_1/n)+\frac{1}{2}\cos(\pi n_2/n)}}  \\
        &\quad =- \frac{((\frac{a-1}{2})^2+(\frac{b-1}{2})^2)}{32} \int_{-\pi}^{\pi} \int_{-\pi}^{\pi} dx dy \frac{\cos(x)}{\sqrt{1+\frac{1}{2}\cos(x)+\frac{1}{2}\cos(y)}} + O\bigg(\frac{n^{\alpha}}{n}\bigg)\;.\\
\end{split}
\end{equation}
Finally, consider the third term. We get:
\begin{equation}\label{eq:terzo}
\begin{split}
         -\frac{\pi^2}{128n^2} &\sum_{\substack{n_1,n_2=0,\cdots,2n-1 \\|n - n_{i}| > n^{\alpha}}} \frac{(\frac{(a-1)}{2} \sin(\pi n_1/n) +\frac{(b-1)}{2} \sin(\pi n_2/n))^2}{(1 + \frac{1}{2}\cos(\pi n_1/n)+\frac{1}{2}\cos(\pi n_2/n))^{3/2}} \\&\quad =
         -\frac{\pi^2((\frac{a-1}{2})^2+(\frac{b-1}{2})^2)}{128 n^2} \sum_{\substack{n_1,n_2=-n,\cdots,n-1 \\|n - n_{i}| > n^{\alpha}}} \frac{\sin(\pi n_1/n)^2 }{(1 + \frac{1}{2}\cos(\pi n_1/n)+\frac{1}{2}\cos(\pi n_2/n))^{3/2}}  \\
         &\quad = -\frac{((\frac{a-1}{2})^2+(\frac{b-1}{2})^2)}{128} \int_{-\pi}^{\pi} \int_{-\pi}^{\pi} dx dy \frac{\sin(x)^2 }{(1 + \frac{1}{2}\cos(x)+\frac{1}{2}\cos(y))^{3/2}} +O\bigg(\frac{n^{\alpha}}{n}\bigg)\\
         &\quad = \frac{((\frac{a-1}{2})^2+(\frac{b-1}{2})^2))}{32} \int_{-\pi}^{\pi} \int_{-\pi}^{\pi} dx dy \frac{\cos(x)}{\sqrt{1+\frac{1}{2}\cos(x)+\frac{1}{2}\cos(y)}}  +O\bigg( \frac{n^{\alpha}}{n}\bigg)\;.
\end{split}
\end{equation}
Observe that the main term in (\ref{eq:terzo}) exactly cancels the main term in (\ref{eq:diff}). Thus, all together we proved (\ref{eq:bden}), for $L/2 = 2n$.

Next, consider the case $L/2 =2n+1$. Here, we compare $e_-(k|a,b)$ with $e_-(k|-1,-1)$, using that:
\begin{equation}
\begin{split}
\cos(2k_{1} - \frac{\pi}{L}(a-1)) &= \cos(2k_{1} + \frac{2\pi}{L} - \frac{\pi}{L}(a+1) ) \\
&= \cos( \frac{\pi}{2n+1} (2n_{1} + 1) - \frac{\pi}{2(2n + 1)}(a+1) )\;.
\end{split}
\end{equation}
The dispersion relation $e_-(k|-1,-1)$ vanishes for $(n_{1}, n_{2}) = (n,n)$. As before, we study differently the contributions $|(n_{1}, n_{2}) - (n,n)| \leq n^{\alpha}$, $|(n_{1}, n_{2}) - (n,n)| > n^{\alpha}$. The former are estimated as in (\ref{eq:err1}). Let us consider the latter. We Taylor expand:
\begin{equation}\label{exp2}
\begin{aligned}
      e_-(k|a,b)^2 &= e_-(k|-1,-1)^2 \\&\quad +4t^2\bigg[\frac{1}{2} \frac{\pi}{2(2n+1)} \bigg((a+1) \sin(\pi\frac{2 n_1+1}{2n+1}) +(b+1) \sin(\pi \frac{2 n_2+1}{2n+1}) \bigg) \\&\quad - \frac{1}{4} \bigg(\frac{\pi}{2(2n+1)}\bigg)^2 \bigg( (a+1)^2 \cos\big(\pi\frac{2 n_1+1}{2n+1}\big) +(b+1)^2 \cos\big(\pi \frac{2n_2+1}{2n+1}\big) \bigg)\bigg]\\
      &\quad  + O\bigg(\frac{1}{n^2}\bigg)\;.
\end{aligned} 
\end{equation}
The analysis of the error terms is performed as before, and we omit the details. We have, for $|(n_{1}, n_{2}) - (n,n)| > n^{\alpha}$:%
\begin{equation}
    \begin{split}
      e_-(k|a,b)&= e_-(k|-1,-1)\bigg(1 +\frac{\pi}{2(2n+1)} \frac{\frac{(a+1)}{2} \sin(\pi\frac{2 n_1+1}{2n+1}) +\frac{(b+1)}{2} \sin(\pi\frac{2 n_2+1}{2n+1})}{1 + \frac{1}{2}\cos(\pi\frac{2 n_1+1}{2n+1})+\frac{1}{2}\cos(\pi\frac{2 n_2+1}{2n+1})} \\&\quad -  \bigg(\frac{\pi}{2(2n+1)}\bigg)^2 \frac{(\frac{a+1}{2})^2 \cos(\pi\frac{2 n_1+1}{2n+1}) +(\frac{b+1}{2})^2 \cos(\pi\frac{2 n_2+1}{2n+1})}{1 + \frac{1}{2}\cos(\pi\frac{2 n_1+1}{2n+1})+\frac{1}{2}\cos(\pi\frac{2 n_2+1}{2n+1})}\\&\quad + o\bigg(\frac{1}{e^{2}_-(k|-1,-1) n^2}\bigg) \bigg)^{\frac{1}{2}} \\
      &=  e_-(k,l|-1,-1) \bigg(1 + \frac{1}{2}
    \frac{\pi}{2(2n+1)} \frac{\frac{(a+1)}{2} \sin(\pi\frac{2 n_1+1}{2n+1}) +\frac{(b+1)}{2} \sin(\pi\frac{2 n_2+1}{2n+1})}{1 + \frac{1}{2}\cos(\pi\frac{2 n_1+1}{2n+1})+\frac{1}{2}\cos(\pi\frac{2 n_2+1}{2n+1})}\\
      &\quad - \frac{1}{2} \bigg(\frac{\pi}{2(2n+1)}\bigg)^2 \frac{(\frac{a+1}{2})^2 \cos(\pi\frac{2 n_1+1}{2n+1}) +(\frac{b+1}{2})^2 \cos(\pi\frac{2 n_2+1}{2n+1})}{1 + \frac{1}{2}\cos(\pi\frac{2 n_1+1}{2n+1})+\frac{1}{2}\cos(\pi\frac{2 n_2+1}{2n+1})} \\
      &\quad -\frac{1}{8} \bigg(\frac{\pi}{2(2n+1)}\bigg)^2 \bigg(\frac{\frac{(a+1)}{2} \sin(\pi\frac{2 n_1+1}{2n+1}) +\frac{(b+1)}{2} \sin(\pi\frac{2 n_2+1}{2n+1})}{1 + \frac{1}{2}\cos(\pi\frac{2 n_1+1}{2n+1})+\frac{1}{2}\cos(\pi\frac{2 n_2+1}{2n+1}) } \bigg)^2\bigg)\\&
      \quad + \tilde r(k|a,b)\;,
\end{split} 
\end{equation}
where the error term $\tilde r(k|a,b)$ is bounded as in (\ref{eq:err3}). Therefore:
\begin{equation}\label{eq:b18}
\begin{split}
  e_-(k|a,b) -   e_-(k|-1,-1) &= \frac{1}{2}
    \frac{\pi}{2(2n+1)} \frac{\frac{(a+1)}{2} \sin(\pi\frac{2 n_1+1}{2n+1}) +\frac{(b+1)}{2} \sin(\pi\frac{2 n_2+1}{2n+1})}{\sqrt{1 + \frac{1}{2}\cos(\pi\frac{2 n_1+1}{2n+1})+\frac{1}{2}\cos(\pi\frac{2 n_2+1}{2n+1})}}\\
&\quad - \frac{1}{2} \bigg(\frac{\pi}{2(2n+1)}\bigg)^2 \frac{(\frac{a+1}{2})^2 \cos(\pi\frac{2 n_1+1}{2n+1}) +(\frac{b+1}{2})^2 \cos(\pi\frac{2 n_2+1}{2n+1})}{\sqrt{1 + \frac{1}{2}\cos(\pi\frac{2 n_1+1}{2n+1})+\frac{1}{2}\cos(\pi\frac{2 n_2+1}{2n+1})}} \\
      &\quad -\frac{1}{8} \bigg(\frac{\pi}{2(2n+1)}\bigg)^2\frac{(\frac{(a+1)}{2} \sin(\pi\frac{2 n_1+1}{2n+1}) +\frac{(b+1)}{2} \sin(\pi\frac{2 n_2+1}{2n+1}))^2}{(1 + \frac{1}{2}\cos(\pi\frac{2 n_1+1}{2n+1})+\frac{1}{2}\cos(\pi\frac{2 n_2+1}{2n+1}))^{3/2}}\\&\quad + \tilde r(k|a,b)\;.
\end{split}
\end{equation}
Let us evaluate the sum over all $(n_{1},n_{2})$ such that $|(n_{1}, n_{2}) - (n,n)| > n^{\alpha}$ of the first three terms in the right-hand side of (\ref{eq:b18}). Consider the first term. We have:
\begin{equation}
\begin{split}
     &\frac{\pi}{4(2n+1)}   \sum_{\substack{n_1,n_2=0,\cdots,2n \\ |n_{i} - n| > n^{\alpha}}} \frac{\frac{(a+1)}{2} \sin(\pi\frac{2 n_1+1}{2n+1}) +\frac{(b+1)}{2} \sin(\pi\frac{2 n_2+1}{2n+1})}{\sqrt{1 + \frac{1}{2}\cos(\pi\frac{2 n_1+1}{2n+1})+\frac{1}{2}\cos(\pi\frac{2 n_2+1}{2n+1})}}\\
&\qquad = \frac{\pi(a+b+2) }{8(2n+1)} \sum_{\substack{n_1,n_2=0,\cdots,2n \\ |n_{i} - n| > n^{\alpha}}} \frac{\sin(\pi\frac{2 n_1+1}{2n+1})}{\sqrt{1 + \frac{1}{2}\cos(\pi\frac{2 n_1+1}{2n+1})+\frac{1}{2}\cos(\pi\frac{2 n_2+1}{2n+1})}}\\
&\qquad =0\;,
\end{split}
\end{equation}
where we used that $n_{1}$ belongs to the summation range if and only if $n_{1}^{*} = 2n - n_{1}$ does so, and the argument of the sum takes opposite values at $n_{1}$ and $n_{1}^{*}$. Consider the second term. We have:
\begin{equation}\label{eq:aa}
    \begin{split}
        &- \frac{\pi^2}{8 (2n+1)^2}  \sum_{\substack{n_1,n_2=0,\cdots,2n \\ |n_{i} - n| > n^{\alpha}}} \frac{(\frac{a+1}{2})^2 \cos(\pi\frac{2 n_1+1}{2n+1}) +(\frac{b+1}{2})^2 \cos(\pi\frac{2 n_2+1}{2n+1})}{\sqrt{1 + \frac{1}{2}\cos(\pi\frac{2 n_1+1}{2n+1})+\frac{1}{2}\cos(\pi\frac{2 n_2+1}{2n+1})}} \\
         &\quad =-\frac{\pi^2((\frac{a+1}{2})^2+(\frac{b+1}{2})^2)}{ 8 (2n+1)^2}  \sum_{\substack{n_1,n_2=0,\cdots,2n \\|n_{i} - n| > n^{\alpha}}} \frac{\cos(\pi\frac{2 n_1+1}{2n+1})}{\sqrt{1 + \frac{1}{2}\cos(\pi\frac{2 n_1+1}{2n+1})+\frac{1}{2}\cos(\pi\frac{2 n_2+1}{2n+1})}}  \\
        &\quad =- \frac{((\frac{a+1}{2})^2+(\frac{b+1}{2})^2)}{8} \int_{-\pi}^{\pi} \int_{-\pi}^{\pi} dx dy \frac{\cos(x)}{\sqrt{1+\frac{1}{2}\cos(x)+\frac{1}{2}\cos(y)}} + O\bigg(\frac{n^{\alpha}}{n}\bigg)\;.
    \end{split}
\end{equation}
Finally, consider the third term. We have:
\begin{equation}\label{eq:aaa}
    \begin{split}
         &-\frac{\pi^2}{32(2n+1)^2} \sum_{\substack{n_1,n_2=0,\cdots,2n \\ |n_{i} - n| > n^{\alpha}}} \frac{(\frac{(a+1)}{2} \sin(\pi\frac{2 n_1+1}{2n+1}) +\frac{(b+1)}{2} \sin(\pi\frac{2 n_2+1}{2n+1}))^2}{(1 + \frac{1}{2}\cos(\pi\frac{2 n_1+1}{2n+1})+\frac{1}{2}\cos(\pi\frac{2 n_2+1}{2n+1}))^{3/2}} \\
         &\quad = -\frac{((\frac{a+1}{2})^2+(\frac{b+1}{2})^2)}{32} \int_{-\pi}^{\pi} \int_{-\pi}^{\pi} dx dy \frac{\sin(x)^2 }{(1 + \frac{1}{2}\cos(x)+\frac{1}{2}\cos(y))^{3/2}} +O\bigg(\frac{n^{\alpha}}{n}\bigg)\\
         &\quad = \frac{((\frac{a-1}{2})^2+(\frac{b-1}{2})^2)}{8} \int_{-\pi}^{\pi} \int_{-\pi}^{\pi} dx dy \frac{\cos(x)}{\sqrt{1+\frac{1}{2}\cos(x)+\frac{1}{2}\cos(y)}}  +O\bigg( \frac{n^{\alpha}}{n}\bigg)\;.
    \end{split}
\end{equation}
Again, the main terms in (\ref{eq:aa}), (\ref{eq:aaa}) cancel out. This concludes the proof of (\ref{eq:bden}) for $L/2 = 2n+1$.

\end{document}